\documentclass[
 reprint,
 amsmath,amssymb,
 aps,
 physrev,
 superscriptaddress,
]{revtex4-2}

\usepackage[utf8]{inputenc}
\usepackage{graphicx}
\usepackage{array}[=2016-10-06]
\usepackage{dcolumn}
\usepackage{bm}
\usepackage{multirow}
\usepackage{braket}
\usepackage{amsmath, amssymb, amsthm}
\usepackage{xcolor}
\usepackage{hyperref}
\usepackage[caption=false]{subfig}
\usepackage{orcidlink}

\hypersetup{
    colorlinks=true,
    allcolors=blue
}

\newtheorem{theorem}{Theorem}
\newtheorem{proposition}[theorem]{Proposition}
\newtheorem{lemma}[theorem]{Lemma}
\newtheorem{corollary}[theorem]{Corollary}
\theoremstyle{definition}
\newtheorem{definition}[theorem]{Definition}

\iffalse
\newcommand{\comments}[1]{\footnote{\textcolor{blue}{\textit{#1}}}}

\else
\newcommand{\comments}[1]{}

\fi

\usepackage{thmtools}
\newenvironment{restate}[2]{%
  \par\noindent\textbf{#1 \ref{#2}. }\itshape%
}{\par\medskip}

\begin{document}

\title{Fourier Symmetrization for Geometric Quantum Machine Learning}


\author{Letao Wang\,\orcidlink{0009-0004-2411-8916}}
\email{letao.wang@centralesupelec.fr}
\affiliation{Laboratory of Signals and Systems, CentraleSup\'elec, CNRS, Paris-Saclay University, Gif-sur-Yvette, 91190, France}
\author{Abdel Lisser\,\orcidlink{0000-0003-1318-6679}}
\affiliation{Laboratory of Signals and Systems, CentraleSup\'elec, CNRS, Paris-Saclay University, Gif-sur-Yvette, 91190, France}
\affiliation{F\'ed\'eration de Math\'ematiques de CentraleSup\'elec, CentraleSup\'elec, CNRS, Paris-Saclay University, Gif-sur-Yvette, 91190, France}
\author{Sreejith Sreekumar\,\orcidlink{0000-0003-3615-8961}}
\affiliation{Laboratory of Signals and Systems, CentraleSup\'elec, CNRS, Paris-Saclay University, Gif-sur-Yvette, 91190, France}
\author{Zeno Toffano\,\orcidlink{0000-0001-8594-3291}}
\affiliation{Laboratory of Signals and Systems, CentraleSup\'elec, CNRS, Paris-Saclay University, Gif-sur-Yvette, 91190, France}


\begin{abstract}
Geometric quantum machine learning incorporates symmetry into quantum models, but how symmetry shapes their expressivity and guides effective model design remains insufficiently understood. We address this question through the Fourier representation of quantum Fourier models (QFMs). Symmetry organizes the frequency spectrum into orbits and sums the Fourier coefficients within each orbit into a symmetrized coefficient. When QFM trainable layers form independent exact 2-designs, the variance of each symmetrized coefficient equals the sum of the coefficient variances in its orbit. For single-layer QFMs with \(\varepsilon\)-approximate 2-design trainable layers, we bound the deviation from this identity. The resulting bound for individual Fourier coefficients can be exponentially tighter than an existing bound. The hyperoctahedral group provides an example of orbit growth that can mitigate vanishing expressivity of the symmetrized coefficients. Symmetrization of QFMs over an elementary abelian 2-group also yields pure multivariate Chebyshev polynomial basis functions. We introduce randomized encoding, which implements invariant models without ancilla qubits or the additional circuit depth for quantum twirling. We evaluate the models as quantum physics-informed neural networks (QPINNs) on two-dimensional screened Poisson and stationary viscous Hamilton-Jacobi equations. Under hard boundary constraints, QPINNs using exact symmetrization and randomized encoding achieve the lowest mean errors in the two benchmarks, respectively.
\end{abstract}

\maketitle

\section{Introduction}

Symmetry is a structural prior in many scientific learning problems. In portfolio optimization, a simultaneous relabeling of assets should not change the underlying objective. In molecular modeling, energies and force fields must respect geometric symmetries of the input coordinates. Encoding such structure into the model can improve data efficiency, generalization, and trainability, and this idea has motivated geometric quantum machine learning (GQML) \cite{Larocca2022,Meyer2023,Michael2023,Cenk2024,Louis2024,Nguyen2024,West2024,Sauvage2024,nha2025,Das2025}. This paper studies the role of symmetry in data re-uploading models \cite{Adrian2020,Schuld2021Effect} from the perspective of quantum Fourier models (QFMs) \cite{Hela2025}. The underlying circuit is the familiar re-uploading architecture, while the QFM viewpoint emphasizes that the induced hypothesis function admits a finite Fourier series. This perspective connects geometric priors in variational quantum models with questions of trainability, spectral approximation, and scientific machine learning~\cite{Schuld2021Effect,Hela2025}.

The Fourier viewpoint provides a new perspective on what makes a symmetry prior useful. QFMs can access a large frequency spectrum, especially under exponential encoding schemes \cite{Shin_2023}. The statistics of Fourier coefficients can affect trainability. In particular, coefficient variances can decay exponentially with the number of qubits, which can make the corresponding frequency components difficult to train. This phenomenon is known as \textit{vanishing expressivity} \cite{Hela2025}. A recent perspective argues that the Fourier structure of QFMs should guide model design and identifies GQML as a promising direction \cite{Oh2026}. However, how imposing invariance changes the frequency spectrum and Fourier coefficients of a QFM remains unclear. We therefore study how invariance reorganizes the frequency spectrum and changes the coefficient statistics.

To address these questions, we enforce invariance by averaging the QFM hypothesis function over inputs transformed by the symmetry. This construction follows the general twirling principle, in which functions, operators, or channels are averaged over a group representation so that the resulting object respects the symmetry \cite{Meyer2023,Michael2023,Louis2024,Nguyen2024}. At the function level, this averaging enforces invariance. At the Fourier level, it organizes frequencies into orbits. Our first contribution is to show that invariant QFMs admit the \textit{Fourier orbit decomposition}, whether the invariant model is obtained by this averaging operator or by another method. In this decomposition, each frequency orbit defines an orbit basis function, whose coefficient in the invariant model is the sum of the original Fourier coefficients belonging to that orbit. We call this sum the \textit{symmetrized coefficient} of the orbit. For the averaging construction, if the original frequency spectrum contains only part of an orbit, averaging completes that orbit. The orbit basis function then contains all frequencies in the completed orbit, whereas its symmetrized coefficient sums only the Fourier coefficients present in the original spectrum. This Fourier level view lets us analyze how invariance changes expressivity, rather than treating invariance only as a formal constraint on the output function.

Building on the Fourier orbit decomposition, we analyze the variance of the symmetrized coefficients under exact and \(\varepsilon\)-approximate 2-design assumptions. We first consider an idealized setting in which each parameterized trainable layer independently forms an exact 2-design. Under this assumption, distinct Fourier coefficients have zero cross covariance, so the variance of each symmetrized coefficient equals the sum of the variances of the Fourier coefficients in its orbit. With the normalized orbit basis, the fact that these variances add provides a mechanism for mitigating vanishing expressivity. If an orbit contains many coefficients with comparable variance, the variance of the symmetrized coefficient can be substantially larger than that of any individual coefficient in the orbit. We then consider realistic ansatzes whose trainable layers may form only \(\varepsilon\)-approximate 2-designs. In this setting, cross covariance terms between distinct Fourier coefficients do not vanish, so the symmetrized variance no longer equals the sum of the individual coefficient variances. We bound this deviation in terms of the approximation error. In some practically relevant regimes, the resulting bounds can be exponentially tighter than those in previous work \cite{Hela2025}.

In particular, the Fourier representation of QFMs has been highlighted as a guide for designing differential-equation solvers \cite{Oh2026}. In this setting, symmetry averaging has a role beyond changing coefficient variance. It can also impose a useful inductive bias for solving partial differential equations (PDEs). Chebyshev polynomials form a powerful basis for function approximation and are especially useful for this purpose. A natural way to bring this structure into quantum models is to use the \textit{Chebyshev feature map} \cite{Kyriienko_2021}, but the direct construction introduces non-polynomial terms that weaken the intended Chebyshev inductive bias and have derivatives that diverge as the input approaches the domain boundaries. We show that averaging over sign flips of the encoded angles removes these terms and yields pure multivariate Chebyshev polynomial basis functions.

After establishing these benefits of invariance, we turn to the question of implementation. Existing GQML methods can construct quantum circuits that respect the symmetry in advance or impose symmetry during circuit execution. Null space and Choi operator methods construct compatible layers or channels through offline computation, whose cost can scale exponentially with the number of qubits \cite{Nguyen2024}. Gate symmetrization can similarly precompute a gateset compatible with the symmetry, but its computational cost scales with the size of the symmetry group \cite{Meyer2023}. Alternatively, twirling can implement the averaging procedure during circuit execution using quantum circuits for the corresponding group representation. Although twirling can be implemented at the circuit level in several ways, these implementations can require additional ancilla qubits or extra circuit depth \cite{Nguyen2024}. The resource costs of these conventional methods become especially significant for large symmetry groups such as the hyperoctahedral group \(B_m\) considered in \cite{Meyer2023}, which consists of coordinate permutations and independent sign flips. It captures symmetries that are practically relevant to multivariate learning problems, but its size \(|B_m|=2^m m!\) grows faster than exponentially with the number of variables, making conventional implementations difficult to scale.

We therefore propose \textit{randomized encoding}, which estimates the same averaged predictor by randomizing the input encoding. Instead of implementing the corresponding group representation on the quantum circuit, randomized encoding samples a transformation from the symmetry group and applies the original encoding with the transformed input while keeping the original QFM circuit architecture. This procedure is equivalent to randomly selecting and applying encoding gates (often parameterized Pauli rotation gates) according to transformations sampled from the symmetry group, and it provides an unbiased estimator of the averaged predictor. We mainly compare randomized encoding with the twirling implementation in \cite{Nguyen2024} because it also samples transformations from the symmetry group but implements the corresponding group representation on the quantum device, whereas randomized encoding removes this depth overhead when the symmetry acts on the input.

Finally, we evaluate symmetry averaging in quantum physics-informed neural network (QPINN) \cite{Kyriienko_2021,Berger2025,panichi2025} models for two PDEs in two dimensions: a screened Poisson equation and a stationary viscous Hamilton-Jacobi equation. We compare QFMs and models using the Chebyshev feature map with their \(B_2\) averaged versions, together with a fully connected PINN with a comparable number of trainable parameters. The averaged versions compute the average operator either by enumerating all group elements or by randomized encoding. We test hard constraints that enforce the boundary conditions exactly and soft constraints that include boundary condition violations in the loss function. Under hard constraints, randomized encoding gives lower final loss and error than the QFM without symmetry averaging, the model using the Chebyshev feature map, and the fully connected PINN for the screened Poisson equation. For the stationary viscous Hamilton-Jacobi equation, it achieves the lowest mean final loss and the lowest mean error among all tested models.

The paper is organized as follows. Section~\ref{sec:GQML} reviews the symmetry framework for data re-uploading quantum models. Section~\ref{sec:fourier_symmetrization} develops the Fourier orbit decomposition, analyzes the variance of the symmetrized coefficients, and presents the inductive bias for solving PDEs. Section~\ref{sec:randomized_encoding} introduces randomized encoding and analyzes its sampling cost and hardware requirements. Section~\ref{sec:experiment} reports the QPINN experiments before we conclude.

\section{Symmetrization Framework}
\label{sec:GQML}

Recent studies exploiting symmetries in quantum machine learning, including GQML have demonstrated promising results across various problem setups, leveraging distinct advantages in complexity, trainability, and generalization \cite{Larocca2022,Meyer2023,Michael2023,Cenk2024,Louis2024,Nguyen2024,Kazi2024,Sauvage2024,Das2025,Sreekumar2026}. Existing GQML frameworks impose symmetry in different circuit settings. Some approaches consider single-layer re-uploading models \cite{Larocca2022,Michael2023}, extend such models to multiple layers while assuming that the encoding layers remain identical across all layers \cite{Meyer2023,nha2025}, or focus on multi-layer quantum neural networks \cite{Michael2023,Louis2024,Nguyen2024}. In this work, we investigate the quantum Fourier model (QFM), a generalized re-uploading architecture that allows flexible encoding across layers. The Fourier structure of these models provides a basis for understanding their expressivity and guiding their design, as highlighted in a recent perspective~\cite{Oh2026}.

Consider an \(n\)-qubit QFM with \(L\) layers, where \(\bm{\theta}\) denotes the set of trainable parameters. The prepared quantum state and the corresponding re-uploading unitary are
\begin{equation}
\label{eq:reupload}
\begin{aligned}
\rho(\bm{x}, \bm{\theta})
&=
U(\bm{x}, \bm{\theta})|0\rangle\langle0|U^\dagger(\bm{x}, \bm{\theta}),\\
U(\bm{x}, \bm{\theta})
&=
W^{(L+1)}(\bm{\theta})\prod_{l=1}^{L}E^{(l)}(\bm{x}) W^{(l)}(\bm{\theta}).
\end{aligned}
\end{equation}
This unitary comprises alternating encoding layers \(E(\bm{x})\) and trainable layers \(W(\bm{\theta})\). Given a Hermitian observable \(\hat{O}\), we consider its Pauli decomposition \(\hat{O} = \sum_{i=1}^{M} h_i O^{(i)}\), where \(h_i \in \mathbb{R}\) are the Pauli coefficients and \(O^{(i)} \in \{I, X, Y, Z\}^{\otimes n}\) are \(n\)-qubit Pauli strings. Let the \(L_1\)-norm of the Pauli coefficients be defined as \(\|h\|_1=\sum_{i=1}^M|h_i|\). For a classical input \(\bm{x} = (x_1, \dots, x_m) \in \mathbb{R}^m\), the model induces a parameterized hypothesis function \(f(\bm{x}, \bm{\theta})\) by linearity:
\begin{equation}
\label{eq:re-uploading_function}
f(\bm{x}, \bm{\theta}) = \operatorname{Tr}(\hat{O} \rho(\bm{x}, \bm{\theta})) = \sum_{i=1}^{M} h_i \operatorname{Tr}(O^{(i)} \rho(\bm{x}, \bm{\theta})),
\end{equation}
where \(|f(\bm{x}, \bm{\theta})| \leq\|h\|_1\). For the \(i\)-th Pauli string \(O^{(i)}\), each circuit shot yields a random variable \(Z^{(i)}_{\bm{x}} \in \{-1, 1\}\) with an outcome distribution governed by Born's rule, yielding an expectation value \(\mathbb{E}[Z^{(i)}_{\bm{x}}] = \operatorname{Tr}(O^{(i)} \rho(\bm{x}, \bm{\theta}))\). Furthermore, this hypothesis function admits a Fourier series form \cite{Schuld2021Effect,Hela2025}:
\begin{equation}
\label{eq:QFM}
f(\bm{x}, \bm{\theta}) = \sum_{\bm{\omega}\in\Omega} c_{\bm{\omega}}(\bm{\theta})\, e^{-i\,\bm{x}\cdot\bm{\omega}}.
\end{equation}
Here \(\Omega\) is the finite frequency spectrum, \(\bm{\omega}\) is a frequency vector with the same dimension as the input \(\bm{x}\), and \(c_{\bm{\omega}}(\bm{\theta})\) denotes the Fourier coefficient associated with \(\bm{\omega}\). To formalize our framework, we first define functions invariant under the symmetry, a central objective in GQML.

\begin{definition}[Invariant function]
\label{def:invariant_function}
Consider a finite group \(\mathcal{S}\) with a group representation \(V:\mathcal{S} \to \operatorname{Aut}(\mathbb{R}^m)\), \(s \mapsto V_s\), where \(V_s\) denotes the element of the representation \(V\) associated with \(s\), and \(\operatorname{Aut}(\mathbb{R}^m)\) denotes the automorphism group of the input space \(\mathbb{R}^m\). A function \(f(\bm{x},\bm{\theta})\) is said to be \(\mathcal{S}\)-invariant if
\begin{equation}
f(V_s[\bm{x}],\bm{\theta})
=
f(\bm{x},\bm{\theta}),
\quad
\forall s\in\mathcal S,\ \bm{x}\in\mathbb R^m,
\end{equation}
for every admissible parameter \(\bm{\theta}\).
\end{definition}

A standard approach to constructing invariant functions is to build them from \textit{equivariant operations} \cite{Larocca2022,Meyer2023,Cenk2024,Louis2024,Nguyen2024,West2024,nha2025,Das2025}. Typically in our QFM setting, the classical input \(\bm{x}\) is embedded via an encoding layer \(E(\bm{x})\) as defined in Eq.~\eqref{eq:reupload}. This encoding layer is said to be \textit{equivariant} with respect to a symmetry group \(\mathcal{S}\) if, for all \(s\in\mathcal{S}\), there exists a unitary \textit{induced representation} \(R(s)\) such that
\begin{equation}
\label{eq:eq_encoding}
R(s) E(\bm{x}) R(s)^{\dagger}=E(V_s[\bm{x}]),
\end{equation}
where the induced representation follows the definition in \cite{Meyer2023}. Similarly, the trainable layer \(W(\bm{\theta})\) is equivariant with respect to a symmetry element \(s \in \mathcal S\) if
\begin{equation}
\label{eq:eq_trainable}
[W(\bm{\theta}), R(s)] = 0.
\end{equation}
Additionally, the initial state and the observable are considered invariant if they satisfy:
\begin{align}
\label{eq:eq_init}
&R(s)|0\rangle\langle0|R(s)^\dagger = |0\rangle\langle0|, \\
\label{eq:eq_observable}
&R(s)^\dagger \hat{O} R(s) = \hat{O}.
\end{align}

Under these conditions, we establish the following lemma, which captures a common construction behind many GQML methods.

\begin{lemma}[Invariance from equivariance]
\label{lemma:invariance}
A re-uploading quantum model defined in Eq.~\eqref{eq:re-uploading_function} induces a hypothesis function invariant under \(\mathcal{S}\) if it satisfies the conditions of equivariant encoding and trainable layers, alongside an invariant initial state and observable, as in Eqs.~\eqref{eq:eq_encoding}--\eqref{eq:eq_observable}.
\end{lemma}

The proof is detailed in Appendix~\ref{sec:invariance}. While this lemma provides a systematic recipe for constructing invariant functions, enforcing the required equivariant operations on quantum circuits can be difficult. One prominent way to realize these operations is \textit{twirling} \cite{Meyer2023}, which averages an operator or quantum channel over a group representation and projects it into its symmetric counterpart. For an operator \(X\) and an induced representation \(R(s)\) of a finite group \(\mathcal{S}\), the twirling method is defined as
\begin{equation}
\label{eq:twirling}
\mathcal{T}_R[X]=\frac{1}{|\mathcal{S}|} \sum_{s \in \mathcal{S}} R(s) X R(s)^{\dagger}.
\end{equation}
By construction, the resulting operator commutes with the group representation. This enables the systematic design of equivariant layers and invariant observables, a framework further generalized by \citet{Nguyen2024} for quantum channels that are not unitary. However, implementing twirling natively in quantum hardware can impose significant resource overhead, including the need for ancilla qubits, additional quantum circuits for the group representation, or quantum control at each shot \cite{Nguyen2024}. For complex symmetry groups, these costs can make GQML implementations impractical, motivating alternatives with lower resource cost for enforcing invariant hypothesis functions. We discuss these costs in Sec.~\ref{sec:construction_resource_comparison}.

\section{Fourier Symmetrization}
\label{sec:fourier_symmetrization}

This section shows how invariance reorganizes the frequency spectrum and Fourier coefficients. We first derive the corresponding Fourier orbit decomposition, then analyze the effect of this reorganization on coefficient variances, and finally show how invariance can be used for other purposes, such as providing an inductive bias for solving PDEs.

\subsection{Fourier orbit decomposition}
\label{sec:fourier_orbit_decomposition}

Instead of enforcing the symmetry by translating the input space representation into an induced representation on the quantum circuit, we consider settings in which the symmetry is kept at the level of the input to the encoding layer. In this section, \(V\) is the same group representation introduced in Definition~\ref{def:invariant_function}, but we restrict it to the orthogonal case. This restriction covers many symmetry representations used in practice, including permutations, sign flips, rotations, reflections, etc. Specifically, \(V_s\in O(m)\) for all \(s\in\mathcal{S}\), where \(O(m)\) denotes the orthogonal group. Equivalently, \(V\) can be written as \(V:\mathcal{S}\to O(m)\). For the hypothesis function \(f\), we define the average operator \(\mathcal{A}\) over \(\mathcal{S}\) to act on the input variable while holding \(\bm{\theta}\) fixed:
\begin{equation}
\label{eq:group_op_def}
(\mathcal Af)(\bm{x},\bm{\theta})
:=
\frac{1}{|\mathcal S|}
\sum_{s\in\mathcal S}
f(V_s[\bm{x}],\bm{\theta}).
\end{equation}
For any $s\in \mathcal S$, let $l_s:\mathcal{S}\to \mathcal{S}$, $t\mapsto ts$, denote right multiplication by $s$. Since $\mathcal S$ is a group, $l_s$ is a bijection with $l_s^{-1}=l_{s^{-1}}$. Combined with the representation properties of \(V\), namely \(V_{s_1 \circ s_2} = V_{s_1} V_{s_2}\) and \(V_s^{-1} = V_{s^{-1}}\), the averaged function \((\mathcal{A} f)\) is \(\mathcal{S}\)-invariant, i.e.,
\begin{equation}
(\mathcal A f)\!\left(V_s[\boldsymbol{x}],\bm{\theta}\right) = (\mathcal A f)(\boldsymbol{x},\bm{\theta}),
\quad \bm{x} \in \mathbb{R}^m,\ \forall s \in \mathcal{S}.
\end{equation}
To understand how this averaging changes the Fourier series defined in Eq.~\eqref{eq:QFM}, we first relate the orthogonal representation \(V\) to the frequencies. The key identity is
\begin{equation}
V_s[\bm{x}] \cdot\bm{\omega}=\bm{x} \cdot V_{s^{-1}}[\bm{\omega}].
\end{equation}
This identity shows that evaluating a Fourier term at \(V_s[\bm{x}]\) is equivalent to keeping the input fixed and replacing \(\bm{\omega}\) by \(V_{s^{-1}}[\bm{\omega}]\). Two frequency vectors \(\bm{\omega}\) and \(\bm{\omega}'\) are in the same orbit if \(\bm{\omega}'=V_s[\bm{\omega}]\) for some \(s\in\mathcal S\). From each distinct orbit that contains at least one frequency in \(\Omega\), choose one representative \(\bm{\nu}\in\Omega\),  called an \textit{orbit frequency}. The orbit indexed by \(\bm{\nu}\) is
\begin{equation}
[\bm{\nu}] = \{ V_s[\bm{\nu}] \mid s \in \mathcal{S} \},
\end{equation}
where \(s\) ranges over all elements of \(\mathcal S\). The set of these representatives is the \textit{orbit frequency spectrum} \(\tilde{\Omega}\). We say that the frequency spectrum \(\Omega\) is closed under the representation \(V:\mathcal S\to O(m)\) on frequency vectors if \(V_s[\bm{\omega}]\in\Omega\) for every \(\bm{\omega}\in\Omega\) and \(s\in\mathcal S\). In this closed case, the orbits partition \(\Omega\), so \(\Omega\) is a union of complete frequency orbits. If \(\Omega\) is not closed, then for any orbit with \([\bm{\nu}]\nsubseteq\Omega\), the frequencies in \([\bm{\nu}]\setminus\Omega\) complete the part of \(\Omega\) in that orbit to the full orbit \([\bm{\nu}]\).

\begin{figure}[t]
    \centering
    \includegraphics[width=0.9\linewidth]{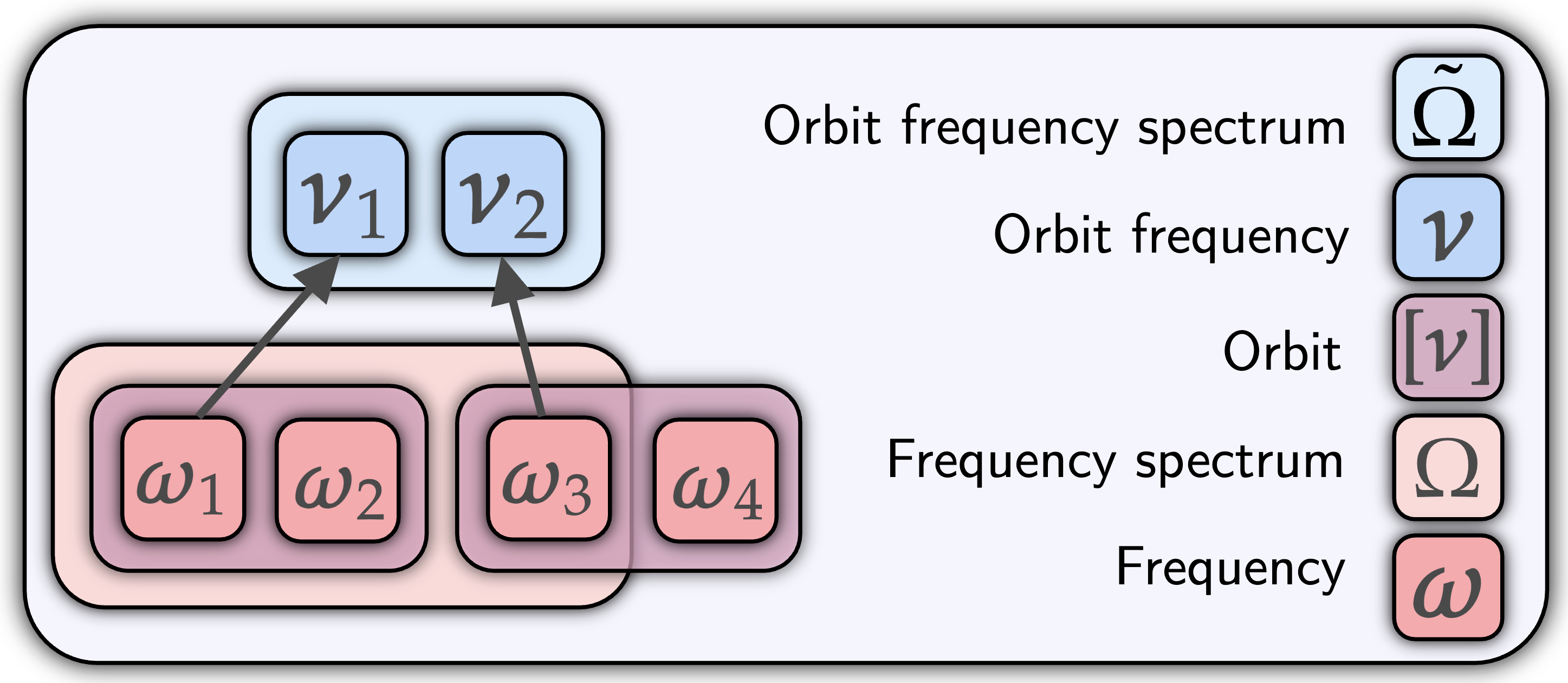}
    \caption{Illustration of the orbit notation. The original frequency spectrum is \(\Omega=\{\bm{\omega}_1,\bm{\omega}_2,\bm{\omega}_3\}\), and the average operator completes the orbit by adding \(\bm{\omega}_4\notin\Omega\). The frequencies \(\bm{\omega}_1\) and \(\bm{\omega}_2\) belong to one orbit, while \(\bm{\omega}_3\) and \(\bm{\omega}_4\) belong to another orbit. We choose \(\bm{\nu}_1\) and \(\bm{\nu}_2\) as orbit frequencies representing \(\bm{\omega}_1\) and \(\bm{\omega}_3\), respectively, so the two orbits are denoted by \([\bm{\nu}_1]\) and \([\bm{\nu}_2]\). The orbit frequencies \(\bm{\nu}_1\) and \(\bm{\nu}_2\) together form the orbit frequency spectrum \(\tilde{\Omega}\).}
    \label{fig:orbit_diagram}
\end{figure}

Figure~\ref{fig:orbit_diagram} illustrates the orbit notation. For example, take \(\mathcal S=\{e,\tau\}\), where \(V_e\) is the identity map and \(V_\tau[(v_1,v_2)]=(v_2,v_1)\). If \(\Omega=\{\bm{\omega}_1,\bm{\omega}_2,\bm{\omega}_3\}\), where \(\bm{\omega}_1=(1,0)\), \(\bm{\omega}_2=(0,1)\), and \(\bm{\omega}_3=(1,2)\), then \([\bm{\omega}_1]=\{(1,0),(0,1)\}\), while \([\bm{\omega}_3]=\{(1,2),(2,1)\}\). This is a non-closed case, since the frequency \(\bm{\omega}_4=(2,1)\notin\Omega\) is introduced to complete the orbit \([\bm{\omega}_3]\). To include such completed orbits in one formula, we set \(c_{\bm{\omega}}(\bm{\theta})=0\) for frequencies outside the original spectrum, \(\bm{\omega}\notin\Omega\). With this convention, averaging over \(V_s[\bm{x}]\) groups Fourier terms by their frequency orbits and yields the Fourier orbit decomposition:
\begin{proposition}[Fourier orbit decomposition]
\label{prop:FOD}
Let \(\mathcal S\) be a finite group, \(V:\mathcal S\to O(m)\) an orthogonal representation, \(\Omega\) a finite frequency spectrum, and \(\tilde{\Omega}\) the orbit frequency spectrum. For \(f(\bm{x},\bm{\theta})=\sum_{\bm{\omega}\in\Omega}c_{\bm{\omega}}(\bm{\theta})e^{-i\,\bm{x}\cdot\bm{\omega}}\) defined in Eq.~\eqref{eq:QFM}, set \(c_{\bm{\omega}}(\bm{\theta})=0\) for \(\bm{\omega}\notin\Omega\). Then the averaged function \(\mathcal{A}f\) defined by Eq.~\eqref{eq:group_op_def} is \(\mathcal{S}\)-invariant and satisfies
\begin{equation}
\label{eq:FOD}
(\mathcal{A}f)(\bm{x}, \bm{\theta}) = \sum_{\bm{\nu} \in \tilde{\Omega}} a_{\bm{\nu}}(\bm{\theta}) \phi_{[\bm{\nu}]}(\bm{x}),
\end{equation}
where the symmetrized coefficient is
\begin{equation}
\label{eq:symmetrized_coefficient}
a_{\bm{\nu}}(\bm{\theta}) := \sum_{\bm{\omega} \in [\bm{\nu}]} c_{\bm{\omega}}(\bm{\theta}),
\end{equation}
and the normalized orbit basis is
\begin{equation}
\label{eq:orbit_basis}
\phi_{[\bm{\nu}]}(\bm{x}) = \frac{1}{|[\bm{\nu}]|} \sum_{\tilde{\bm{\omega}} \in [\bm{\nu}]} e^{-i \bm{x} \cdot \tilde{\bm{\omega}}}.
\end{equation}

\end{proposition}

The proof is detailed in Appendix~\ref{app:FOD}. When \(f(\bm{x},\bm{\theta})\) is already \(\mathcal{S}\)-invariant, then
\(f(\bm{x},\bm{\theta})=\frac{1}{|\mathcal{S}|}\sum_{s\in\mathcal{S}}f(V_s[\bm{x}],\bm{\theta})\),
so \(f\) coincides with its averaged form \(\mathcal{A}f\) in Eq.~\eqref{eq:group_op_def}. Therefore, Proposition~\ref{prop:FOD} also shows that any \(\mathcal{S}\)-invariant QFM under this orthogonality assumption on \(V\) admits a Fourier orbit decomposition.

\begin{corollary}[Invariant Fourier orbit decomposition]
\label{cor:invariant_qfm_orbit_decomposition}
Under the assumption of Proposition~\ref{prop:FOD}, any \(\mathcal S\)-invariant \(f(\bm{x},\bm{\theta})\) defined by Eq.~\eqref{eq:QFM} and satisfying Definition~\ref{def:invariant_function} admits the Fourier orbit decomposition in Eq.~\eqref{eq:FOD}.
\end{corollary}

Consequently, all such invariant models built from the data re-uploading circuit in Eq.~\eqref{eq:reupload} and Eq.~\eqref{eq:re-uploading_function} admit this Fourier orbit decomposition, which covers a broad class of GQML models. The Fourier orbit decomposition is useful for interpreting invariant QFMs in several ways. First, in the normalized orbit basis of Eq.~\eqref{eq:orbit_basis}, symmetrization sums the Fourier coefficients in each orbit into \(a_{\bm{\nu}}\). Its variance can therefore be much larger than that of any individual Fourier coefficient, which can mitigate vanishing expressivity \cite{Hela2025} at the level of the symmetrized coefficients. Second, because frequencies related by the symmetry are grouped into frequency orbits, the orbit frequency spectrum \(\tilde{\Omega}\) can be substantially smaller than the original frequency spectrum \(\Omega\), which can facilitate dequantization through the construction of a classical surrogate \cite{Schreiber2023}. Finally, the orbit basis reveals the structure imposed by the symmetry in a QFM and can help identify a useful inductive bias. The usefulness of the orbit basis therefore depends on whether the chosen symmetry matches the target or produces a useful basis through the feature map. This point can be illustrated in Section \ref{sec:basis_function} by sign-flip averaging, which produces Chebyshev polynomial basis functions for solving PDEs.

\subsection{Symmetrized coefficient variance}
\label{sec:coefficient_analysis}

For the variance analysis below, the redundancy of the frequencies within each orbit is important for determining the coefficient variances. We follow the notation of \citet{Hela2025} for redundancy and extend it to orbits.
\begin{definition}[Orbit Frequency Generator]
\label{def:frequency_generator}
For each encoding layer in Eq.~\eqref{eq:reupload}, write
\begin{equation}
\label{eq:encoding_lambda}
E^{(l)}(\bm{x})|j\rangle=e^{-i\bm{x}\cdot\bm{\lambda}^{(l)}_j}|j\rangle.
\end{equation}
Here \(l\in[1,L]\), \(j\in[1,d]\), and \(d=2^n\). The vector \(\bm{\lambda}^{(l)}_j=(\lambda^{(l)}_{j,1},\ldots,\lambda^{(l)}_{j,m})\in\mathbb R^m\) is the eigenvalue vector associated with \(|j\rangle\) in the \(l\)-th encoding layer, where \(\lambda^{(l)}_{j,k}\) is the eigenvalue corresponding to the input coordinate \(x_k\). For index paths \(J=(j_1,\ldots,j_L)\in[1,d]^L\), \(J'=(j'_1,\ldots,j'_L)\in[1,d]^L\), and a frequency orbit \([\bm{\nu}]\), define the \textit{orbit frequency generator}
\begin{equation}
\label{eq:orbit_frequency_generator}
\mathfrak R_{[\bm{\nu}]}
:=
\left\{
(J,J')\in[1,d]^{2L}:
\sum_{l=1}^{L}
\left(
\bm{\lambda}^{(l)}_{j_l}
-
\bm{\lambda}^{(l)}_{j'_l}
\right)
\in[\bm{\nu}]
\right\},
\end{equation}
and the \textit{frequency generator} is
\begin{equation}
\label{eq:single_frequency_generator}
\mathfrak R(\bm{\omega})
:=
\left\{
(J,J')\in[1,d]^{2L}:
\sum_{l=1}^{L}
\left(
\bm{\lambda}^{(l)}_{j_l}
-
\bm{\lambda}^{(l)}_{j'_l}
\right)
=\bm{\omega}
\right\}.
\end{equation}
\end{definition}

The cardinality \(|\mathfrak R_{[\bm{\nu}]}|\) is the total redundancy of the orbit, while \(|\mathfrak R(\bm{\omega})|\) measures the redundancy of an individual frequency. The next proposition records a simple compatibility condition under which this redundancy is constant along a symmetry orbit.

\begin{figure*}[!t]
    \centering
    \includegraphics[width=1\linewidth]{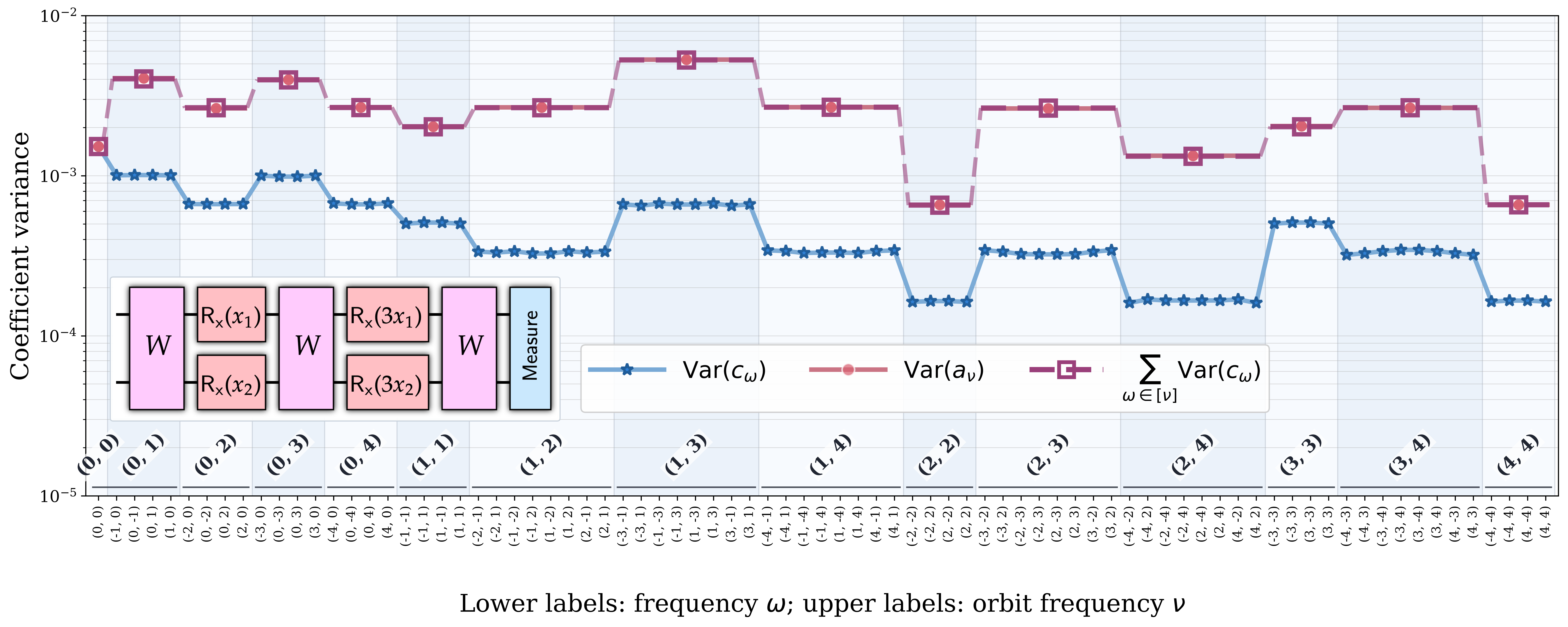}
    \caption{Fourier coefficient variances reorganized into frequency orbits with respect to the hyperoctahedral group under the exact 2-design assumption. The circuit in the lower left is the QFM used in this numerical example, and its output is obtained by measuring the global observable \(\hat O=(\ket{0}\!\bra{0})^{\otimes 2}\). The x-axis uses two rows of labels: the lower row lists the original frequencies \(\bm{\omega}\), and the upper row labels each orbit by its orbit frequency \(\bm{\nu}\). For example, under the coordinate permutations and independent sign flips in the hyperoctahedral group, the orbit frequency \((0,1)\) represents the orbit containing frequencies \((-1,0)\), \((0,-1)\), \((0,1)\), and \((1,0)\). The blue curve shows the variance of each coefficient \(c_{\bm{\omega}}\), the pink curve shows the variance of each symmetrized coefficient \(a_{\bm{\nu}}\), and the purple curve shows the sum \(\sum_{\bm{\omega}\in[\bm{\nu}]}\operatorname{Var}(c_{\bm{\omega}})\) over the same orbit. The overlap between the pink and purple curves numerically confirms Theorem~\ref{theorem:exp_variance}: the variance of a symmetrized coefficient equals the sum of the variances inside its orbit. Their separation from the blue curve shows the gain in coefficient variance produced by grouping frequencies into orbits.}
    \label{fig:coefficient_variance_distribution}
\end{figure*}

\begin{proposition}[Redundancy Preservation]
\label{prop:redundancy_preservation_representation_elements}
With the notation of the QFM in Eq.~\eqref{eq:reupload} and the frequency generators in Definition~\ref{def:frequency_generator}, let \(V:\mathcal S\to O(m)\), \(s\mapsto V_s\), be the group representation and fix \(s\in\mathcal S\). Suppose that for every \(l\in[1,L]\), there exists a bijection \(\tau_{s,l}:[1,d]\to[1,d]\) such that
\begin{equation}
\label{eq:redundancy_condition}
\bm{\lambda}^{(l)}_{\tau_{s,l}(j)}
=
V_s[\bm{\lambda}^{(l)}_j],
\qquad
j=1,\ldots,d.
\end{equation}
Then, for all \(\bm{\omega}\in\Omega\),
\begin{equation}
\left|\mathfrak R(V_s[\bm{\omega}])\right|
=
\left|\mathfrak R(\bm{\omega})\right|,
\qquad
\bm{\omega}\in\Omega.
\end{equation}
\end{proposition}
The proof is detailed in Appendix~\ref{app:redundancy_preservation}. At the level of eigenvalue sets, Eq.~\eqref{eq:redundancy_condition} requires \(V_s\) to introduce no eigenvalue vectors beyond those already present in each encoding layer. In the Fourier orbit decomposition, \(V_s\) is applied to frequencies to organize them into frequency orbits. In Proposition~\ref{prop:redundancy_preservation_representation_elements}, the same \(V_s\) is applied to the eigenvalues that generate these frequencies. Under the layerwise relabeling condition, the transformed eigenvalues are a bijective reindexing of the eigenvalues in the same layer. This reindexing gives a bijection between the corresponding frequency generators. Hence, when the condition holds for all \(s\in\mathcal S\), all frequencies in the same orbit have the same redundancy.

The orbit structure fixes which coefficients are grouped, while the redundancy describes how the encoding populates the frequencies in each orbit. We now connect this structure to trainability through the variance of the symmetrized coefficient, which decomposes as
\begin{equation}
\operatorname{Var}_{\bm{\theta}}(a_{\bm{\nu}})
=
\sum_{\bm{\omega}\in[\bm{\nu}]}
\operatorname{Var}_{\bm{\theta}}(c_{\bm{\omega}})
+
\sum_{\substack{\bm{\omega},\bm{\omega}'\in[\bm{\nu}]\\ \bm{\omega}\neq \bm{\omega}'}}
\operatorname{Cov}_{\bm{\theta}}(c_{\bm{\omega}},c_{\bm{\omega}'}).
\end{equation}
Thus, the variance of a symmetrized coefficient depends on the original coefficient variances and the cross-covariances between distinct frequencies in the same orbit. We analyze this statistical question next.

\subsubsection{Exact 2-designs}

Let \(\mathcal H^{(2)}\) denote the Haar second-moment channel,
\begin{equation}
\label{eq:Haar_operator}
\mathcal{H}^{(2)}[X]
:=
\mathbb{E}_{U\sim\mu_H}
\left[
(U\otimes U)
X
(U^\dagger\otimes U^\dagger)
\right],
\end{equation}
where \(X\) is an arbitrary operator on \((\mathbb C^d)^{\otimes 2}\) and \(\mu_H\) is the Haar probability measure on the unitary group \(U(d)\), with \(d=2^n\) for an \(n\)-qubit Hilbert space. For a trainable layer \(W\), define its second-moment channel by
\begin{equation}
\label{eq:second_moment_channel}
\mathcal{M}^{(2)}[X]
:=
\mathbb{E}_{\bm{\theta}}
\left[
(W\otimes W)
X
(W^\dagger\otimes W^\dagger)
\right].
\end{equation}
We say that \(W\) forms an exact 2-design when \(\mathcal M^{(2)}=\mathcal H^{(2)}\). We first consider the idealized case in which the trainable layers are independent exact 2-designs. The following lemma removes the cross-covariances under this assumption.

\begin{lemma}[Fourier Coefficient Decoupling]
\label{lemma:fourier_coefficient_decoupling}
Consider the QFM hypothesis function \(f(\bm{x},\bm{\theta})\) in Eq.~\eqref{eq:QFM}, with Fourier coefficients \(c_{\bm{\omega}}(\bm{\theta})\). If each trainable layer \(W^{(l)}\) forms an independent exact 2-design, then for any distinct frequencies \(\bm{\omega}\neq\bm{\omega}'\),
\begin{align}
\mathbb{E}_{\bm{\theta}} [c_{\bm{\omega}} c_{\bm{\omega}'}^*] &= 0,\\
\operatorname{Cov}_{\bm{\theta}}(c_{\bm{\omega}}, c_{\bm{\omega}'}) &= 0.
\end{align}
\end{lemma}

The proof is detailed in Appendix~\ref{app:fourier_coefficient_decoupling}. Lemma~\ref{lemma:fourier_coefficient_decoupling} turns the deterministic orbit sum in Proposition~\ref{prop:FOD} into a statistical variance identity: once distinct Fourier coefficients decouple, the variance of \(a_{\bm{\nu}}\) is the sum of the variances of the coefficients in its orbit. Theorem~\ref{theorem:exp_variance} states this identity explicitly.

\begin{theorem}[Symmetrized Coefficient Variance under Exact 2-designs]
\label{theorem:exp_variance}
Consider the QFM hypothesis function \(f(\bm{x},\bm{\theta})\) in Eq.~\eqref{eq:QFM}, with Fourier coefficients \(c_{\bm{\omega}}(\bm{\theta})\) satisfying \(c_{\bm{\omega}}(\bm{\theta})=0\) for \(\bm{\omega}\notin\Omega\). Let \(a_{\bm{\nu}}(\bm{\theta})\) be the symmetrized coefficient defined in Eq.~\eqref{eq:symmetrized_coefficient}. If each trainable layer forms an independent exact 2-design, then for any orbit frequency \(\bm{\nu}\) in the orbit frequency spectrum \(\tilde{\Omega}\),
\begin{equation}
\label{eq:exp_variance}
\operatorname{Var}_{\bm{\theta}}(a_{\bm{\nu}}) = \sum_{\bm{\omega} \in [\bm{\nu}]} \operatorname{Var}_{\bm{\theta}}(c_{\bm{\omega}}) = |[\bm{\nu}]| \cdot \overline{\operatorname{Var}}_{\bm{\theta}}(c_{\bm{\omega}}),
\end{equation}
where \(\overline{\operatorname{Var}}_{\bm{\theta}}(c_{\bm{\omega}})\) denotes the mean variance over the distinct frequencies within the orbit \([\bm{\nu}]=\{ V_s[\bm{\nu}] \mid s \in \mathcal{S} \}\).
\end{theorem}

The proof is detailed in Appendix~\ref{app:symmetrized_coefficient_variance}. Theorem~\ref{theorem:exp_variance} combines the two preceding steps. Proposition~\ref{prop:FOD} identifies each symmetrized coefficient as a sum of Fourier coefficients, while Lemma~\ref{lemma:fourier_coefficient_decoupling} removes the cross-covariance terms that would otherwise appear in the variance of this sum. Under the exact 2-design assumption, the variance gain from symmetrization is therefore controlled by the orbit size and the variances of the coefficients in the orbit. For re-uploading models, \citet[Theorem~6]{Hela2025} express these coefficient variances in terms of redundancy counts derived from the encoding layers. The layerwise reindexing in Proposition~\ref{prop:redundancy_preservation_representation_elements} preserves these counts along each orbit. Consequently, if the condition in Proposition~\ref{prop:redundancy_preservation_representation_elements} holds for all group elements, the coefficient variances are equal within each orbit, so $\overline{\operatorname{Var}}_{\bm{\theta}}(c_{\bm{\omega}})=\operatorname{Var}_{\bm{\theta}}(c_{\bm{\omega}})$ for any \(\bm{\omega}\in[\bm{\nu}]\).
In this case, Theorem~\ref{theorem:exp_variance} reduces to
\begin{equation}
\operatorname{Var}_{\bm{\theta}}(a_{\bm{\nu}})=|[\bm{\nu}]|\operatorname{Var}_{\bm{\theta}}(c_{\bm{\omega}}), \; \forall \bm{\omega}\in[\bm{\nu}].
\end{equation}

Figure~\ref{fig:coefficient_variance_distribution} illustrates this mechanism using a two-qubit exponential encoding QFM with two encoding layers. As shown by the blue curve, \(\operatorname{Var}_{\bm{\theta}}(c_{\bm{\omega}})\) is (almost) the same for all frequencies \(\bm{\omega}\) within each frequency orbit. The QFM shown in the figure satisfies the condition in Proposition~\ref{prop:redundancy_preservation_representation_elements}, as verified below. Here we consider the hyperoctahedral group \(B_m\) because it captures two common coordinate symmetries at once: permutations of input variables and independent sign flips. The hyperoctahedral group \(B_m\) is the semidirect product \(\mathbb{Z}_2^m\rtimes S_m\), where \(\mathbb{Z}_2^m\) is the elementary abelian 2-group and \(S_m\) is the symmetric group. Its order is \(|B_m|=2^m m!\). We use its standard signed permutation representation \(V:B_m\to O(m)\), \(s\mapsto V_s\). For \(s=(\bm{\epsilon},\pi)\), with \(\bm{\epsilon}\in\{\pm1\}^m\) and \(\pi\in S_m\), \(V_s\) permutes the coordinates by \(\pi\) and then flips their signs according to \(\bm{\epsilon}\). Equivalently,
\begin{equation}
\label{eq:B2_Vs}
(V_s[\bm{\omega}])_i
=
\epsilon_i\omega_{\pi^{-1}(i)},
\end{equation}
so the \(i\)-th output coordinate comes from the \(\pi^{-1}(i)\)th input coordinate. For the two-variable model considered here, \(m=2\), so the relevant group is \(B_2\).

To verify the condition in Proposition~\ref{prop:redundancy_preservation_representation_elements}, note that encoding layer \(l\in\{1,2\}\) applies \(R_X(\alpha_l x_1)\otimes R_X(\alpha_l x_2)\), where \(\alpha_l\) is the angle coefficient. Since \(R_X(t)=e^{-itX/2}\) and \(X\) has eigenvalues \(\pm1\), each independent choice of \(\sigma_1,\sigma_2\in\{\pm1\}\) gives an eigenvalue vector \((\sigma_1\alpha_l/2,\sigma_2\alpha_l/2)\) in layer \(l\). These four choices give four distinct eigenvalue vectors, which we index by \(j=1,\ldots,4\) and denote by \(\bm{\lambda}^{(l)}_j\). For \(s=(\bm{\epsilon},\pi)\in B_2\), applying \(V_s\) to \(\bm{\lambda}^{(l)}_j=(\sigma_1\alpha_l/2,\sigma_2\alpha_l/2)\) gives \(V_s[\bm{\lambda}^{(l)}_j]=(\epsilon_1\sigma_{\pi^{-1}(1)}\alpha_l/2,\epsilon_2\sigma_{\pi^{-1}(2)}\alpha_l/2)\), which is again one of the four eigenvalue vectors in the same layer. Since \(V_s\) is invertible, it permutes these four eigenvalue vectors bijectively. Let \(\tau_{s,l}(j)\) be the index of \(V_s[\bm{\lambda}^{(l)}_j]\). Then \(\tau_{s,l}\) is a bijection and \(\bm{\lambda}^{(l)}_{\tau_{s,l}(j)}=V_s[\bm{\lambda}^{(l)}_j]\), as required by Eq.~\eqref{eq:redundancy_condition}. Proposition~\ref{prop:redundancy_preservation_representation_elements} therefore shows that all frequencies in the same \(B_2\) orbit have the same redundancy. Under the exact 2-design assumption, this equal redundancy implies equal Fourier coefficient variance, explaining the orbitwise equality shown by the blue curve.

In Figure~\ref{fig:coefficient_variance_distribution}, the two encoding layers use angle coefficients \(\alpha_1=3^0\) and \(\alpha_2=3^1\), following the exponential encoding scheme \cite{Shin_2023}. Applied to both input variables, these layers give the frequency spectrum \(\Omega=\{-4,-3,\ldots,4\}^2\), containing \(9^2=81\) frequencies \cite[Appendix~A.2.2]{Hela2025}. The \(B_2\) representation organizes these Fourier coefficients into frequency orbits. Following the numerical protocol in \citet[Sec.~5.2]{Hela2025}, we consider each trainable layer \(W\) as an exact 2-design using five repetitions of the Strongly Entangling Ansatz \cite{Schuld2020}. Using the global observable \(\hat O=(\ket{0}\!\bra{0})^{\otimes 2}\), we estimate the coefficient variances from \(3\times10^4\) random parameter samples.

The close overlap between the measured variance of the symmetrized coefficient and the sum of the coefficient variances confirms Theorem~\ref{theorem:exp_variance} in this finite setting. The separation between the pink and blue curves shows how the normalized orbit basis mitigates vanishing expressivity: a symmetrized coefficient can retain substantially more variance than an individual Fourier coefficient. This variance gain concerns \(a_{\bm{\nu}}\) in the normalized orbit basis and does not increase the variance of the full hypothesis function. This gain can be viewed as arising from the factor $1/|[\bm{\nu}]|$ in the normalized orbit basis of Eq.~\eqref{eq:orbit_basis}. This factor gives $\phi_{[\bm{\nu}]}(\bm{0})=1$ and $|\phi_{[\bm{\nu}]}(\bm{x})|\leq 1$, the same range as an individual Fourier basis function $e^{-i\bm{x}\cdot\bm{\omega}}$, so that $a_{\bm{\nu}}$ and $c_{\bm{\omega}}$ are coefficients of basis functions with the same range. Under sign-flip averaging, it also makes $\phi_{[\bm{\nu}]}$ a product of Chebyshev polynomials, so that $a_{\bm{\nu}}$ are exactly the Chebyshev coefficients in Proposition~\ref{prop:quantum_chebyshev_model}. The coefficient variance analysis does not by itself guarantee larger loss gradients or faster optimization. 

We now recall the notion of vanishing expressivity and convert this gain in coefficient variance into an asymptotic statement. The expressivity of a re-uploading model relies on its accessible frequency spectrum, and exponential encoding schemes can make this spectrum grow exponentially with the number of qubits or layers \cite{Shin_2023}. The trade-off is that, with exact 2-design trainable layers, the variance of frequency coefficients may decay exponentially, which severely limits their trainability. This phenomenon, formalized as \emph{vanishing expressivity} \cite{Hela2025}, occurs when there exists a frequency
 \(\bm{\omega} \in \Omega\) such that
\begin{equation}
\operatorname{Var}_{\bm{\theta}}\left[c_{\bm{\omega}}(\bm{\theta})\right] = \Theta\!\left(\frac{1}{b^{n}}\right)
\end{equation}
for some constant \(b>1\), where \(n\) denotes the number of qubits. Theorem~\ref{theorem:exp_variance} gives the precise way in which symmetrization can counteract this decay with the normalized orbit basis: the variance of \(a_{\bm{\nu}}\) includes an orbit-size factor. Throughout this work, mitigation of vanishing expressivity in the symmetrized model is understood with respect to the symmetrized coefficients \(a_{\bm{\nu}}\) in the normalized orbit basis of Eq.~\eqref{eq:orbit_basis}. The following corollary states when this mitigation occurs.

\begin{corollary}[Mitigating Vanishing Expressivity]
\label{cor:mitigating_vanishing_expressivity}
Under the assumptions of Theorem~\ref{theorem:exp_variance}, suppose the Fourier coefficients \(c_{\bm{\omega}}(\bm{\theta})\) of the QFM in Eq.~\eqref{eq:QFM} exhibit vanishing expressivity on a frequency orbit \([\bm{\nu}]\), such that \(\operatorname{Var}_{\bm{\theta}}[c_{\bm{\omega}}] = \Theta(1/b^{n})\) for all \(\bm{\omega} \in [\bm{\nu}]\), with some \(b>1\). If the chosen symmetry group \(\mathcal{S}\) yields an orbit size satisfying
\begin{equation}
    |[\bm{\nu}]| \in \Omega\left(\frac{b^n}{c(n)}\right),
\end{equation}
where $c(n)$ is a function, then the symmetrized coefficient \(a_{\bm{\nu}}(\bm{\theta})\) defined in Eq.~\eqref{eq:symmetrized_coefficient} satisfies \(\operatorname{Var}_{\bm{\theta}}[a_{\bm{\nu}}] \in \Omega(1/c(n))\).
\end{corollary}

The usefulness of Corollary~\ref{cor:mitigating_vanishing_expressivity} depends on whether the chosen symmetry can produce large frequency orbits. Typically, the relevant orbit size grows with the number of variables \(m\), and \(m\) scales proportionally with the number of qubits \(n\). Therefore, an orbit size that grows exponentially with \(m\) also grows exponentially with \(n\). For the hyperoctahedral group \(B_m\), this orbit size condition can be made explicit. Consider its representation on an \(m\)-dimensional frequency vector \(\bm{\nu}\), where the elements \(V_s\) permute coordinates and flip signs. Suppose that \(\bm{\nu}\) has \(k\) nonzero coordinates and \(m-k\) zero coordinates. Let \(q_1,\ldots,q_r\) be the multiplicities of the distinct nonzero absolute values among these \(k\) coordinates, so \(\sum_{j=1}^r q_j=k\). Its stabilizer
\(\operatorname{Stab}(\bm{\nu})=\{s\in B_m:V_s[\bm{\nu}]=\bm{\nu}\}\)
has size \(2^{m-k}(m-k)!\prod_{j=1}^r q_j!\). Since \(|B_m|=2^m m!\), the Orbit-Stabilizer Theorem gives
\begin{equation}
\label{eq:orbit-stabilizer}
\begin{aligned}
|[\bm{\nu}]|
=
\frac{|B_m|}{|\operatorname{Stab}(\bm{\nu})|}
=
\frac{2^k m!}{(m-k)!\prod_{j=1}^r q_j!}.
\end{aligned}
\end{equation}
When the nonzero absolute values are all distinct, every \(q_j=1\), and Eq.~\eqref{eq:orbit-stabilizer} reduces to \(2^k m!/(m-k)!\). When \(k\) scales linearly with \(m\), the factor \(2^k\) already gives exponential orbit growth in \(m\) or \(n\) under the hyperoctahedral group. More generally, Corollary~\ref{cor:mitigating_vanishing_expressivity} shows that once the orbit size satisfies the required lower bound, symmetrization mitigates vanishing expressivity with the normalized orbit basis when \(c(n)=c^n\) with \(c<b\), and can raise the symmetrized coefficient variance to a polynomial scale when \(c(n)=\operatorname{poly}(n)\). The exact 2-design analysis is the cleanest version of the mechanism, but it is not the only regime of interest. The next subsection relaxes the assumption and asks how much deviation from an exact 2-design can be tolerated before the orbit-variance gain is lost.

\subsubsection{\texorpdfstring{\(\varepsilon\)-approximate}{Epsilon-approximate} 2-designs}

The exact 2-design assumption is analytically useful but idealized. We therefore consider the \(\varepsilon\)-approximate 2-design setting for a single layer \(L=1\). For \(W\in\{W^{(1)},W^{(2)}\}\), where \(W^{(1)}\) and \(W^{(2)}\) are the two trainable layers in the single layer \(L=1\) specialization of Eq.~\eqref{eq:reupload}, let \(\mathcal M^{(2)}\) be the corresponding second-moment channel in Eq.~\eqref{eq:second_moment_channel} constructed from that particular \(W\) and its parameter distribution. We assume that both choices of \(W\) are \(\varepsilon\)-approximate 2-designs in the sense that
\begin{equation}
\label{eq:approx_2_design_channel_assumption}
\left\|
\mathcal{M}^{(2)}-\mathcal{H}^{(2)}
\right\|_{\infty}
\leq
\varepsilon .
\end{equation}
Here \(\mathcal{H}^{(2)}\) is defined in Eq.~\eqref{eq:Haar_operator}, and \(\|\cdot\|_{\infty}\) denotes the spectral norm (Schatten $\infty$-norm) of the matrix form of the corresponding superoperator. With this notion of approximation, the first step is to ask whether the decoupling of Fourier coefficients in Lemma~\ref{lemma:fourier_coefficient_decoupling} survives perturbatively.
\begin{lemma}
\label{lemma:distinct_frequency_properties}
Consider the single-uploading QFM hypothesis function \(f(\bm{x},\bm{\theta})\) in Eq.~\eqref{eq:QFM} with \(L=1\) and Fourier coefficients \(c_{\bm{\omega}}(\bm{\theta})\). If each trainable layer \(W^{(l)}\) forms an \(\varepsilon\)-approximate 2-design in the sense of Eq.~\eqref{eq:approx_2_design_channel_assumption}, then for any two distinct nonzero frequencies \(\bm{\omega}\neq\bm{\omega}'\), we have
\begin{align}
&\left|
\mathbb{E}_{\bm{\theta}}
\left[
c_{\bm{\omega}}
\right]
\right|
\leq
\varepsilon^2\|\hat O\|_F,\\
&\left|
\mathbb{E}_{\bm{\theta}}
\left[
c_{\bm{\omega}}c_{\bm{\omega}'}^*
\right]
\right|
\leq
\varepsilon^2
\|\hat O\|_F^2,\\
&\left|
\operatorname{Cov}_{\bm{\theta}}
\left(
c_{\bm{\omega}},c_{\bm{\omega}'}
\right)
\right|
\leq
\|\hat O\|_F^2
\left(
\varepsilon^2+\varepsilon^4
\right),
\end{align}
where \(\|\cdot\|_{F}\) denotes the Frobenius norm.
\end{lemma}

The proof is detailed in Appendix~\ref{sec:fourier_coefficient_approximate_2design}, specifically in Lemmas~\ref{lemma:nonzero_frequency_first_moment_bound} and~\ref{lemma:approx_fourier_coefficient_decoupling} and Corollary~\ref{cor:approx_fourier_coefficient_covariance_decoupling}. Lemma~\ref{lemma:distinct_frequency_properties} plays the same structural role as Lemma~\ref{lemma:fourier_coefficient_decoupling}, but with exact cancellation replaced by a controlled covariance error that vanishes as \(\varepsilon\to0\). The next theorem gives the approximate analogue of Theorem~\ref{theorem:exp_variance}.

\begin{figure*}[hbt]
    \centering
    \makebox[\textwidth][c]{%
    \subfloat[Coefficient variances.]{
        \includegraphics[width=0.475\textwidth]{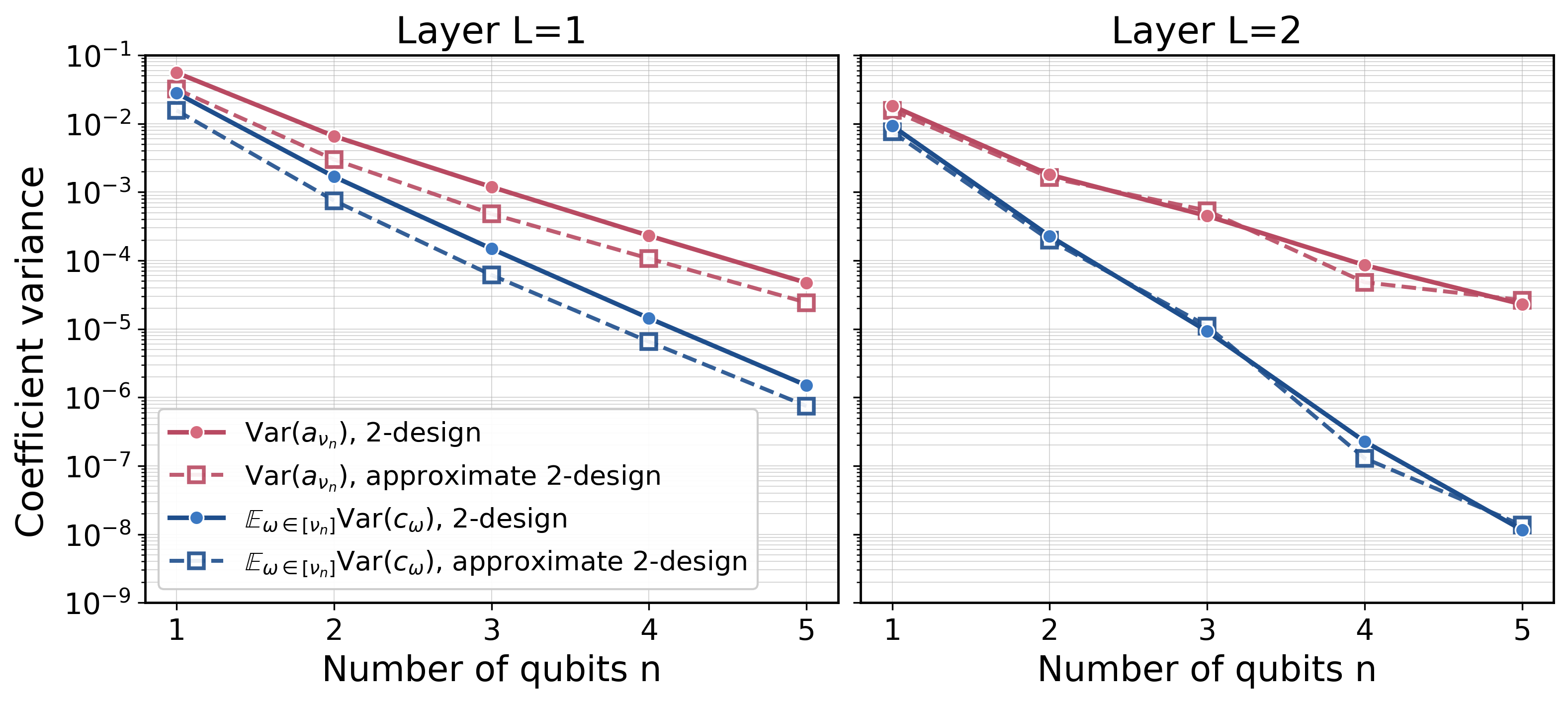}
        \label{fig:variance_values_comparison}
    }%
    \hfill
    \subfloat[Variance ratios.]{
        \includegraphics[width=0.475\textwidth]{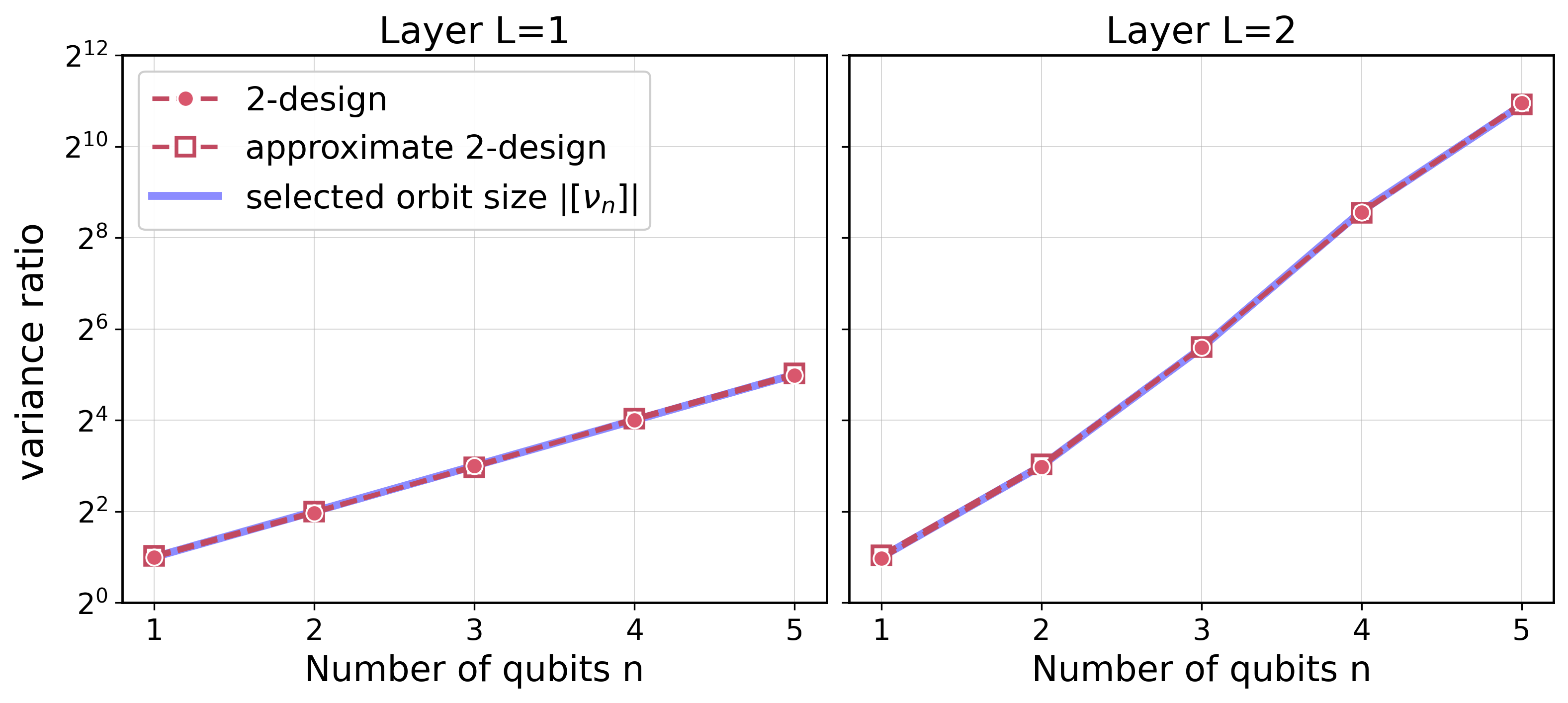}
        \label{fig:variance_ratio_comparison}
    }%
    }
    \caption{Fourier coefficient variance analysis for exponential encoding QFMs. We consider an \(n\)-qubit, \(n\)-variable QFM under the hyperoctahedral group \(B_n\), with \(L=1\) or \(L=2\). The \(l\)-th encoding layer is \(\bigotimes_{j=1}^{n}R_Z(3^{l-1}x_j)\), and the output is obtained by measuring the local observable \(\hat O_{\mathrm L}\) defined in Eq.~\eqref{eq:local_observable}. Each trainable layer uses five repetitions of the Strongly Entangling Ansatz for the exact 2-design reference and one repetition for the \(\varepsilon\)-approximate 2-design setting. The variances are estimated from \(10^4\) random parameter samples. For \(L=1\), we select \(\bm{\nu}_n=(1,\ldots,1)\). For \(L=2\), we select \(\bm{\nu}_n=(1,\ldots,n)\) for \(n=1,2,3,4\) and \(\bm{\nu}_5=(1,1,2,3,4)\). (a) Red curves show \(\operatorname{Var}(a_{\bm{\nu}_n})\), and blue curves show \(\mathbb{E}_{\bm{\omega}\in[\bm{\nu}_n]}\operatorname{Var}(c_{\bm{\omega}})\). Solid circles and dashed squares denote the exact (five repetitions of the Strongly Entangling Ansatz) and \(\varepsilon\)-approximate 2-design (one repetitions of the Strongly Entangling Ansatz) settings, respectively. (b) The corresponding variance ratio \(\operatorname{Var}(a_{\bm{\nu}_n})/\mathbb{E}_{\bm{\omega}\in[\bm{\nu}_n]}\operatorname{Var}(c_{\bm{\omega}})\) is compared with the selected orbit size \(|[\bm{\nu}_n]|\) presented in Eq.~\eqref{eq:orbit-stabilizer}. The curves nearly overlap, numerically confirming Eq.~\eqref{eq:exp_variance} in Theorem~\ref{theorem:exp_variance}.}
    \label{fig:variance_values}
\end{figure*}

\begin{theorem}[Symmetrized Coefficient Variance under \(\varepsilon\)-Approximate 2-designs]
\label{theorem:tight_approx_exp_variance}
Consider the \(n\)-qubit single-uploading QFM in Eq.~\eqref{eq:QFM} with \(L=1\) and symmetrized coefficients \(a_{\bm{\nu}}(\bm{\theta})\) defined in Eq.~\eqref{eq:symmetrized_coefficient}. Suppose the two trainable layers form independent \(\varepsilon\)-approximate 2-designs in the sense of Eq.~\eqref{eq:approx_2_design_channel_assumption}. For any frequency orbit $[\bm{\nu}]$ excluding the zero frequency, let \(R=|\mathfrak R_{[\bm{\nu}]}|\) be its total redundancy as in Definition~\ref{def:frequency_generator}. Then the deviation from the exact 2-design variance satisfies
\begin{equation}
\label{eq:exact_tight_bound}
\left|
\operatorname{Var}_{\bm{\theta}}(a_{\bm{\nu}})
-
\operatorname{Var}_{\text{2-design}}(a_{\bm{\nu}})
\right|
\in
\mathcal{O}\left( \varepsilon \|\hat{O}\|_F^2 \left( \frac{\sqrt{R}}{4^n} + \varepsilon \right) \right).
\end{equation}
Here \(\operatorname{Var}_{\text{2-design}}(a_{\bm{\nu}})\) denotes the variance of the same symmetrized coefficient when the two trainable layers form independent exact 2-designs.
\end{theorem}

The proof is detailed in Appendix \ref{app:approximate_theorem}, and the exact bound is shown in Eq.~\eqref{eq:exact_coeff_bound}. Theorem~\ref{theorem:tight_approx_exp_variance} is therefore a stability statement for Theorem~\ref{theorem:exp_variance}: the symmetrized coefficient remains close to its exact 2-design variance when the second moment channels are close to Haar and the accumulated orbit redundancy \(R\) is controlled. It is a worst case sufficient bound in spectral norm, so it should be distinguished from the empirical behavior of a particular shallow ansatz.

Figure~\ref{fig:variance_values} illustrates the coefficient variance gain produced by symmetrization. In contrast to Figure~\ref{fig:coefficient_variance_distribution}, the QFMs in this figure are measured with the \textit{local observable}
\begin{equation}
\label{eq:local_observable}
\hat O_{\mathrm L}
=
\frac{1}{n}
\sum_{j=1}^{n}
\left(\ket{0}\!\bra{0}\right)_j
\otimes I_{\bar j},
\end{equation}
where \(I_{\bar j}\) denotes the identity operator on all qubits other than qubit \(j\). Figure~\ref{fig:variance_values_comparison} compares the variance of the selected symmetrized coefficient with the mean variance of the Fourier coefficients in its orbit. For \(L=1\), the symmetrized coefficient variance decays much more slowly with the number of qubits in both the exact and \(\varepsilon\)-approximate 2-design settings. The separation between the corresponding exact and approximate variance curves does not grow with \(n\). For \(L=2\), the same behavior is observed. The faster growth of the selected orbit produces a larger separation between the symmetrized coefficient variance and the mean coefficient variance, while the corresponding exact and approximate curves remain close as \(n\) increases. These results show how the growth of the frequency orbit can mitigate coefficient variance decay.

Figure~\ref{fig:variance_ratio_comparison} shows the corresponding variance ratio together with the selected orbit size. For \(L=1\), the ratio in the exact 2-design setting nearly coincides with the selected orbit size, numerically confirming Eq.~\eqref{eq:exp_variance} in Theorem~\ref{theorem:exp_variance}. The ratio in the \(\varepsilon\)-approximate 2-design setting also nearly coincides with the selected orbit size, indicating that the cross covariance terms are small for the ansatz considered. For \(L=2\), the ratios in both settings also remain close to the selected orbit size. Although Theorem~\ref{theorem:tight_approx_exp_variance} is restricted to \(L=1\), this observation suggests that the same variance gain may persist for practical ansatzes with \(L>1\) and continue to mitigate vanishing expressivity. A formal analysis of the \(\varepsilon\)-approximate 2-design setting for \(L>1\) remains open.

The approximate analysis above has focused on symmetrized coefficients defined over frequency orbits. The same analysis also yields a result for an individual Fourier coefficient of the underlying QFM. Taking the symmetry group to be trivial makes every orbit a singleton, \([\bm{\nu}]=\{\bm{\omega}\}\), and therefore \(a_{\bm{\nu}}=c_{\bm{\omega}}\). Specializing Theorem~\ref{theorem:tight_approx_exp_variance} to this case gives an individual coefficient bound that can be compared directly with \citet[Theorem~7]{Hela2025}. To isolate the term shared by the two bounds, define
\begin{equation}
\mathcal{P}
:=
\frac{\varepsilon \sqrt{|\mathfrak{R}(\bm{\omega})|}}{d(d^2-1)}
\Big( (2d-1)\|\hat{O}\|_F^2 - (\operatorname{Tr} \hat{O})^2 \Big).
\end{equation}
For \(L=1\), \citet[Theorem~7]{Hela2025} give the following bound under \(\varepsilon\)-approximate 2-designs with spectral norm control:
\begin{align}
\operatorname{Var}_{\bm{\theta}}(c_{\bm{\omega}})
-
\operatorname{Var}_{\text{2-design}}(c_{\bm{\omega}})
\leq \mathcal{P} + \varepsilon^2 \| \hat{O}\|_F^2 |\mathfrak{R}(\bm{\omega})|.
\label{eq:hela_single_approx_variance_bound}
\end{align}
The second term grows linearly with the frequency redundancy \(|\mathfrak{R}(\bm{\omega})|\). The following corollary states our corresponding bound for an individual Fourier coefficient.

\begin{corollary}[Individual Coefficient Variance under \(\varepsilon\)-Approximate 2-designs]
\label{corollary:single_tight_approx_exp_variance}
Consider the \(n\)-qubit single-uploading QFM in Eq.~\eqref{eq:QFM} with \(L=1\) and Fourier coefficients \(c_{\bm{\omega}}(\bm{\theta})\). Suppose the two trainable layers form independent \(\varepsilon\)-approximate 2-designs. For any nonzero frequency $\bm{\omega} \neq 0$, let $|\mathfrak{R}(\bm{\omega})|$ be the redundancy size. The deviation from the exact 2-design variance is bounded by:
\begin{align}
\left|
\operatorname{Var}_{\bm{\theta}}(c_{\bm{\omega}})
-
\operatorname{Var}_{\text{2-design}}(c_{\bm{\omega}})
\right|
\leq \mathcal{P} + (\varepsilon^2+ \varepsilon^4) \|\hat{O}\|_F^2.
\label{eq:exact_single_approx_variance_bound}
\end{align}

\end{corollary}
Since \(\bm{\omega}\neq0\), its frequency generator contains no diagonal index pair and therefore satisfies \(|\mathfrak R(\bm{\omega})|\leq d(d-1)<d^2=4^n\). Both results compare \(\operatorname{Var}_{\bm{\theta}}(c_{\bm{\omega}})\) with \(\operatorname{Var}_{\text{2-design}}(c_{\bm{\omega}})\). The result of \citet[Theorem~7]{Hela2025} gives only a one-sided upper bound on \(\operatorname{Var}_{\bm{\theta}}(c_{\bm{\omega}})-\operatorname{Var}_{\text{2-design}}(c_{\bm{\omega}})\), whereas Eq.~\eqref{eq:exact_single_approx_variance_bound} controls the absolute deviation \(\left|\operatorname{Var}_{\bm{\theta}}(c_{\bm{\omega}})-\operatorname{Var}_{\text{2-design}}(c_{\bm{\omega}})\right|\). For \(0<\varepsilon\leq1\), their right-hand sides have the following worst-case scalings:
\begin{equation}
\label{eq:single_coefficient_worst_case_scalings}
\begin{aligned}
\text{Existing bound~\cite{Hela2025}:}\quad
&\mathcal O\!\left(\varepsilon\|\hat O\|_F^2
\left(2^{-n}+\varepsilon 4^n\right)\right),\\
\text{This work:}\quad
&\mathcal O\!\left(\varepsilon\|\hat O\|_F^2
\left(2^{-n}+\varepsilon\right)\right).
\end{aligned}
\end{equation}

\begin{figure}[!t]
    \centering
    \includegraphics[width=1\linewidth]{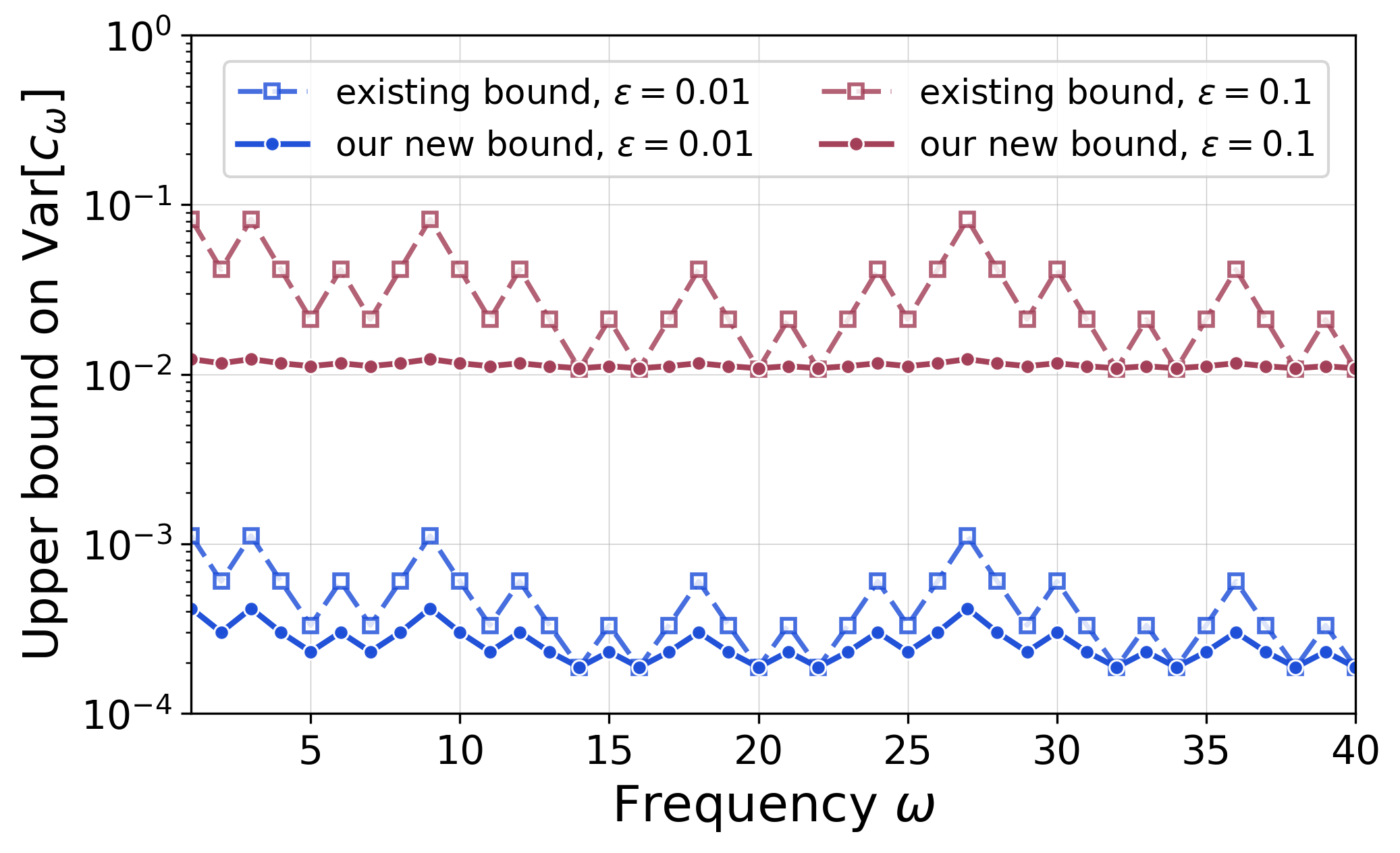}
    \caption{Comparison between the existing \(\varepsilon\)-approximate 2-design bound in Eq.~\eqref{eq:hela_single_approx_variance_bound} and the bound in Corollary~\ref{corollary:single_tight_approx_exp_variance}. Dashed curves show the existing bound, and solid curves show our bound.}
    \label{fig:operator_norm_bound_manual_comparison}
\end{figure}

To illustrate the improvement provided by our bound, Figure~\ref{fig:operator_norm_bound_manual_comparison} compares the two analytic upper bounds on \(\operatorname{Var}(c_{\omega})\) for a single-variable QFM with one encoding layer (\(L=1\)), four qubits (\(n=4\)), and Hilbert space dimension \(d=16\). The model uses the exponential encoding layer \(E(x)=\bigotimes_{i=1}^{4}R_Z(3^{i-1}x)\), the global observable \(\hat O=(\ket{0}\!\bra{0})^{\otimes4}\), and the frequency spectrum \(\Omega=\{-40,-39,\ldots,40\}\). The figure evaluates the two bounds directly rather than using circuit simulations. For this encoding, the bounds are symmetric under \(\omega\mapsto-\omega\), so only positive nonzero frequencies are shown.

The exponential separation in Eq.~\eqref{eq:single_coefficient_worst_case_scalings} is a worst-case statement rather than a uniform improvement over all frequencies. For an individual frequency, the separation between the bounds is governed by its redundancy \(|\mathfrak R(\omega)|\). Figure~\ref{fig:operator_norm_bound_manual_comparison} shows that the two bounds nearly overlap for many frequencies because these frequencies have \(|\mathfrak R(\omega)|=1\). More generally, the redundancy of an exponential encoding with \(L\) layers is \(|\mathfrak R(\omega)|=2^{nL-T(|\omega|)}\), where \(T(|\omega|)\in\{0,\ldots,nL\}\) \cite[Appendix~A.2.2]{Hela2025}. This redundancy can therefore grow exponentially with \(nL\). For a Pauli encoding with \(L\) layers, the redundancy is \(|\mathfrak R(\omega)|=\binom{2nL}{nL-|\omega|}\) for \(\omega\in\{-nL,\ldots,nL\}\) \cite[Appendix~A.2.1]{Hela2025}. Thus, our bound can be exponentially tighter for both encodings, with the improvement more pronounced for the Pauli encoding because the redundancies of most frequencies grow exponentially.

For \(L=1\), the one-sided bound of \citet[Theorem~7]{Hela2025} contains a square root term proportional to \(\sqrt{|\mathfrak{R}(\bm{\omega})|}\) and a linear term proportional to \(|\mathfrak{R}(\bm{\omega})|\). Corollary~\ref{corollary:single_tight_approx_exp_variance} strengthens the comparison to an absolute-deviation bound, retains the common square root dependence through \(\mathcal{P}\), and removes the linear redundancy factor from the \(\varepsilon^2\) contribution. This tighter redundancy scaling does not determine whether the variance deviation itself depends linearly on redundancy. Extending the analysis to approximate 2-designs with \(L>1\) remains open.

\subsection{Inductive bias for solving PDEs}
\label{sec:basis_function}

The previous subsection shows that symmetry averaging can improve coefficient variances by reorganizing Fourier coefficients into symmetrized coefficients. For solving PDEs, the same operation has a second role: it can change the effective hypothesis space and thereby impose a useful inductive bias. This is especially relevant when the target solution is expected to be well represented by a structured basis.

Chebyshev polynomials provide such a basis. They are widely used in spectral methods for solving PDEs and form an optimal set of basis functions in the sense of the uniform \(L_\infty\) norm, enabling effective approximation of smooth functions \cite{Kyriienko_2021}. A hypothesis space spanned by Chebyshev polynomials is therefore a natural target for the class of PDE problems considered here. The question is how to use quantum models to construct such a hypothesis space.

This question is directly relevant to quantum physics-informed neural networks (QPINNs), which use quantum neural networks to represent solutions of PDEs \cite{Berger2025,panichi2025,Farea2025,Setty2025,Chen2026,Lantigua2026,wang2026, Li2026}. The same basic mechanism appears in several related variational quantum approaches to differential equations \cite{Kyriienko_2021,Jaderberg_2024,Hunout_2025,Setty2025}. Many works using discrete-variable architectures adopt the data re-uploading model \cite{Kyriienko_2021,Jaderberg_2024,Hunout_2025,Setty2025,Berger2025,Farea2025}. In these models, the choice of data encoding determines the hypothesis space of the quantum model.

\begin{table*}[t]
\begin{ruledtabular}
\begin{tabular}{llllll}
\textrm{Method} &\textrm{Target}&
\textrm{Construction Cost}&
\textrm{Qubits}&
\textrm{Depth}&
\textrm{Reference}\\ \colrule

Null Space &Equivariant layer& Exponential in \(n\)& Depends& Depends& \cite{Nguyen2024}, Sec.~V.B.1\\

Choi Operator
 &Equivariant layer& Exponential in \(n\)& Depends& Depends& \cite{Nguyen2024}, Sec.~V.B.3\\

Gate Symmetrization  &Equivariant gateset& \(\mathcal O(|\mathcal S|)\)& \(n\)& Depends& \cite{Meyer2023}, Sec.~III.B\\

Deterministic Twirling &Equivariant layer& Negligible& \(n+\lceil \log_2|\mathcal{S}| \rceil\)& \(D+\mathcal{O}(|\mathcal{S}|D_R)\)& \multirow{2}{*}{\cite{Nguyen2024}, App.~D.2}\\

Stochastic Twirling
 &Equivariant layer& Randomized $R(s)$& \(n\)
& \(D+D_R\)& \\

\textbf{Randomized encoding}
 &Invariant full model& Randomized $E_s(\bm{x})$& \(n\)
& \(D\)& This work\\
\textbf{Baseline} &None& None& \(n\)& \(D\)&--\\

\end{tabular}
\end{ruledtabular}
\caption{\label{tab:resource_complexity}
Resource comparison for different approaches to constructing invariant quantum models.
We introduce the terms ``deterministic twirling'' and ``stochastic twirling'' to distinguish the two methods discussed but not named in \cite{Nguyen2024}. The baseline model is the general QFM in Eq.~\eqref{eq:QFM}, with no invariance imposed. The column \emph{Target} indicates the object on which symmetry is imposed.
The column \emph{Construction Cost} reports the offline construction cost for deterministic constructions and the randomized runtime component for stochastic implementations.
The column \emph{Qubits} reports the number of qubits needed for each implementation. The column \emph{Depth} gives the total circuit depth.
Here \(D\) is the depth of the QFM circuit when no invariance is imposed, \(|\mathcal S|\) is the order of the finite symmetry group, and \(D_R\) is the depth required to implement the induced representation \(R(s)\) defined in Eq.~\eqref{eq:eq_encoding} or Eq.~\eqref{eq:twirling}.
``Depends'' indicates that the corresponding resource requirement depends on the particular circuit implementation.
Compared with stochastic twirling, our randomized encoding avoids the additional circuit depth \(D_R\), thereby reducing the required quantum resources.
}
\end{table*}

To incorporate Chebyshev polynomials into the hypothesis space of a quantum model, \citet{Kyriienko_2021} introduced the \textit{Chebyshev feature map}. We refer to the model obtained by directly applying this feature map as the \textit{quantum Chebyshev model (QCM)}. The construction uses the identity \(T_n(x)=\cos(n\arccos x)\) for nonnegative integers \(n\), where \(T_n\) is the Chebyshev polynomial of the first kind, and applies \(\arccos\) componentwise to the spatial coordinates \(\bm{x}=(x_1,x_2,\dots,x_m)\in[-1,1]^m\) before encoding them into the circuit.

For a signed integer frequency \(\omega_j\), define \(T_{\omega_j}(x_j):=\cos(\omega_j\arccos x_j)=T_{|\omega_j|}(x_j)\) and \(S_{\omega_j}(x_j):=\sin(\omega_j\arccos x_j)\). For \(\omega_j>0\), \(S_{\omega_j}(x_j)=\sqrt{1-x_j^2}U_{\omega_j-1}(x_j)\), where \(U_{\omega_j-1}\) is the Chebyshev polynomial of the second kind. Substituting \(\arccos\bm{x}\) into Eq.~\eqref{eq:QFM} gives
\begin{equation}
\label{eq:direct_QCM}
f(\arccos\bm{x},\bm{\theta})=
\sum_{\bm{\omega}\in\Omega}c_{\bm{\omega}}(\bm{\theta})
\prod_{j=1}^m
\left(T_{\omega_j}(x_j)-iS_{\omega_j}(x_j)\right).
\end{equation}

The presence of these non-polynomial terms weakens the intended Chebyshev inductive bias and can introduce numerical instability when computing input derivatives for solving PDEs. Taking the partial derivative of the non-polynomial term \(S_{\omega_j}(x_j)\) with respect to input \(x_j\) for \(\omega_j > 0\) yields
\begin{equation}
\frac{\partial}{\partial x_j} S_{\omega_j}(x_j) = \frac{-\omega_j T_{\omega_j}(x_j)}{\sqrt{1-x_j^2}}.
\end{equation}
As the input approaches the domain boundaries \(x_j \to \pm 1\), the denominator approaches zero, causing the derivative to diverge. These divergent derivatives can destabilize the evaluation of PDE residuals near the domain boundaries.

To formalize an ideal hypothesis space free of these non-polynomial terms, we consider the \emph{truncated multivariate Chebyshev series}, defined as a finite linear combination of pure multivariate Chebyshev polynomials with nonnegative degrees \cite{Trefethen2017}:
\begin{equation}
\sum_{\boldsymbol{\nu}} a_{\boldsymbol{\nu}} \prod_{j=1}^m T_{\nu_j}(x_j),
\qquad
\nu_j\in\mathbb Z^{+}.
\end{equation}
Here \(\mathbb Z^{+}:=\{0,1,2,\ldots\}\) denotes the set of nonnegative integers. We refer to the space spanned by these truncated series as the \emph{Chebyshev space}. To systematically eliminate the non-polynomial terms from the quantum model and project its output into the Chebyshev space, let \(\mathcal S=\{\pm1\}^m\) be the elementary abelian 2-group \(\mathbb Z_2^m\) introduced in Section~\ref{sec:coefficient_analysis}. We can also call \(\{\pm1\}^m\) the sign-flip group. Let \(V:\mathcal S\to O(m)\), \(s\mapsto V_s\), be its standard representation on the input angle vector \(\boldsymbol{u}\), defined by \(V_s[\boldsymbol{u}]=(s_1 u_1,\ldots,s_m u_m)\) for \(s=(s_1,\ldots,s_m)\). The operator \(\mathcal{A}\) acting on a function \(f\) is then given by
\begin{equation}
\label{eq:sign_avg}
(\mathcal{A}f)(\boldsymbol{u},\bm{\theta}) = \frac{1}{|\mathcal S|} \sum_{s \in \mathcal S} f(V_s[\boldsymbol{u}],\bm{\theta}).
\end{equation}
Applying the operator \(\mathcal{A}\) to \(f(\boldsymbol{u},\bm{\theta})\) yields
\begin{align}
(\mathcal{A}f)(\boldsymbol{u},\bm{\theta})
&=\sum_{\boldsymbol{\omega}\in\Omega} c_{\boldsymbol{\omega}}(\bm{\theta}) \frac{1}{|\mathcal S|} \sum_{s \in \mathcal S} e^{-iV_s[\boldsymbol{u}] \cdot \boldsymbol{\omega}}\\
&=\sum_{\boldsymbol{\omega}\in\Omega} c_{\boldsymbol{\omega}}(\bm{\theta}) \prod_{j=1}^m \frac{e^{-i u_j \omega_j}+e^{i u_j \omega_j}}{2}\\
&=\sum_{\boldsymbol{\omega}\in\Omega} c_{\boldsymbol{\omega}}(\bm{\theta}) \prod_{j=1}^m \cos \left(u_j \omega_j\right).
\end{align}
Since \(u_j=\arccos x_j\), we have \(\cos(\omega_j u_j)=T_{|\omega_j|}(x_j)\). Therefore, the averaging eliminates all terms containing sine factors and leaves only products of Chebyshev polynomials. We formalize the resulting averaged model in the following proposition.

\begin{proposition}[Symmetrized Quantum Chebyshev Model]
\label{prop:quantum_chebyshev_model}
Let \(f\) be the hypothesis function defined in Eq.~\eqref{eq:QFM}, and assume that its frequency spectrum satisfies \(\Omega\subset\mathbb Z^m\). Applying the operator \(\mathcal{A}\) with the sign-flip group \(\mathcal S=\{\pm1\}^m\) and its representation \(V_s\) defined above yields a truncated multivariate Chebyshev series:
\begin{equation}
\label{eq:QCM}
(\mathcal{A}f)(\arccos \boldsymbol{x},\bm{\theta}) = \sum_{\boldsymbol{\nu}\in\Omega^{+}} a_{\boldsymbol{\nu}}(\bm{\theta}) \prod_{j=1}^m T_{\nu_j}(x_j),
\end{equation}
where \(\Omega^{+}:=\{|\boldsymbol{\omega}|:\boldsymbol{\omega}\in\Omega\}\), with \(|\boldsymbol{\omega}| = (|\omega_1|, \dots, |\omega_m|)\) denoting the componentwise absolute value. The corresponding coefficients are given by \(a_{\boldsymbol{\nu}}(\bm{\theta}):=\sum_{\{\boldsymbol{\omega}\in\Omega:\,|\boldsymbol{\omega}|=\boldsymbol{\nu}\}} c_{\boldsymbol{\omega}}(\bm{\theta})\).
\end{proposition}

The averaged function is invariant under independent sign flips of the angle variable \(\boldsymbol{u}=\arccos\boldsymbol{x}\), and is therefore even in each coordinate of \(\boldsymbol{u}\). This coordinatewise evenness removes all terms containing sine factors and gives the Chebyshev form in Eq.~\eqref{eq:QCM}.

The symmetry group can also be enlarged from \(\mathcal S\) to the hyperoctahedral group \(B_m\) while preserving this Chebyshev form. This adds coordinate permutation symmetry to the independent sign flips. We examine this extension for \(B_2\) in the QPINN experiments in Section~\ref{sec:experiment}.

\section{Randomized Encoding}
\label{sec:randomized_encoding}

The previous section presented a new Fourier perspective on QFMs by showing how invariance with respect to a symmetry group reorganizes the frequency spectrum and Fourier coefficients. However, existing methods for implementing invariant models face resource limitations. This section first examines these limitations and then introduces randomized encoding as a more resource-efficient alternative.

\subsection{Construction and resource comparison}
\label{sec:construction_resource_comparison}

Lemma~\ref{lemma:invariance} gives a standard route to invariant quantum models: one may combine an equivariant data embedding, equivariant trainable layers, an invariant initial state, and an invariant observable to obtain an invariant hypothesis function. This componentwise construction is powerful, but it also shifts the problem to the design of quantum objects compatible with the symmetry. In particular, the symmetry of the classical data must first be encoded as a group representation on the Hilbert space, and the trainable part of the model must then be restricted to maps or gates that are equivariant with respect to this group representation \cite{Meyer2023}.

Several techniques have been developed to construct such equivariant objects, as summarized in Table~\ref{tab:resource_complexity}. The null space and Choi operator methods construct equivariant layers or channels by solving the corresponding constraints induced by the group representation. These methods are systematic, but their offline cost can scale exponentially with the number of qubits. Gate symmetrization instead starts from a standard gateset and projects its generators onto the commutant of the group representation, thereby producing an equivariant gateset. This avoids solving for a full space of channels, but evaluating this projection can scale as \(\mathcal O(|\mathcal S|)\). Since \(|\mathcal S|\) often grows exponentially with system size, this projection can impose a substantial computational burden.

Another way to impose symmetry is twirling, which means averaging over the group representation. Twirling can be implemented directly at runtime rather than used only as an offline construction principle. For the resource comparison in this work, we call the two implementations at the circuit level discussed by \citet{Nguyen2024} \textit{deterministic twirling} and \textit{stochastic twirling}. These are two circuit implementations of the same twirling operation. Deterministic twirling realizes the group average using an ancilla register of size \(\lceil\log_2|\mathcal S|\rceil\) and controlled implementations of the induced representation \(R(s)\). Since each implementation of \(R(s)\) has depth \(D_R\), the controlled twirling construction can have additional depth \(\mathcal O(|\mathcal S|D_R)\). Stochastic twirling samples a group element \(s\in\mathcal S\) and implements \(R(s)\) once in each sampled circuit, so its additional depth is \(D_R\). Thus, stochastic twirling removes the ancilla register and reduces the twirling depth overhead from \(\mathcal O(|\mathcal S|D_R)\) to \(D_R\). These resource savings make stochastic twirling more practical than deterministic twirling and can be favorable on near-term devices \cite{Nguyen2024}. However, stochastic twirling still adds \(D_R\) to the baseline circuit depth \(D\). This overhead can remain substantial when the symmetry group is large or complex.

\begin{figure}[t]
    \centering
    \includegraphics[width=0.9\linewidth]{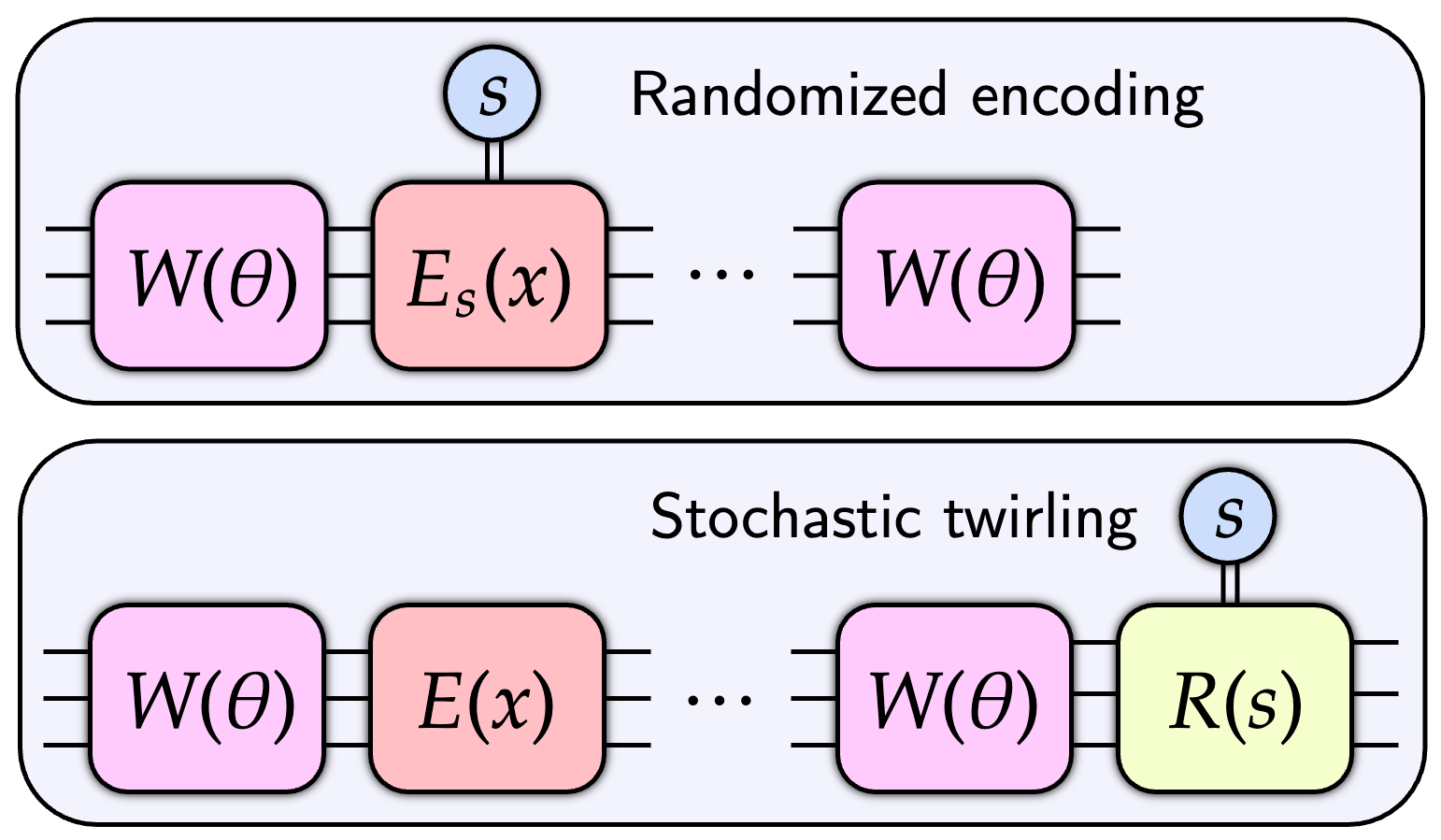}
    \caption{Circuit comparison of randomized encoding (top) and stochastic twirling (bottom) for a sampled symmetry element \(s\). Randomized encoding replaces \(E(\bm{x})\) with the transformed encoding \(E_s(\bm{x})=E(V_s[\bm{x}])\), whereas stochastic twirling keeps \(E(\bm{x})\) and implements the induced representation \(R(s)\) on the quantum device. The relation between \(E(V_s[\bm{x}])\) and \(R(s)\) is given by the equivariance condition in Eq.~\eqref{eq:eq_encoding}.}
    \label{fig:random_encoding}
\end{figure}

To overcome this difficulty, we propose \textit{randomized encoding}, which keeps the input side form of the equivariance relation. For a sampled \(s\in\mathcal S\), it applies the encoding layer
\[
E_s(\bm{x}):=E(V_s[\bm{x}]).
\]
Figure~\ref{fig:random_encoding} compares randomized encoding with stochastic twirling. Randomized encoding directly applies \(E_s(\bm{x})\), as shown in the top circuit. For stochastic twirling, Eq.~\eqref{eq:eq_encoding} gives \(E_s(\bm{x})=R(s)E(\bm{x})R(s)^\dagger\). Since the trainable layers commute with \(R(s)\), these representation operations can be reduced to \(R(s)^\dagger\) at the beginning of the circuit and \(R(s)\) at the end. Because \(R(s)^\dagger|0\rangle=|0\rangle\) by Eq.~\eqref{eq:eq_init}, the \(R(s)^\dagger\) at the beginning need not be implemented, leaving only the final \(R(s)\) shown in the bottom circuit.

For the theoretical analysis, the same circuit implementation can be viewed as first replacing \(\bm{x}\) by \(V_s[\bm{x}]\) and then applying \(E\). Because \(E_s(\bm{x})\) has the same circuit structure as \(E(\bm{x})\), randomized encoding keeps the circuit depth at \(D\) and leaves the trainable ansatz and observable unchanged. Since invariance is imposed by averaging the full predictions \(f(V_s[\bm{x}],\bm{\theta})\), randomized encoding does not require the trainable layers to commute with \(R(s)\), unlike componentwise equivariant constructions \cite{Meyer2023}.

\subsection{Sampling cost}
\label{sec:sampling_cost}

Quantum expectation values are estimated from repeated measurements. For the model considered here, repeatedly preparing \(\rho(\bm{x},\bm{\theta})\), measuring \(\hat O\), and averaging the measurement outcomes estimates \(f(\bm{x},\bm{\theta})\). Under symmetry averaging, the quantity to be estimated becomes the averaged predictor \((\mathcal A f)(\bm{x},\bm{\theta})\) defined in Eq.~\eqref{eq:group_op_def}. Throughout this section, the input \(\bm{x}\) and model parameters \(\bm{\theta}\) are fixed, and we consider ideal circuit devices. For a sampled symmetry element \(s\in\mathcal S\), let \(\rho_s(\bm{x},\bm{\theta}):=\rho(V_s[\bm{x}],\bm{\theta})\) denote the state prepared for the transformed input \(V_s[\bm{x}]\). The averaged state and averaged predictor are
\begin{align}
\bar\rho(\bm{x},\bm{\theta})
&:=
\frac{1}{|\mathcal S|}
\sum_{s\in\mathcal S}
\rho_s(\bm{x},\bm{\theta}),\\
(\mathcal A f)(\bm{x},\bm{\theta})
&=
\operatorname{Tr}\!\left(\hat O\bar\rho(\bm{x},\bm{\theta})\right).
\end{align}
We use the Pauli decomposition \(\hat O=\sum_i h_i O^{(i)}\) from Eq.~\eqref{eq:re-uploading_function}. Let \(\mathcal I=\{i:h_i\neq0\}\), and let \(I\) be the Pauli index random variable with \(\Pr[I=i]=q_i=|h_i|/\|h\|_1\). Let \(S\) be the symmetry random variable, sampled uniformly from \(\mathcal S\). In a realized shot, we denote the sampled values by lowercase \(s\) and \(i\), and the observed measurement outcome by \(z^{(i)}_{s,\bm{x}}\in\{-1,1\}\). The corresponding measurement random variable for fixed \(s\) and \(i\) is denoted by \(Z^{(i)}_{s,\bm{x}}\), with expectation
\begin{equation}
\mathbb E
\left[
Z^{(i)}_{s,\bm{x}}
\right]
=
\operatorname{Tr}\!\left(
O^{(i)}
\rho_s(\bm{x},\bm{\theta})
\right).
\end{equation}
The corresponding single shot estimator \(Y_{\bm{x}}\) is defined through its conditional distribution. For any \(s\in\mathcal S\) and \(i\in\mathcal I\),
\begin{equation}
\label{eq:randomized_encoding_single_shot}
\left(
Y_{\bm{x}}
\mid
S=s,
I=i
\right)
\overset{d}{=}
\frac{h_i}{q_i}
Z^{(i)}_{s,\bm{x}},
\end{equation}
where \(\overset{d}{=}\) denotes equality in distribution.
For a realized shot with samples \(s\), \(i\), and observed outcome \(z^{(i)}_{s,\bm{x}}\), the realized estimator value is $y_{\bm{x}}
=\frac{h_i}{q_i}z^{(i)}_{s,\bm{x}}$. The empirical estimator over \(N\) independent shots is the random variable
\begin{equation}
\label{eq:randomized_encoding_empirical_mean}
\widehat{\mathcal A f}_N(\bm{x},\bm{\theta})
:=
\frac{1}{N}
\sum_{k=1}^{N}
Y_{\bm{x}}^{(k)}.
\end{equation}
For an experimental record with realized values \(y_{\bm{x}}^{(1)},\dots,y_{\bm{x}}^{(N)}\), the numerical estimate is obtained by replacing \(Y_{\bm{x}}^{(k)}\) in Eq.~\eqref{eq:randomized_encoding_empirical_mean} with \(y_{\bm{x}}^{(k)}\).
This is the symmetry sampling estimator analyzed below.

\begin{proposition}[Symmetry sampling shot complexity]
\label{prop:randomized_encoding_shot_complexity}
The estimator in Eq.~\eqref{eq:randomized_encoding_empirical_mean} is unbiased for \((\mathcal A f)(\bm{x},\bm{\theta})\). Its single shot variance is
\begin{equation}
\label{eq:randomized_encoding_variance}
\operatorname{Var}(Y_{\bm{x}})
=
\|h\|_1^2
-
\left((\mathcal A f)(\bm{x},\bm{\theta})\right)^2.
\end{equation}
Moreover, for any \(\epsilon>0\) and \(\delta\in(0,1)\),
\begin{equation}
\label{eq:randomized_encoding_hoeffding}
\Pr\!\left[
\left|
\widehat{\mathcal A f}_N(\bm{x},\bm{\theta})
-
(\mathcal A f)(\bm{x},\bm{\theta})
\right|
\leq
\epsilon
\right]
\geq
1-\delta
\end{equation}
whenever
\begin{equation}
\label{eq:randomized_encoding_sample_bound}
N
\geq
\frac{2\|h\|_1^2}{\epsilon^2}
\log\!\left(\frac{2}{\delta}\right).
\end{equation}
\end{proposition}
The proof is given in Appendix~\ref{app:randomized_encoding_shot_complexity}. The Hoeffding bound in Eq.~\eqref{eq:randomized_encoding_sample_bound} is an absolute worst case guarantee and should not be interpreted as saying that the symmetry sampling estimator always requires the same number of shots as the baseline estimator. A more informative comparison is obtained from the variance. Since the estimator is unbiased, the central limit theorem gives the asymptotic Gaussian distribution
\begin{equation}
\widehat{\mathcal A f}_N(\bm{x},\bm{\theta})
\sim
\mathcal N\!\left(
(\mathcal A f)(\bm{x},\bm{\theta}),
\frac{\operatorname{Var}(Y_{\bm{x}})}{N}
\right)
\end{equation}
for sufficiently large \(N\). Therefore, in the asymptotic regime and at a fixed input, the number of shots required to reach a given estimation error is proportional to the single shot variance. We therefore compare \(\operatorname{Var}(Y_{\bm{x}})\) with the corresponding baseline variance.

For the comparison averaged over the input domain, we say that the representation elements \(V_s\), \(s\in\mathcal S\), are \textit{uniformity-preserving} on \(\mathcal X\) if \(V_s[\bm{x}]\in\mathcal X\) for every \(\bm{x}\in\mathcal X\) and
\begin{equation}
\int_{\mathcal X} g(V_s[\bm{x}])\,\mathrm d\bm{x}
=
\int_{\mathcal X} g(\bm{x})\,\mathrm d\bm{x}
\end{equation}
for every integrable function \(g\) and every \(s\in\mathcal S\). Let \(K_{\bm{x}}\) denote the baseline estimator without symmetry sampling. When \(V\) is the trivial representation, \(V_s[\bm{x}]=\bm{x}\) for every \(s\in\mathcal S\), so \(Y_{\bm{x}}\) reduces to \(K_{\bm{x}}\) and \((\mathcal A f)(\bm{x},\bm{\theta})=f(\bm{x},\bm{\theta})\). Equation~\eqref{eq:randomized_encoding_variance} then gives
\begin{equation}
\label{eq:baseline_single_shot_variance}
\operatorname{Var}(K_{\bm{x}})
=
\|h\|_1^2
-
f(\bm{x},\bm{\theta})^2.
\end{equation}
Together with Eq.~\eqref{eq:randomized_encoding_variance}, this gives the following variance comparison.

\begin{proposition}[Variance comparison]
\label{prop:variance_comparison}
Let \(Y_{\bm{x}}\) be the symmetry sampling estimator in Eq.~\eqref{eq:randomized_encoding_single_shot}. Then \(Y_{\bm{x}}\) and the baseline estimator \(K_{\bm{x}}\) satisfy
\begin{equation}
\label{eq:variance_difference_sampling_cost}
\operatorname{Var}(Y_{\bm{x}})
-
\operatorname{Var}(K_{\bm{x}})
=
f(\bm{x},\bm{\theta})^2
-
\left((\mathcal A f)(\bm{x},\bm{\theta})\right)^2.
\end{equation}
Consequently, for a fixed input \(\bm{x}\),
\begin{equation}
\operatorname{Var}(Y_{\bm{x}})
<
\operatorname{Var}(K_{\bm{x}})
\Longleftrightarrow
\left|
(\mathcal A f)(\bm{x},\bm{\theta})
\right|
>
\left|
f(\bm{x},\bm{\theta})
\right|.
\end{equation}
If \(\bm{x}\) is uniformly distributed over \(\mathcal X\) and the representation elements \(V_s\) are uniformity-preserving on \(\mathcal X\), then
\begin{equation}
\label{eq:average_variance_inflation}
\mathbb E_{\bm{x}}
\left[
\operatorname{Var}(Y_{\bm{x}})
\right]
\geq
\mathbb E_{\bm{x}}
\left[
\operatorname{Var}(K_{\bm{x}})
\right].
\end{equation}
\end{proposition}

The proof of Proposition~\ref{prop:variance_comparison} is given in Appendix~\ref{app:variance_comparison}. At a fixed input \(\bm{x}\), Eq.~\eqref{eq:variance_difference_sampling_cost} shows that symmetry averaging can either increase or decrease the single shot variance. If the averaged predictor has smaller magnitude than the baseline prediction, the symmetry sampling estimator has larger variance. If the averaged predictor has larger magnitude, the symmetry sampling estimator has smaller variance. The two variances are equal when the prediction magnitudes are equal. Although either ordering can occur at an individual input, the comparison after averaging over the input domain is different. When \(\bm{x}\) is uniformly distributed over \(\mathcal X\) and the representation elements \(V_s\) are uniformity-preserving, Eq.~\eqref{eq:average_variance_inflation} shows that the mean variance is no smaller for the symmetry sampling estimator than for the baseline estimator. Under these conditions, variance reductions at some inputs are therefore offset by equal or larger variance increases at other inputs.

In the asymptotic normal regime, the corresponding change in the required number of shots at a fixed input and confidence level is quantified by the variance ratio
\begin{equation}
\label{eq:shot_ratio_sampling_cost}
\frac{N_{\mathrm{sym}}}{N_{\mathrm{base}}}
\approx
\frac{
\|h\|_1^2
-
\left((\mathcal A f)(\bm{x},\bm{\theta})\right)^2
}{
\|h\|_1^2
-
f(\bm{x},\bm{\theta})^2
},
\end{equation}
when the denominator is nonzero. Here \(N_{\mathrm{sym}}\) and \(N_{\mathrm{base}}\) denote the shot counts required by the symmetry sampling and baseline estimators, respectively. This conclusion also applies to the stochastic twirling implementation in \cite{Nguyen2024} when it is used to estimate the same averaged predictor. Randomized encoding and stochastic twirling therefore require the same number of shots in this setting, while stochastic twirling retains the additional circuit depth \(D_R\). Randomized encoding offers a depth advantage, but neither implementation is guaranteed to require fewer shots than the baseline estimator.

The Hoeffding guarantee in Eqs.~\eqref{eq:randomized_encoding_hoeffding} and \eqref{eq:randomized_encoding_sample_bound} does not depend on \(|\mathcal S|\). Randomized encoding can therefore sample one transformation per shot from a very large symmetry group without enumerating its elements, and the worst case shot bound does not grow with the size of the group. More generally, randomized encoding does not require the sampled input transformations to form a group. A general set of transformations can be used instead, although the resulting averaged predictor is not guaranteed to be invariant. Such constructions could be used to estimate more complex quantities. For example, they could be combined with parameter shift rules to estimate derivatives of the averaged predictor directly. We leave this direction for future work.

\subsection{Hardware requirements and feasibility}
\label{sec:physical_implementation}

The resource advantage of randomized encoding is realized only if the sampled encoding layer can be selected and applied without becoming a new bottleneck. In the ideal implementation at each shot, each measurement shot samples \(s\in\mathcal S\) and applies the corresponding layer \(E_s(\bm{x})\). Thus the relevant hardware requirement is not the implementation of an additional quantum operation, but the ability to update the classical controls that define \(E_s\) at a rate comparable to the measurement repetition rate. If this update is slow, randomized encoding can still be implemented in batches, but the runtime advantage over conventional twirling methods may be reduced.

Platforms in which encoding parameters are controlled directly by classical signals are natural candidates for this implementation. Continuous variable photonic processors provide one example: encoding unitaries that depend on the input can be realized through phase shifts, displacements, or other electro-optic controls, so the sampled layer \(E_s(\bm{x})\) can be applied by changing the corresponding classical drive signals rather than by changing the optical circuit. This is consistent with prior photonic experiments in which the input configuration is randomized from shot to shot \cite{Marco2015}.

More generally, randomized encoding requires fast classical reconfiguration rather than a specific hardware platform. Several recent control protocols demonstrate that updates at each shot or low latency updates of quantum control parameters are becoming experimentally realistic. Examples include randomized unitary rotations for readout error mitigation \cite{Andrew2023}, reproducible phase resets for changes of the measurement direction from shot to shot \cite{Ding2024}, fast feedback calibration protocols that update control parameters after each measurement outcome \cite{Alicia2025}, and control parameter updates from shot to shot with low latency FPGA/cryo-CMOS control \cite{Miller2025}. These results provide evidence that quantum hardware can be updated or reconfigured at or near the shot rate, thereby supporting the hardware feasibility of randomized encoding.

\section{Numerical experiments}
\label{sec:experiment}
To illustrate how symmetry averaging affects quantum machine learning models, we use quantum physics-informed neural networks (QPINNs) as the experimental framework for solving PDEs. QPINNs combine the physics-informed neural network (PINN) training principle with parameterized quantum circuits. In a classical PINN, a neural network represents the unknown solution and is trained by minimizing the differential equation residual and the boundary condition error at collocation points \cite{raissi2019physics}. In a QPINN, the solution model is built from a quantum circuit expectation value, and the same residual loss is evaluated through derivatives of the circuit output. Many works on variational quantum algorithms and QPINNs have used the same idea of residual minimization to solve PDEs \cite{Kyriienko_2021,Jaderberg_2024,Hunout_2025,Setty2025,Berger2025,panichi2025,Farea2025,Chen2026,Lantigua2026,wang2026, Li2026}.

Our quantum models used in the QPINNs are built from a two-qubit QFM with two layers. Let \(\bm z=(z_1,z_2)\) denote the input after preprocessing. Specializing Eq.~\eqref{eq:reupload} to \(L=2\), the implemented QFM unitary is
\begin{equation}
\begin{aligned}
U_{\mathrm{QFM}}(\bm z,\bm\theta)
&=
W^{(3)}(\bm\theta)
E(\bm z)W^{(2)}(\bm\theta)
E(\bm z)W^{(1)}(\bm\theta),\\
E(\bm z)
&=
\bigotimes_{j=1}^{2}R_X(\beta_j z_j),
\qquad
\beta_j=3^{j-1}.
\end{aligned}
\end{equation}
Here \(E(\bm z)\) denotes the encoding layer, and each \(W^{(\ell)}\) is a trainable layer built from the Strongly Entangling Ansatz \cite{Schuld2020}. The corresponding QFM output is the expectation value of the local observable \(\hat O_{\mathrm L}\) defined in Eq.~\eqref{eq:local_observable},
\begin{equation}
f_{\mathrm{QFM}}(\bm z,\bm\theta)
=
\left\langle00\right|
U_{\mathrm{QFM}}^\dagger(\bm z,\bm\theta)
\hat O_{\mathrm L}
U_{\mathrm{QFM}}(\bm z,\bm\theta)
\left|00\right\rangle.
\end{equation}
The choice \(\bm z=\bm x\) gives the QFM input, while \(\bm z=\arccos\bm x\) gives the QCM hypothesis function defined in Eq.~\eqref{eq:direct_QCM} and introduced as the Chebyshev feature map in \cite{Kyriienko_2021}.

In our experiments, we use the 2-dimensional hyperoctahedral group \(B_2\) to represent the symmetry in the average operator \(\mathcal A\) from Eq.~\eqref{eq:group_op_def}. The group \(B_2\) consists of coordinate permutations and independent sign flips. We write an element as \(s=(\bm{\epsilon},\pi)\in B_2\) with \(|B_2|=8\), where \(\pi\in S_2\) and \(\bm{\epsilon}=(\epsilon_1,\epsilon_2)\in\{\pm1\}^2\). Here \(S_2\) denotes the symmetric group on two elements, and \(\bm{\epsilon}\) records the two sign choices. For any vector \(\bm{z}=(z_1,z_2)\), we write \(V_s[\bm{z}]=(\epsilon_1 z_{\pi^{-1}(1)},\epsilon_2 z_{\pi^{-1}(2)})\) (Eq.~\eqref{eq:B2_Vs}). For more information on the hyperoctahedral group \(B_2\), see the discussion surrounding Eq.~\eqref{eq:B2_Vs} and Eq.~\eqref{eq:orbit-stabilizer}.

The experimental choice \((\beta_1,\beta_2)=(1,3)\) gives a frequency spectrum that is not closed under the representation \(V:B_2\to O(2)\) before averaging. Applying the \(B_2\) average operator \(\mathcal A\) completes the part of the original spectrum in each frequency orbit to the full orbit and expresses the averaged model in the corresponding orbit basis. Consequently, the \(B_2\)-QCM and the QCM without symmetry averaging have different hypothesis spaces.

Table~\ref{tab:experiment_hypothesis_functions} summarizes the four quantum hypothesis functions used in the experiments. These are the QFM and QCM baselines, together with the \(B_2\)-QFM and \(B_2\)-QCM obtained by applying \(\mathcal A\) over \(B_2\). We also evaluate randomized \(B_2\)-QFM and randomized \(B_2\)-QCM models. Instead of enumerating all elements of \(B_2\), randomized encoding samples group elements and averages the circuit outputs evaluated at the corresponding transformed inputs. Because the group elements are sampled uniformly, this sample average is an unbiased estimator of the corresponding exact \(B_2\)-averaged model. The randomized models therefore reduce the number of circuit evaluations required for each symmetry average by replacing full group enumeration with a finite sampled average, at the cost of stochastic estimation error. A fully connected PINN is included as the classical baseline.


\begin{table}[t]
\scriptsize
\renewcommand{\arraystretch}{1.35}
\begin{ruledtabular}
\begin{tabular}{@{}lll@{}}
\textrm{Model} & \textrm{Hypothesis function} & \textrm{Ref.}\\ \colrule
QFM
&
\(f_{\mathrm{QFM}}(\bm{x},\bm{\theta})=\sum_{\bm{\omega}\in\Omega}c_{\bm{\omega}}(\bm{\theta})e^{-i\bm{x}\cdot\bm{\omega}}\)
&
Eq.~\eqref{eq:QFM}
\\
QCM
&
\(f_{\mathrm{QCM}}(\bm{x},\bm{\theta})=f_{\mathrm{QFM}}(\arccos\bm{x},\bm{\theta})\)
&
Eq.~\eqref{eq:direct_QCM}
\\
\(B_2\)-QFM
&
\(\begin{aligned}
f_{B_2\text{-}\mathrm{QFM}}(\bm{x},\bm{\theta})=\frac{1}{|B_2|}\sum_{s\in B_2}f_{\mathrm{QFM}}(V_s[\bm{x}],\bm{\theta})
\end{aligned}\)
&
Eq.~\eqref{eq:group_op_def}
\\
\(B_2\)-QCM
&
\(\begin{aligned}
f_{B_2\text{-}\mathrm{QCM}}(\bm{x},\bm{\theta})=\frac{1}{|B_2|}\sum_{s\in B_2}f_{\mathrm{QFM}}(V_s[\arccos\bm{x}],\bm{\theta})
\end{aligned}\)
&
Eq.~\eqref{eq:sign_avg}
\end{tabular}
\end{ruledtabular}
\caption{\label{tab:experiment_hypothesis_functions}
Quantum hypothesis functions used in the QPINN experiments. Here \(B_2\) is the two-dimensional hyperoctahedral group, and \(V_s\) is the representation element associated with \(s\in B_2\), acting by coordinate permutations and independent sign flips.}
\end{table}

In both PDE benchmarks, each training step samples \(50\) interior collocation points from \([-0.96,0.96]^2\) and \(50\) boundary collocation points from \(\partial[-1,1]^2\). We use Adam for \(300\) steps with an initial learning rate of \(10^{-1}\) and a multiplicative decay factor of \(0.99\) at each step. The randomized \(B_2\) models use \(N_s=6\) group elements per forward evaluation, sampled without replacement from \(|B_2|=8\). In each forward evaluation, the sampled subset is shared by all input points in the batch. The soft-constrained runs use boundary weight \(\gamma=10\). All derivatives required for the PDE residuals are evaluated by automatic differentiation through the simulator, and the parameter gradient norm is clipped to \(5\) during optimization. GPT-5.5 Pro was used to assist in implementing the code for the numerical experiments. The authors specified all algorithms and numerical settings, manually reviewed and tested the AI-assisted code, and verified the correctness of the reported results by comparing them with numerical reference solutions.

For each benchmark, every model and constraint setting is evaluated over ten independent runs. The loss curves report the geometric mean over these runs, and the shaded bands show one standard deviation in log space. For each model and constraint setting, the representative run shown in the solution comparison figures is selected by the lowest final loss value. Each panel title reports the relative \(L^2\) error of the selected run. The final relative \(L^2\) error statistics over all ten runs are reported in Table~\ref{tab:final_relative_l2_errors}.

\subsection{Screened Poisson equation}

\begin{figure*}[t]
    \centering
    \includegraphics[width=0.98\textwidth]{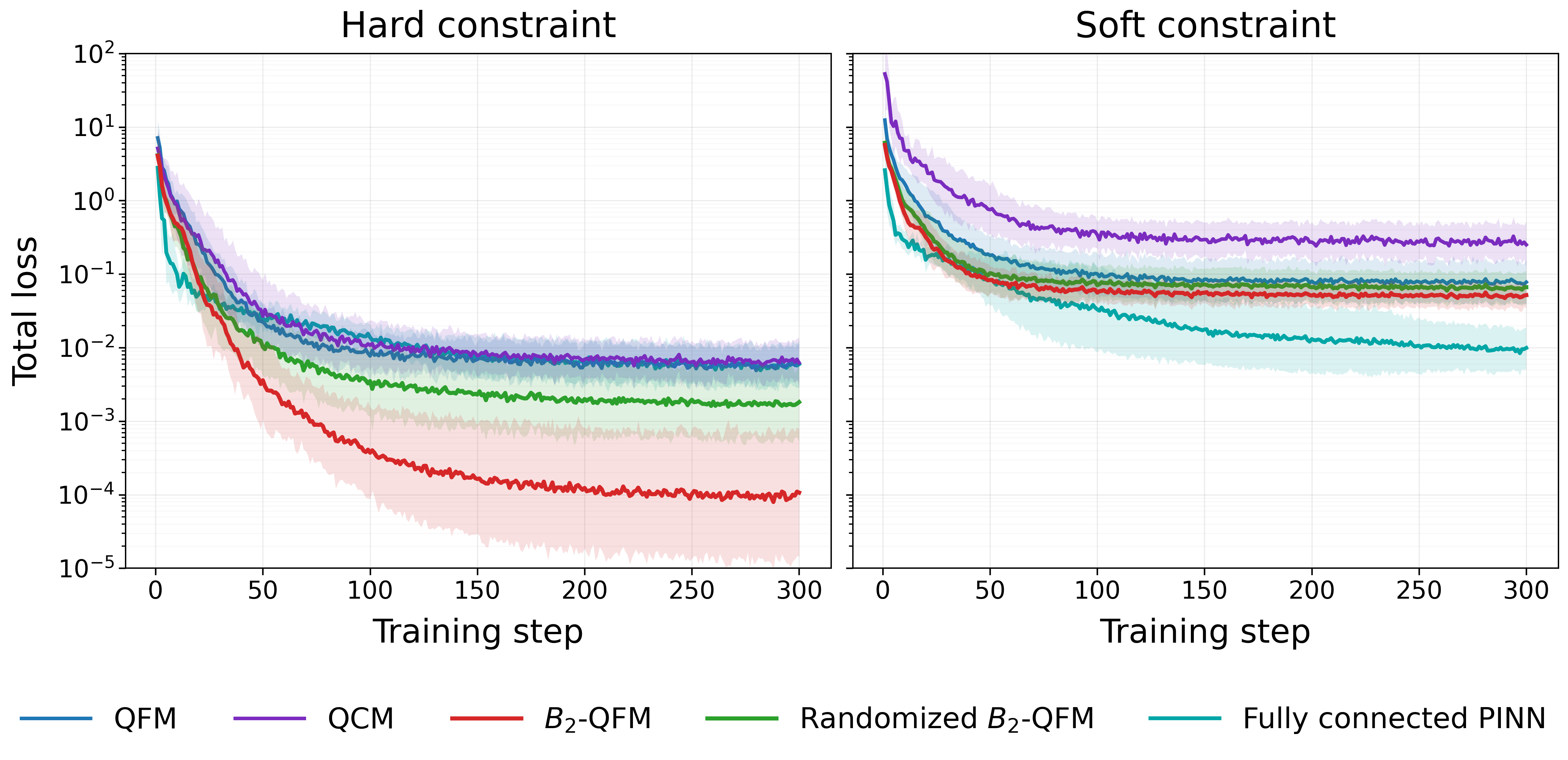}
    \caption{Training dynamics for the screened Poisson benchmark. Hard and soft constraints differ in whether the model is forced to satisfy the boundary condition exactly, leading to different definitions of the plotted loss. Curves show the geometric mean over ten independent runs, and shaded bands indicate one standard deviation in log space. Among all models, the \(B_2\)-QFM with the hard constraint has the lowest final loss.}
    \label{fig:qpinn_screened_poisson_training_losses}
\end{figure*}

\begin{figure*}[t]
    \centering
    \includegraphics[width=0.98\textwidth]{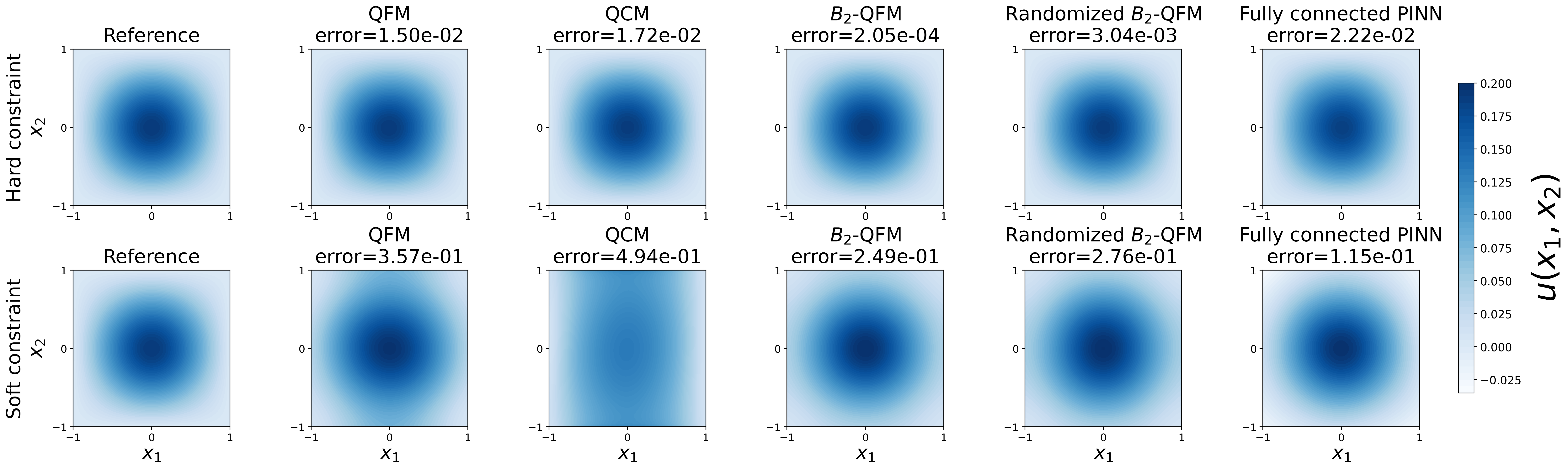}
    \caption{Solution comparison for the screened Poisson benchmark. Rows show results obtained by imposing the boundary condition as a hard or soft constraint. Columns show the analytical solution, QFM, QCM, \(B_2\)-QFM, randomized \(B_2\)-QFM, and fully connected PINN. For the randomized \(B_2\)-QFM, each forward evaluation independently samples six of the eight group elements without replacement. Model panel titles report the relative \(L^2\) error of the representative run. For each model and constraint setting, this run is selected by the lowest final loss value among ten independent runs. Among all models, the \(B_2\)-QFM with the hard constraint has the lowest error.}
    \label{fig:qpinn_screened_poisson_solutions}
\end{figure*}

We first use a two-dimensional screened Poisson equation as the QPINN benchmark. Consider the domain \(\Omega=(-1,1)^2\), with closure \(\bar\Omega=[-1,1]^2\) and boundary \(\partial\Omega=\bar\Omega\setminus\Omega\).
Fix \(\lambda>0\) and a source function
\(q:\Omega\to\mathbb R\). The target function \(u:\bar\Omega\to\mathbb R\) solves
\begin{equation}
\begin{aligned}
    \left[-\Delta+\lambda\right]u(\bm{x})
    &=
    q(\bm{x}),
    \quad
    \bm{x}\in\Omega,\\
    u(\bm{x})
    &=
    0,
    \quad
    \bm{x}\in\partial\Omega.
\end{aligned}
\end{equation}
For a model prediction \(u_{\bm{\theta}}\), the screened Poisson residual is
\begin{align}
r_{\bm{\theta}}^{\mathrm{SP}}(\bm{x})
&=
-\frac{\partial^2 u_{\bm{\theta}}}{\partial x_1^2}(\bm{x})
-\frac{\partial^2 u_{\bm{\theta}}}{\partial x_2^2}(\bm{x})
+\lambda u_{\bm{\theta}}(\bm{x})
-q(\bm{x}).
\end{align}
Let \(\bm z(\bm{x})=(z_1(x_1),z_2(x_2))\) denote the preprocessing map used before encoding. The choice \(\bm z(\bm{x})=\bm{x}\) gives the Fourier input, and \(\bm z(\bm{x})=\arccos \bm{x}\) gives the Chebyshev input. Let \(\{\bm{x}_i\}_{i=1}^{N_\Omega}\) be interior collocation points and \(\{\bm{y}_j\}_{j=1}^{N_{\partial\Omega}}\) be boundary collocation points. The conventional soft-constrained formulation uses \(u_{\bm{\theta}}(\bm{x})=v_{\bm{\theta}}(\bm z(\bm{x}))\). Here \(v_{\bm{\theta}}\) represents the quantum part of the model, excluding preprocessing and other maps, while \(u_{\bm{\theta}}\) is the complete model prediction. We denote the resulting residual by \(r_{\bm{\theta},\mathrm{soft}}^{\mathrm{SP}}\). The soft loss penalizes the boundary condition:
\begingroup
\setlength{\jot}{2pt}
\begin{align}
\mathcal L_{\mathrm{soft}}^{\mathrm{SP}}(\bm{\theta})
&=
\frac{1}{N_\Omega}
\sum_{i=1}^{N_\Omega}
\left[
r_{\bm{\theta},\mathrm{soft}}^{\mathrm{SP}}(\bm{x}_i)
\right]^2
+
\gamma
\frac{1}{N_{\partial\Omega}}
\sum_{j=1}^{N_{\partial\Omega}}
u_{\bm{\theta}}(\bm{y}_j)^2,
\notag\\
\frac{\partial \mathcal L_{\mathrm{soft}}^{\mathrm{SP}}}{\partial\theta_\ell}
&=
\frac{2}{N_\Omega}
\sum_{i=1}^{N_\Omega}
r_{\bm{\theta},\mathrm{soft}}^{\mathrm{SP}}(\bm{x}_i)
\frac{\partial r_{\bm{\theta},\mathrm{soft}}^{\mathrm{SP}}}{\partial\theta_\ell}(\bm{x}_i)
\notag\\[-0.2em]
&\quad
+
\frac{2\gamma}{N_{\partial\Omega}}
\sum_{j=1}^{N_{\partial\Omega}}
u_{\bm{\theta}}(\bm{y}_j)
\frac{\partial u_{\bm{\theta}}}{\partial\theta_\ell}(\bm{y}_j).
\label{eq:screened_poisson_soft_loss_gradient}
\end{align}
\endgroup

We can also impose the boundary condition exactly, following a similar hard-constraint idea mentioned in \cite{Kyriienko_2021}, by writing \(u_{\bm{\theta}}(\bm{x})=B(\bm{x})v_{\bm{\theta}}(\bm z(\bm{x}))\), where \(B(\bm{x})=\cos(\pi x_1/2)\cos(\pi x_2/2)\). This hard-constrained ansatz satisfies \(u_{\bm{\theta}}=0\) on \(\partial\Omega\) for every \(\bm{\theta}\). We denote the resulting residual by \(r_{\bm{\theta},\mathrm{hard}}^{\mathrm{SP}}\). The hard loss and its derivative are
\begingroup
\setlength{\jot}{2pt}
\begin{align}
\mathcal L_{\mathrm{hard}}^{\mathrm{SP}}(\bm{\theta})
&=
\frac{1}{N_\Omega}
\sum_{i=1}^{N_\Omega}
\left[
r_{\bm{\theta},\mathrm{hard}}^{\mathrm{SP}}(\bm{x}_i)
\right]^2,
\label{eq:screened_poisson_hard_loss}
\\
\frac{\partial \mathcal L_{\mathrm{hard}}^{\mathrm{SP}}}{\partial\theta_\ell}
&=
\frac{2}{N_\Omega}
\sum_{i=1}^{N_\Omega}
r_{\bm{\theta},\mathrm{hard}}^{\mathrm{SP}}(\bm{x}_i)
\frac{\partial r_{\bm{\theta},\mathrm{hard}}^{\mathrm{SP}}}{\partial\theta_\ell}(\bm{x}_i).
\label{eq:screened_poisson_hard_loss_gradient}
\end{align}
\endgroup
The displayed loss gradient expressions are not yet written in terms of quantities that can be obtained directly by running the quantum circuits used in our implementation. We can rewrite the residuals \(r_{\bm{\theta},\mathrm{soft}}^{\mathrm{SP}}(\bm{x}_i)\) and \(r_{\bm{\theta},\mathrm{hard}}^{\mathrm{SP}}(\bm{x}_i)\), together with their derivatives with respect to \(\theta_\ell\), in terms of \(v_{\bm{\theta}}\) and its required derivatives with respect to \(z_d\) and \(\theta_\ell\), all of which can be estimated directly from quantum circuit evaluations. Appendix~\ref{app:loss_gradient_expression_screen} gives the corresponding residual formulas. In this work, these derivatives are evaluated by automatic differentiation of the simulator. On hardware, the same derivative quantities could instead be estimated by differentiable circuit methods such as parameter shift rules.

For the screened Poisson numerical benchmark in Figs.~\ref{fig:qpinn_screened_poisson_training_losses} and~\ref{fig:qpinn_screened_poisson_solutions}, we use the Strongly Entangling Ansatz with two repetitions in each trainable layer, set \(\lambda=1\), let \(\phi_m(t)=\cos(m\pi t/2)\), and choose
\begin{align}
q(x_1,x_2)
&=
\phi_1(x_1)\phi_1(x_2)
\notag\\
&\quad
+
0.3\left[
\phi_1(x_1)\phi_3(x_2)
+
\phi_3(x_1)\phi_1(x_2)
\right].
\end{align}
The analytical solution used for evaluation is
\begin{align}
u_\star(x_1,x_2)
&=
\frac{\phi_1(x_1)\phi_1(x_2)}
{1+\pi^2/2}
\notag\\
&\quad
+
\frac{0.3\left[
\phi_1(x_1)\phi_3(x_2)
+
\phi_3(x_1)\phi_1(x_2)
\right]}
{1+5\pi^2/2}.
\label{eq:screened_poisson_exact_solution}
\end{align}
Each quantum model has 36 trainable circuit parameters. The fully connected PINN has layer widths \(2\!-\!4\!-\!4\!-\!1\), with two \(\tanh\)-activated hidden layers and 37 trainable parameters. Figures~\ref{fig:qpinn_screened_poisson_training_losses} and~\ref{fig:qpinn_screened_poisson_solutions} show that, under the hard constraint, the exact \(B_2\)-QFM reaches the lowest final loss among all models, followed by the randomized \(B_2\)-QFM. In the representative solution comparison, the exact \(B_2\)-QFM attains the lowest relative \(L^2\) error among all models, followed by the randomized \(B_2\)-QFM. Under the soft constraint, the fully connected PINN reaches the lowest final loss and the lowest relative \(L^2\) error in the representative solution comparison, while the exact \(B_2\)-QFM performs best among the quantum models in both comparisons. The chosen source function is \(B_2\)-invariant (with the parameter argument \(\bm{\theta}\) omitted), satisfying \(q(V_s[\bm{x}])=q(\bm{x})\) for all \(s\in B_2\) and \(\bm{x}\in\Omega\). The domain and homogeneous Dirichlet boundary condition are invariant under the same signed-permutation action, and the operator \(-\Delta+\lambda\) commutes with the \(B_2\) action. Since the Dirichlet problem has a unique solution for \(\lambda>0\), the target solution in Eq.~\eqref{eq:screened_poisson_exact_solution} is also \(B_2\)-invariant and satisfies \(u_\star(V_s[\bm{x}])=u_\star(\bm{x})\) for all \(s\in B_2\) and \(\bm{x}\in\bar\Omega\). The \(B_2\) average therefore restricts the hypothesis space to \(B_2\)-invariant functions, providing an inductive bias aligned with the target solution. The strong performance of the exact \(B_2\)-QFM is consistent with this symmetry alignment.

\subsection{Stationary viscous Hamilton-Jacobi equation}

\begin{figure*}[t]
    \centering
    \includegraphics[width=0.98\textwidth]{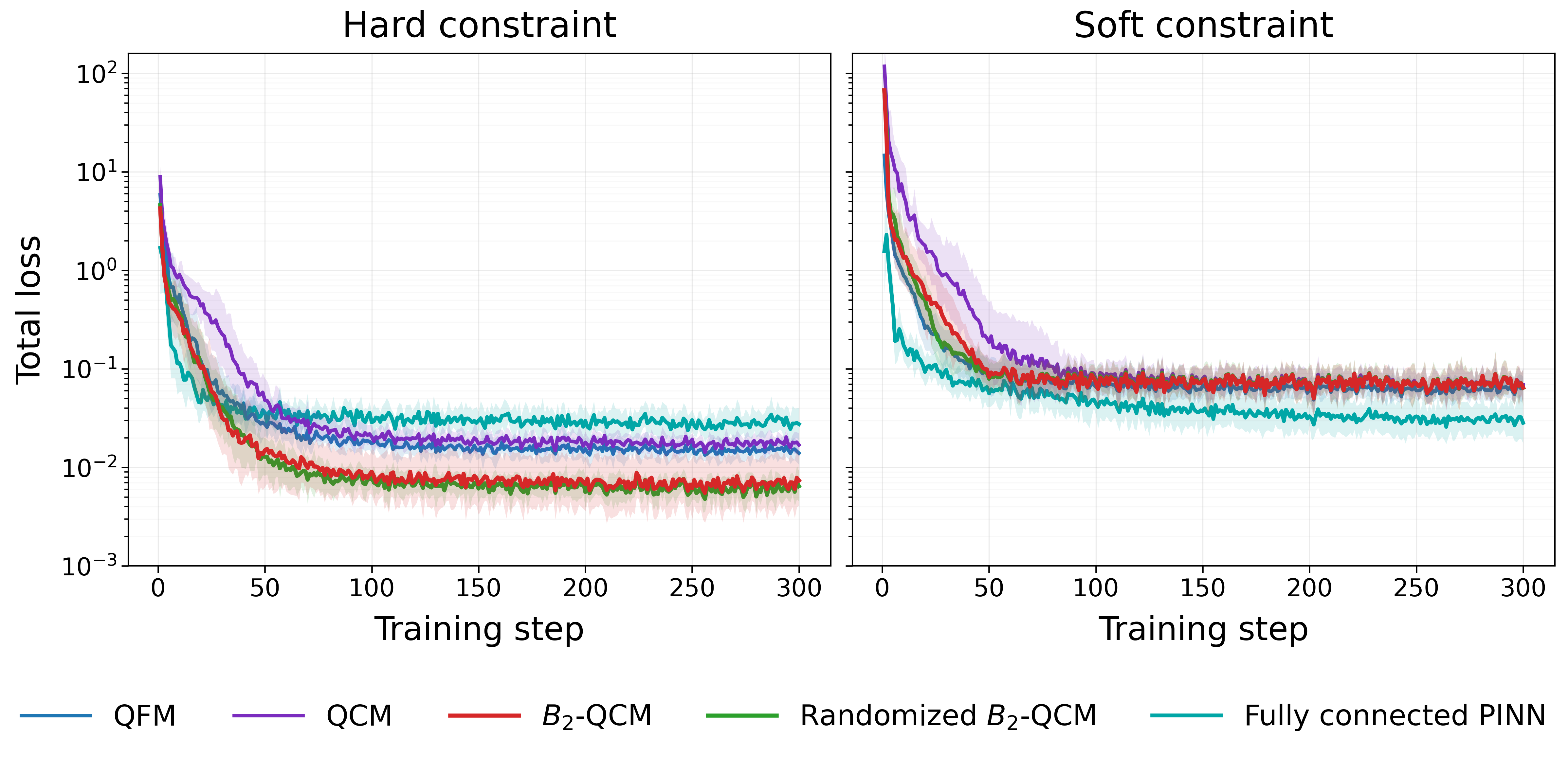}
    \caption{Training dynamics for the stationary viscous Hamilton-Jacobi benchmark. Hard and soft constraints differ in whether the model is forced to satisfy the boundary condition exactly, leading to different definitions of the plotted loss. Curves show the geometric mean over ten independent runs, and shaded bands indicate one standard deviation in log space. Among all models, the exact and randomized \(B_2\)-QCM models achieve the two lowest final losses, with very similar values.}
    \label{fig:qpinn_boundary_constraint_loss}
\end{figure*}

\begin{figure*}[t]
    \centering
    \includegraphics[width=0.98\textwidth]{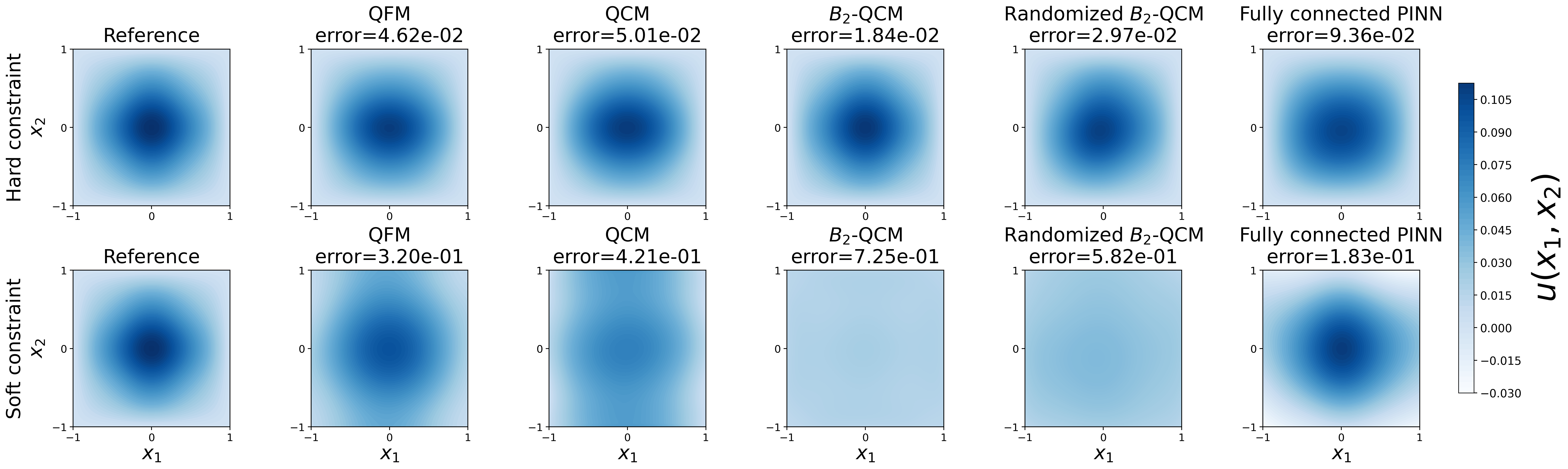}
    \caption{Solution comparison for the stationary viscous Hamilton-Jacobi benchmark. Rows show results obtained by imposing the boundary condition as a hard or soft constraint. Columns show the numerical reference solution, QFM, QCM, \(B_2\)-QCM, randomized \(B_2\)-QCM, and fully connected PINN. For the randomized \(B_2\)-QCM, each forward evaluation independently samples six of the eight group elements without replacement. Model panel titles report the relative \(L^2\) error of the representative run. For each model and constraint setting, this run is selected by the lowest final loss value among ten independent runs. Among all models, the \(B_2\)-QCM with the hard constraint has the lowest error.}
    \label{fig:qpinn_boundary_constraint_solution}
\end{figure*}

We next use a two-dimensional stationary viscous Hamilton-Jacobi equation as the second QPINN benchmark. We use \(\Omega=(-1,1)^2\), its closure \(\bar\Omega=[-1,1]^2\), its boundary \(\partial\Omega=\bar\Omega\setminus\Omega\), and \(\bm{x}=(x_1,x_2)\), as in the screened Poisson benchmark. Let \(H:\Omega\times\mathbb R^2\to\mathbb R\) be the Hamiltonian and let \(\alpha_0\in\mathbb R\) be a constant. The target function \(u:\bar\Omega\to\mathbb R\) solves
\begin{equation}
\begin{aligned}
&-\Delta u(\bm{x})
+
H(\bm{x},\nabla u(\bm{x}))
=
\alpha_0,
\quad
\bm{x}\in\Omega,
\\
&u(\bm{x})
=
0,
\quad
\bm{x}\in\partial\Omega.
\end{aligned}
\label{eq:hj_pde}
\end{equation}

Here, for \(\bm{p}\in\mathbb R^2\), the Hamiltonian is
\begin{equation}
H(\bm{x},\bm{p})
=
\rho\|\bm{p}\|_2^\beta
-
q(\bm{x}).
\end{equation}
Here \(\rho>0\), \(\beta>1\), and \(q:\Omega\to\mathbb R\) is a source function. Substituting the Hamiltonian into Eq.~\eqref{eq:hj_pde} gives
\begin{align}
&-
\left(
\frac{\partial^2 u}{\partial x_1^2}
+
\frac{\partial^2 u}{\partial x_2^2}
\right)
+
\rho
\left[
\left(
\frac{\partial u}{\partial x_1}
\right)^2
+
\left(
\frac{\partial u}{\partial x_2}
\right)^2
\right]^{\beta/2}
\notag\\[-0.2em]
&\quad
-
q(x_1,x_2)
=
\alpha_0.
\label{eq:hj_expanded}
\end{align}
For evaluation, we compute a numerical reference solution using a Newton--Krylov method. In both the training residual and the reference solver, we evaluate the term \(\rho\|\bm{p}\|_2^\beta\) numerically as \(\rho[(\|\bm{p}\|_2^2+\delta)^{\beta/2}-\delta^{\beta/2}]\), with \(\delta=10^{-12}\), to avoid numerical instability when computing derivatives near \(\bm{p}=\bm{0}\). For a model prediction \(u_{\bm{\theta}}\), the Hamilton-Jacobi residual corresponding to Eq.~\eqref{eq:hj_expanded} is
\begin{align}
r_{\bm{\theta}}^{\mathrm{HJ}}(\bm{x})
&=
-\left(
\frac{\partial^2 u_{\bm{\theta}}}{\partial x_1^2}
+
\frac{\partial^2 u_{\bm{\theta}}}{\partial x_2^2}
\right)
\notag\\[-0.2em]
&\quad
+
\rho
\left[
\left(
\frac{\partial u_{\bm{\theta}}}{\partial x_1}
\right)^2
+
\left(
\frac{\partial u_{\bm{\theta}}}{\partial x_2}
\right)^2
\right]^{\beta/2}
\notag\\[-0.2em]
&\quad
-
q(x_1,x_2)
-
\alpha_0.
\label{eq:hj_residual}
\end{align}
Using the same interior collocation points \(\{\bm{x}_i\}_{i=1}^{N_\Omega}\) and boundary collocation points \(\{\bm{y}_j\}_{j=1}^{N_{\partial\Omega}}\), the soft-constrained formulation uses \(u_{\bm{\theta}}(\bm{x})=v_{\bm{\theta}}(\bm z(\bm{x}))\). We denote the resulting residual by \(r_{\bm{\theta},\mathrm{soft}}^{\mathrm{HJ}}\). The soft loss and its derivative are
\begingroup
\setlength{\jot}{2pt}
\begin{align}
\mathcal L_{\mathrm{soft}}^{\mathrm{HJ}}(\bm{\theta})
&=
\frac{1}{N_\Omega}
\sum_{i=1}^{N_\Omega}
\left[
r_{\bm{\theta},\mathrm{soft}}^{\mathrm{HJ}}(\bm{x}_i)
\right]^2
+
\gamma
\frac{1}{N_{\partial\Omega}}
\sum_{j=1}^{N_{\partial\Omega}}
u_{\bm{\theta}}(\bm{y}_j)^2,
\notag\\
\frac{\partial \mathcal L_{\mathrm{soft}}^{\mathrm{HJ}}}{\partial\theta_\ell}
&=
\frac{2}{N_\Omega}
\sum_{i=1}^{N_\Omega}
r_{\bm{\theta},\mathrm{soft}}^{\mathrm{HJ}}(\bm{x}_i)
\frac{\partial r_{\bm{\theta},\mathrm{soft}}^{\mathrm{HJ}}}{\partial\theta_\ell}(\bm{x}_i)
\notag\\[-0.2em]
&\quad
+
\frac{2\gamma}{N_{\partial\Omega}}
\sum_{j=1}^{N_{\partial\Omega}}
u_{\bm{\theta}}(\bm{y}_j)
\frac{\partial u_{\bm{\theta}}}{\partial\theta_\ell}(\bm{y}_j).
\label{eq:hj_soft_loss_gradient}
\end{align}
\endgroup
The hard-constrained ansatz uses the same boundary factor \(B(\bm{x})\) as the screened Poisson benchmark and writes \(u_{\bm{\theta}}(\bm{x})=B(\bm{x})v_{\bm{\theta}}(\bm z(\bm{x}))\). We denote the resulting residual by \(r_{\bm{\theta},\mathrm{hard}}^{\mathrm{HJ}}\). The hard loss and its derivative are
\begingroup
\setlength{\jot}{2pt}
\begin{align}
\mathcal L_{\mathrm{hard}}^{\mathrm{HJ}}(\bm{\theta})
&=
\frac{1}{N_\Omega}
\sum_{i=1}^{N_\Omega}
\left[
r_{\bm{\theta},\mathrm{hard}}^{\mathrm{HJ}}(\bm{x}_i)
\right]^2,
\notag\\
\frac{\partial \mathcal L_{\mathrm{hard}}^{\mathrm{HJ}}}{\partial\theta_\ell}
&=
\frac{2}{N_\Omega}
\sum_{i=1}^{N_\Omega}
r_{\bm{\theta},\mathrm{hard}}^{\mathrm{HJ}}(\bm{x}_i)
\frac{\partial r_{\bm{\theta},\mathrm{hard}}^{\mathrm{HJ}}}{\partial\theta_\ell}(\bm{x}_i).
\label{eq:hj_hard_loss_gradient}
\end{align}
\endgroup
As in the screened Poisson benchmark, these residuals and their derivatives with respect to \(\theta_\ell\) can be rewritten in terms of \(v_{\bm{\theta}}\) and its derivatives, which can be estimated directly from quantum circuit evaluations. The corresponding formulas are given in Appendix~\ref{app:loss_gradient_expression_HJ}.

For the Hamilton-Jacobi numerical benchmark in Figs.~\ref{fig:qpinn_boundary_constraint_loss} and~\ref{fig:qpinn_boundary_constraint_solution}, each trainable layer uses three repetitions of the Strongly Entangling Ansatz. We multiply the quantum model output \(v_{\bm{\theta}}\) by \(0.75\) and set \(\alpha_0=0\), \(\rho=0.15\), and \(\beta=1.5\). Using \(\phi_m(t)=\cos(m\pi t/2)\), we choose
\begin{align}
q(x_1,x_2)
&=
0.5\,\phi_1(x_1)\phi_1(x_2)
\notag\\
&\quad
+
0.15\left[
\phi_1(x_1)\phi_3(x_2)
+
\phi_3(x_1)\phi_1(x_2)
\right]
\notag\\
&\quad
+
0.20\left[
\phi_1(x_1)\phi_5(x_2)
+
\phi_5(x_1)\phi_1(x_2)
\right].
\label{eq:hj_source}
\end{align}
Each quantum model has 54 trainable circuit parameters. The fully connected PINN has layer widths \(2\!-\!4\!-\!7\!-\!1\), with two \(\tanh\)-activated hidden layers and 55 trainable parameters.
Figures~\ref{fig:qpinn_boundary_constraint_loss} and~\ref{fig:qpinn_boundary_constraint_solution} show that, under the hard constraint, the exact and randomized \(B_2\)-QCM models reach similar final losses, both lower than those of the other models. In the representative solution comparison, the exact \(B_2\)-QCM attains the lowest relative \(L^2\) error among all models, followed by the randomized \(B_2\)-QCM. Under the soft constraint, the fully connected PINN reaches the lowest final loss and the lowest relative \(L^2\) error in the representative solution comparison. Because the sign flips in the \(B_2\)-QCM act on the angle variables \(\arccos\bm{x}\), they do not impose the sign flip symmetry of the physical input coordinates enforced by the \(B_2\)-QFM. The observed advantage of the \(B_2\)-QCM instead mainly comes from the additional frequencies supplied by orbit completion under symmetry averaging and the Chebyshev inductive bias.

\subsection{Discussion}

\begin{table}[t]
\scriptsize
\setlength{\tabcolsep}{1pt}
\renewcommand{\arraystretch}{1.12}
\begin{ruledtabular}
\begin{tabular}{@{}>{\raggedright\arraybackslash}p{0.22\columnwidth}>{\raggedright\arraybackslash}p{0.16\columnwidth}>{\raggedright\arraybackslash}p{0.27\columnwidth}>{\raggedright\arraybackslash}p{0.30\columnwidth}@{}}
\textrm{Benchmark} & \textrm{Constraint} & \textrm{Model} & \textrm{Relative \(L^2\) error}\\ \colrule
\multirow{10}{0.22\columnwidth}{Screened Poisson} & \multirow{5}{0.16\columnwidth}{hard} & QFM & \((3.06\pm1.76)\times10^{-2}\)\\
& & QCM & \((3.19\pm1.61)\times10^{-2}\)\\
& & \(B_2\)-QFM & \((2.59\pm2.37)\times10^{-3}\)\\
& & randomized \(B_2\)-QFM & \((1.99\pm1.35)\times10^{-2}\)\\
& & fully connected PINN & \((2.38\pm1.50)\times10^{-2}\)\\
\cline{2-4}
& \multirow{5}{0.16\columnwidth}{soft} & QFM & \((4.20\pm0.71)\times10^{-1}\)\\
& & QCM & \((8.17\pm6.82)\times10^{-1}\)\\
& & \(B_2\)-QFM & \((3.91\pm0.80)\times10^{-1}\)\\
& & randomized \(B_2\)-QFM & \((4.66\pm1.80)\times10^{-1}\)\\
& & fully connected PINN & \((1.45\pm0.40)\times10^{-1}\)\\
\colrule
\multirow{10}{0.22\columnwidth}{Stationary viscous Hamilton-Jacobi} & \multirow{5}{0.16\columnwidth}{hard} & QFM & \((5.70\pm0.99)\times10^{-2}\)\\
& & QCM & \((6.70\pm1.43)\times10^{-2}\)\\
& & \(B_2\)-QCM & \((8.22\pm14.9)\times10^{-2}\)\\
& & randomized \(B_2\)-QCM & \((3.68\pm1.01)\times10^{-2}\)\\
& & fully connected PINN & \((10.0\pm1.48)\times10^{-2}\)\\
\cline{2-4}
& \multirow{5}{0.16\columnwidth}{soft} & QFM & \((4.35\pm1.02)\times10^{-1}\)\\
& & QCM & \((5.39\pm1.06)\times10^{-1}\)\\
& & \(B_2\)-QCM & \((5.20\pm1.23)\times10^{-1}\)\\
& & randomized \(B_2\)-QCM & \((5.36\pm1.37)\times10^{-1}\)\\
& & fully connected PINN & \((2.42\pm1.39)\times10^{-1}\)
\end{tabular}
\end{ruledtabular}
\caption{\label{tab:final_relative_l2_errors}
Final relative \(L^2\) errors over ten independent runs for the screened Poisson and stationary viscous Hamilton-Jacobi benchmarks. The reported values are mean \(\pm\) standard deviation.}
\end{table}

Table~\ref{tab:final_relative_l2_errors} reports the mean final relative \(L^2\) errors over ten runs, rather than the errors of the representative runs shown in the solution figures. A clear pattern is that the hard constraint is more favorable for the quantum models than the soft constraint. Under the hard constraint, a quantum model gives the lowest mean error in each benchmark, whereas under the soft constraint, the fully connected PINN performs best. The mean errors of the quantum models are also generally close to one order of magnitude smaller under the hard constraint. The quantum models considered here may represent more complex correlations in the solution, but their hypothesis spaces may be less flexible when fitting the interior equation and the boundary condition simultaneously. By enforcing the boundary condition exactly, \(B(\bm{x})\) removes the boundary penalty and allows the quantum models to focus on the PDE residual, which may help explain their stronger performance under the hard constraint. Under the hard constraint, the large standard deviation of the exact \(B_2\)-QCM in the Hamilton-Jacobi benchmark is due to one of the ten runs, whose relative \(L^2\) error is \(5.0\times10^{-1}\); the median over the ten runs is \(3.8\times10^{-2}\).

The \(B_2\) average plays different roles in the two benchmarks. For the screened Poisson equation, it matches the coordinate exchange and independent sign flip symmetries of the analytical solution in Eq.~\eqref{eq:screened_poisson_exact_solution}. For the Hamilton-Jacobi equation, coordinate permutations still impose exchange symmetry on the physical inputs, while sign flips act on \(\arccos\bm{x}\) rather than imposing physical sign flip symmetry. The full \(B_2\) average also reorganizes the feature space through orbit completion and the Chebyshev inductive bias.

The loss gradients under the hard constraint in Eqs.~\eqref{eq:screened_poisson_hard_loss_gradient} and~\eqref{eq:hj_hard_loss_gradient} contain products between the residual and its parameter derivative. When both factors are evaluated using sampled group elements, the expected stochastic gradient can differ from the gradient of the exact \(B_2\)-averaged model. In this work, we use automatic differentiation only to simplify the numerical simulations. On quantum hardware, randomized encoding could instead be combined with the doubly stochastic gradient descent optimizers introduced in \cite{Sweke2020} to estimate the required loss derivatives. The symmetry elements and every term in the loss derivatives that can be estimated directly from quantum circuit evaluations could then be sampled jointly to construct a single stochastic estimator of the loss derivatives. We leave this fully stochastic implementation for future work.

\section{Conclusion}

This work connects invariance in GQML with the Fourier representation of QFMs to study the expressivity and inductive bias of invariant quantum models. We introduced the Fourier orbit decomposition as a new perspective on invariant QFMs, showing how invariance reorganizes the frequency spectrum into frequency orbits and the corresponding Fourier coefficients into symmetrized coefficients. Under exact 2-designs, we showed that the variance of a symmetrized coefficient equals the sum of the variances of the Fourier coefficients in its orbit. For \(\varepsilon\)-approximate 2-designs, we derived coefficient variance bounds that can be exponentially tighter than existing bounds in relevant regimes. The hyperoctahedral group provided a concrete example of the resulting orbit structure and orbit growth. We also showed that sign flip averaging yields pure multivariate Chebyshev polynomial basis functions, providing an inductive bias for solving PDEs.

To address the resource limitations of existing methods for implementing large or complex symmetry groups, we proposed \textit{randomized encoding}. It provides an unbiased estimator of the same averaged predictor by sampling transformed encoding layers without adding qubits or quantum circuit depth to the original QFM. We evaluated these models through numerical QPINN experiments on a screened Poisson equation and a stationary viscous Hamilton-Jacobi equation. The results show that the benefit of symmetry averaging depends on the relation between the symmetry, feature map, and constraint setting. More broadly, our results show how symmetry can be leveraged for expressivity of quantum models, shape their inductive bias, and guide their resource-efficient implementation.

\begin{acknowledgments}
The authors thank Xiabing Ruan (Advanced Research Institute of Mathematics, University of Strasbourg, France) for helpful discussions on group representations. This work used GPT-5.5 Pro in discussions of parts of the proofs, mainly to suggest possible ideas and identify relevant published references for further reading. All AI-assisted output was reviewed and verified by the authors. This research was supported by the French government under the France 2030 program, reference ANR-11-IDEX-0003 within the OI H-Code.
\end{acknowledgments}

\section*{Author Contributions}
L.W.: Conceptualization, Data curation, Formal analysis, Investigation,
Methodology, Software, Visualization, Writing -- original draft.
A.L., S.S., and Z.T.: Funding acquisition, Supervision,
Writing -- review \& editing. All authors reviewed and approved
the final version of the manuscript.

\section*{Data Availability}

The numerical data underlying the figures and tables, together with the
code used to run the simulations and experiments, are publicly available
in Ref.~\cite{FourierSymmetrizationCode}.

\onecolumngrid
\appendix

\section{Invariance from equivariance}
\label{sec:invariance}

\begin{restate}{Lemma}{lemma:invariance}
A data re-uploading quantum model defined in Eq.~\eqref{eq:re-uploading_function} induces an \(\mathcal{S}\)-invariant hypothesis function if it satisfies four conditions with respect to a symmetry group \(\mathcal{S}\): equivariant encoding and trainable layers, and an invariant initial state and observable.
\end{restate}

\begin{proof}
To prove that the hypothesis function \(f(\bm{x}, \bm{\theta})\) is \(\mathcal{S}\)-invariant, we must show that \(f(V_s[\bm{x}], \bm{\theta}) = f(\bm{x}, \bm{\theta})\) for all \(s \in \mathcal{S}\). We begin by formalizing the four given requirements using a unitary induced representation \(R(s)\):
\begin{enumerate}
    \item \textbf{Equivariant encoding layer:}
    \begin{equation}
    \label{eq:app_equivariant_encoding}
    E^{(l)}(V_s[\bm{x}]) = R(s) E^{(l)}(\bm{x}) R(s)^\dagger \quad \forall s \in \mathcal{S}.
    \end{equation}
    \item \textbf{Equivariant trainable layers:} The parameterized layers commute with the induced representation:
    \begin{equation}
    \label{eq:app_equivariant_trainable}
    [W^{(l)}(\bm{\theta}), R(s)] = 0 \implies W^{(l)}(\bm{\theta}) = R(s) W^{(l)}(\bm{\theta}) R(s)^\dagger.
    \end{equation}
    \item \textbf{Invariant initial state:}
    \begin{equation}
    \label{eq:app_invariant_initial_state}
    R(s)|0\rangle\langle0|R(s)^\dagger = |0\rangle\langle0|.
    \end{equation}
    \item \textbf{Invariant observable:}
    \begin{equation}
    \label{eq:app_invariant_observable}
    R(s)^\dagger \hat{O} R(s) = \hat{O}.
    \end{equation}
\end{enumerate}
First, we expand the re-uploading unitary \(U(V_s[\bm{x}], \bm{\theta})\) defined in Eq.~\eqref{eq:reupload}. Applying Eqs.~\eqref{eq:app_equivariant_encoding} and~\eqref{eq:app_equivariant_trainable}, we find that adjacent \(R(s)^\dagger\) and \(R(s)\) terms cancel, so the entire parameterized circuit transforms equivariantly:
\begin{equation}
\begin{aligned}
U(V_s[\bm{x}], \bm{\theta})
&= W^{(L+1)}(\bm{\theta})
\prod_{l=1}^{L}\bigl[E^{(l)}(V_s[\bm{x}])W^{(l)}(\bm{\theta})\bigr] \\
&= \bigl(R(s)W^{(L+1)}(\bm{\theta})R(s)^\dagger\bigr)
\prod_{l=1}^{L}
\Bigl[
\bigl(R(s)E^{(l)}(\bm{x})R(s)^\dagger\bigr)
\bigl(R(s)W^{(l)}(\bm{\theta})R(s)^\dagger\bigr)
\Bigr] \\
&= R(s)W^{(L+1)}(\bm{\theta})
\prod_{l=1}^{L}\bigl[E^{(l)}(\bm{x})W^{(l)}(\bm{\theta})\bigr]
R(s)^\dagger \\
&= R(s) U(\bm{x}, \bm{\theta}) R(s)^\dagger.
\end{aligned}
\end{equation}
Next, applying this unitary equivariance together with the invariant initial state condition in Eq.~\eqref{eq:app_invariant_initial_state} gives
\begin{equation}
\begin{aligned}
\rho(V_s[\bm{x}], \bm{\theta})
&= U(V_s[\bm{x}], \bm{\theta}) |0\rangle\langle0| U^\dagger(V_s[\bm{x}], \bm{\theta}) \\
&= R(s) U(\bm{x}, \bm{\theta}) |0\rangle\langle0| U^\dagger(\bm{x}, \bm{\theta}) R(s)^\dagger \\
&= R(s) \rho(\bm{x}, \bm{\theta}) R(s)^\dagger.
\end{aligned}
\end{equation}
Finally, using the cyclic property of the trace and the invariant observable condition in Eq.~\eqref{eq:app_invariant_observable}, we obtain
\begin{equation}
\begin{aligned}
f(V_s[\bm{x}], \bm{\theta})
&= \operatorname{Tr}\left(\hat{O}\rho(V_s[\bm{x}], \bm{\theta})\right) \\
&= \operatorname{Tr}\left(\hat{O}R(s)\rho(\bm{x}, \bm{\theta})R(s)^\dagger\right) \\
&= \operatorname{Tr}\left(R(s)^\dagger \hat{O} R(s)\rho(\bm{x}, \bm{\theta})\right) \\
&= \operatorname{Tr}\left(\hat{O}\rho(\bm{x}, \bm{\theta})\right) \\
&= f(\bm{x}, \bm{\theta}).
\end{aligned}
\end{equation}
This concludes the proof that the model induces an \(\mathcal{S}\)-invariant hypothesis function.
\end{proof}

\section{Fourier representation and orbit structure}
\label{app:fourier_representation}

\subsection{Fourier representation of re-uploading models}

For this derivation, we consider an input vector \(\bm{x}=(x_1,\dots,x_m)\) and \(L\) encoding layers each of the form
\[
E^{(l)}(\bm{x})
\;=\;
\bigotimes_{i=1}^m e^{-i\,x_i\,H_l^i},
\qquad l=1,\dots,L,
\]
where \(H_l^i\) acts on a \(d_i\)-dimensional factor, and the total Hilbert space has dimension \(d=\prod_{i=1}^m d_i\). For each \(i\), \(H_l^i\) has eigenvectors \(\{\lvert j_i^l\rangle\}_{j_i^l=1}^{d_i}\) satisfying \(\;H_l^i\lvert j_i^l\rangle=\lambda_{j_i^l}\lvert j_i^l\rangle\), and we write
\[
j^l=(j^l_1,\dots,j^l_m)\in\prod_{i=1}^m[1,d_i],
\qquad
\left|\prod_{i=1}^m[1,d_i]\right|=d,\qquad
|j^l\rangle=|j^l_1\rangle\otimes\dots\otimes|j^l_m\rangle .
\]
Then
\[
E^{(l)}(\bm{x})\lvert j^l\rangle=\bigotimes_{i=1}^m e^{-i\,x_i\,H_l^i}|j^l_i\rangle
=\bigotimes_{i=1}^m e^{-i\,x_i\,\lambda_{j_i^l}}|j^l_i\rangle
=e^{-i\sum_{i=1}^m x_i\,\lambda_{j_i^l}}\lvert j^l\rangle.
\]
Using the identity $I=\sum_{j^1\in\prod_{i=1}^m[1,d_i]}\lvert j^1\rangle\langle j^1\rvert$ and writing $\lvert0\rangle=\lvert j^0\rangle$, we insert the identity before \(W^{(1)}\):
\[
W^{(1)}\lvert0\rangle
=I\,W^{(1)}\lvert0\rangle
=\sum_{j^1\in\prod_{i=1}^m[1,d_i]}\lvert j^1\rangle\langle j^1\rvert W^{(1)}\lvert j^0\rangle.
\]
Define $W^{(1)}_{\,j^1,j^0}
=\langle j^1\rvert W^{(1)}\lvert j^0\rangle$, here the multi-index \(j^1=(j_1^1,\ldots,j_m^1)\) is treated as a single row index of \(W^{(1)}\). So that
\[
W^{(1)}\lvert0\rangle
=\sum_{j^1\in\prod_{i=1}^m[1,d_i]}W^{(1)}_{\,j^1,j^0}\,\lvert j^1\rangle.
\]
Applying the first encoding layer gives
\[
E^{(1)}(\bm{x})\,W^{(1)}\lvert0\rangle
=\sum_{j^1\in\prod_{i=1}^m[1,d_i]}W^{(1)}_{\,j^1,j^0}
\,e^{-i\sum_{i=1}^m x_i\,\lambda_{j^1_i}}
\lvert j^1\rangle.
\]
Applying the second encoding layer gives
\[
\begin{aligned}
&E^{(2)}(\bm{x})\,W^{(2)}\,E^{(1)}(\bm{x})\,W^{(1)}\lvert0\rangle\\
&= E^{(2)}(\bm{x})
  \Bigl(\sum_{j^2\in\prod_{i=1}^m[1,d_i]}\ket{j^2}\bra{j^2}\Bigr) W^{(2)}
  \sum_{j^1\in\prod_{i=1}^m[1,d_i]}
  W^{(1)}_{\,j^1,j^0}\,
  e^{-i\sum_{i=1}^m x_i\,\lambda_{j^1_i}}
  \,\ket{j^1}.\\
&=\sum_{j^2,j^1\in\prod_{i=1}^m[1,d_i]}
\bra{j^2}W^{(2)}\ket{j^1}\,
  W^{(1)}_{\,j^1,j^0}\,
  e^{-i\sum_{i=1}^m x_i\,\lambda_{j^1_i}}
  \;E^{(2)}(\bm{x})\ket{j^2}\\
&=\sum_{j^2,j^1\in\prod_{i=1}^m[1,d_i]}
  W^{(2)}_{\,j^2,j^1}\,
  W^{(1)}_{\,j^1,j^0}\,
  e^{-i\sum_{i=1}^m x_i\,\lambda_{j^1_i}}\,
  e^{-i\sum_{i=1}^m x_i\,\lambda_{j^2_i}}\,
  \ket{j^2}
\end{aligned}
\]
Continuing in this way, after \(L\) layers we obtain
\[
\prod_{l=1}^L E^{(l)}(\bm{x})\,W^{(l)}\lvert0\rangle
=\sum_{j^L,\dots,j^1\in\prod_{i=1}^m[1,d_i]}
\biggl[\prod_{l=1}^L W^{(l)}_{\,j^l,j^{l-1}}
e^{-i\sum_{i=1}^m x_i\,\lambda_{j_i^l}}\biggr]
\lvert j^L\rangle.
\]
Define the total path \(J=(j^1,\dots,j^L)\in(\prod_{i=1}^m[1,d_i])^L\). For each \(l\), \(|j^l\rangle\) is an eigenvector of \(\bigotimes_{i=1}^m H_l^i\), and we set
\[
\boldsymbol\Lambda_J
=\bigl(\Lambda_J^1,\dots,\Lambda_J^m\bigr),
\quad
\Lambda_J^i=\sum_{l=1}^L\lambda_{j_i^l}.
\]
Then
\[
\begin{aligned}
&\sum_{j^L,\dots,j^1\in\prod_{i=1}^m[1,d_i]}
\biggl[(\prod_{l=1}^L W^{(l)}_{\,j^l,j^{l-1}})
e^{-i\sum_{i=1}^m x_i\,\lambda_{j^1_i}-\ldots-i\sum_{i=1}^m x_i\,\lambda_{j_i^L}}
\biggr]
\lvert j^L\rangle\\
&=\sum_{j^L,\dots,j^1\in\prod_{i=1}^m[1,d_i]}
\biggl[(\prod_{l=1}^L W^{(l)}_{\,j^l,j^{l-1}})
e^{-i\sum_{i=1}^m x_i(\lambda_{j^1_i}+\ldots+\lambda_{j_i^L})}
\biggr]
\lvert j^L\rangle\\
&=\sum_{j^L,\dots,j^1\in\prod_{i=1}^m[1,d_i]}
\biggl[(\prod_{l=1}^L W^{(l)}_{\,j^l,j^{l-1}})
e^{-i\sum_{i=1}^m x_i\Lambda_J^i}
\biggr]
\lvert j^L\rangle\\
\end{aligned}
\]
So
\[
\prod_{l=1}^L E^{(l)}(\bm{x})\,W^{(l)}\lvert0\rangle
=\sum_{j^L,\dots,j^1\in\prod_{i=1}^m[1,d_i]}
\Bigl(\prod_{l=1}^L W^{(l)}_{\,j^l,j^{l-1}}\Bigr)\,
e^{-i\,\bm{x}\cdot\boldsymbol\Lambda_J}\,
\lvert j^L\rangle.
\]
After \(L\) layers the state is
\[
\ket{\psi_L(\bm{x})}
=\sum_{J}
\Bigl(\prod_{l=1}^L W^{(l)}_{\,j^l,j^{l-1}}\Bigr)
e^{-i\,\bm{x}\cdot\boldsymbol\Lambda_J}
\;\ket{j^L}.
\]
Applying the final trainable layer \(W^{(L+1)}\) gives
\[
U(\bm{x},\bm{\theta})\ket{0}
=W^{(L+1)}\ket{\psi_L(\bm{x})}
=\sum_{J}
\Bigl(\prod_{l=1}^L W^{(l)}_{\,j^l,j^{l-1}}\Bigr)
e^{-i\,\bm{x}\cdot\boldsymbol\Lambda_J}
\,W^{(L+1)}\ket{j^L}.
\]
Insert the identity \(\sum_{k=1}^d\ket{k}\bra{k}\), with \(k\) enumerating the \(d\)-dimensional product basis, to expand \(W^{(L+1)}\ket{j^L}\):
\[
W^{(L+1)}\ket{j^L}
=\sum_{k=1}^d\ket{k}\bra{k}W^{(L+1)}\ket{j^L}
=\sum_{k=1}^d W^{(L+1)}_{\,k,j^L}\,\ket{k}.
\]
Hence
\[
U(\bm{x},\bm{\theta})\ket{0}
=\sum_{J}\sum_{k=1}^d
W^{(L+1)}_{\,k,j^L}
\Bigl(\prod_{l=1}^L W^{(l)}_{\,j^l,j^{l-1}}\Bigr)
e^{-i\,\bm{x}\cdot\boldsymbol\Lambda_J}
\;\ket{k}
\]
For readability, we suppress the dependence of the trainable-layer matrix elements on \(\bm{\theta}\) in the following expressions. Since $f(\bm{x},\bm{\theta})=\langle0|U(\bm{x},\bm{\theta})^\dagger O U(\bm{x},\bm{\theta})\ket{0}$, we have
\[
f(\bm{x},\bm{\theta})
=\sum_{k',k}\sum_{J,J'}
\bigl(W^{(L+1)}_{\,k',j'^L}\cdots W^{(1)}_{\,j'^1,j^0}\bigr)^*
\,O_{k'k}\,
\bigl(W^{(L+1)}_{\,k,j^L}\cdots W^{(1)}_{\,j^1,j^0}\bigr)
\,e^{-i\,\bm{x}\cdot(\boldsymbol\Lambda_J-\boldsymbol\Lambda_{J'})}.
\]
Grouping terms with the same \(\boldsymbol\omega=\boldsymbol\Lambda_J-\boldsymbol\Lambda_{J'}\) defines
\[
c_{\boldsymbol\omega}(\bm{\theta})
=\sum_{(J,J'):\,\boldsymbol\Lambda_J-\boldsymbol\Lambda_{J'}=\boldsymbol\omega}
\sum_{k',k}
\bigl(W^{(L+1)}_{\,k',j'^L}\cdots W^{(1)}_{\,j'^1,j^0}\bigr)^*
\,O_{k'k}\,
\bigl(W^{(L+1)}_{\,k,j^L}\cdots W^{(1)}_{\,j^1,j^0}\bigr),
\]
which can be rewritten as
\begin{equation}
\label{eq:QFM_coefficient}
c_{\bm{\omega}}(\bm{\theta})=\sum_{\substack{\left(J, J^{\prime}\right): \\ \bm{\Lambda}_J-\bm{\Lambda}_{J^{\prime}}=\bm{\omega}}} \sum_{j^{L+1}, j^{\prime L+1}} O_{j^{\prime L+1} j^{L+1}} \prod_{l=1}^{L+1}\left(W_{j^{\prime l}, j^{\prime l-1}}^{(l)}\right)^* W_{j^l, j^{l-1}}^{(l)}
\end{equation}
and the expectation becomes the multivariate Fourier series:
\[
f(\bm{x},\bm{\theta})
=\sum_{\boldsymbol\omega\in\Omega}
c_{\boldsymbol\omega}(\bm{\theta})\,
e^{-i\,\bm{x}\cdot\boldsymbol\omega},\quad
\Omega=\{\boldsymbol\Lambda_J-\boldsymbol\Lambda_{J'}\}.
\]

\subsection{Redundancy preservation under the group representation}
\label{app:redundancy_preservation}

\begin{restate}{Proposition}{prop:redundancy_preservation_representation_elements}
With the notation of the QFM in Eq.~\eqref{eq:reupload} and the frequency generators in Definition~\ref{def:frequency_generator}, let \(V:\mathcal S\to O(m)\), \(s\mapsto V_s\), be the group representation and fix \(s\in\mathcal S\). Suppose that for every \(l\in[1,L]\), there exists a bijection \(\tau_{s,l}:[1,d]\to[1,d]\) such that
\begin{equation}
\bm{\lambda}^{(l)}_{\tau_{s,l}(j)}
=
V_s[\bm{\lambda}^{(l)}_j],
\qquad
j=1,\ldots,d.
\end{equation}
Here \(j\) is the flattened label of the multi-index \(j^l\) defined above.
Then, for all \(\bm{\omega}\in\Omega\),
\begin{equation}
\left|\mathfrak R(V_s[\bm{\omega}])\right|
=
\left|\mathfrak R(\bm{\omega})\right|,
\qquad
\bm{\omega}\in\Omega.
\end{equation}
\end{restate}

\begin{proof}
For \((J,J')\in[1,d]^{2L}\), define
\begin{equation}
\Delta(J,J')
:=
\sum_{l=1}^{L}
\left(
\bm{\lambda}^{(l)}_{j_l}
-
\bm{\lambda}^{(l)}_{j'_l}
\right).
\end{equation}
Then
\begin{equation}
\mathfrak R(\bm{\omega})
=
\left\{(J,J')\in[1,d]^{2L}:\Delta(J,J')=\bm{\omega}\right\}.
\end{equation}
For \(J=(j_1,\ldots,j_L)\), define
\begin{equation}
\tau_s(J)
:=
\left(\tau_{s,1}(j_1),\ldots,\tau_{s,L}(j_L)\right),
\qquad
\Phi_s(J,J')
:=
\left(\tau_s(J),\tau_s(J')\right).
\end{equation}
We use the following bijection property: if \(f:A\to A\) and \(g:A\to A\) satisfy
\begin{equation}
g(f(x))=x,
\qquad
f(g(y))=y,
\qquad
x,y\in A,
\end{equation}
then \(f\) is a bijection, because \(g(f(x))=x\) guarantees injectivity (since \(f(x_1)=f(x_2) \implies x_1=x_2\)), and \(f(g(y))=y\) guarantees surjectivity (since for any \(y\), \(x=g(y)\) satisfies \(f(x)=y\)). Below, we verify the two identities with two arbitrary elements \(J,K\in[1,d]^L\) to keep the forward and inverse substitutions visually distinct. Since each \(\tau_{s,l}\) is a bijection on \([1,d]\), each inverse \(\tau_{s,l}^{-1}\) exists. For \(K=(k_1,\ldots,k_L)\), define
\begin{equation}
\tau_s^{-1}(K)
=
\left(\tau_{s,1}^{-1}(k_1),\ldots,\tau_{s,L}^{-1}(k_L)\right),
\end{equation}
which satisfies
\begin{equation}
\begin{aligned}
\tau_s^{-1}(\tau_s(J))
&=
\left(
\tau_{s,1}^{-1}(\tau_{s,1}(j_1)),
\ldots,
\tau_{s,L}^{-1}(\tau_{s,L}(j_L))
\right)
=J,\\
\tau_s(\tau_s^{-1}(K))
&=
\left(
\tau_{s,1}(\tau_{s,1}^{-1}(k_1)),
\ldots,
\tau_{s,L}(\tau_{s,L}^{-1}(k_L))
\right)
=K.
\end{aligned}
\end{equation}
By the bijection property with \(A=[1,d]^L\), \(\tau_s\) is a bijection on \([1,d]^L\). Define
\begin{equation}
\Phi_s^{-1}(K,K')
=
\left(\tau_s^{-1}(K),\tau_s^{-1}(K')\right).
\end{equation}
Then
\begin{equation}
\begin{aligned}
\Phi_s^{-1}(\Phi_s(J,J'))
&=
\left(
\tau_s^{-1}(\tau_s(J)),
\tau_s^{-1}(\tau_s(J'))
\right)
=(J,J'),\\
\Phi_s(\Phi_s^{-1}(K,K'))
&=
\left(
\tau_s(\tau_s^{-1}(K)),
\tau_s(\tau_s^{-1}(K'))
\right)
=(K,K').
\end{aligned}
\end{equation}
By the same criterion with \(A=[1,d]^{2L}\), \(\Phi_s\) is a bijection on \([1,d]^{2L}\). For every \((J,J')\in[1,d]^{2L}\),
\begin{equation}
\begin{aligned}
\Delta(\Phi_s(J,J'))
&=
\sum_{l=1}^{L}
\left(
\bm{\lambda}^{(l)}_{\tau_{s,l}(j_l)}
-
\bm{\lambda}^{(l)}_{\tau_{s,l}(j'_l)}
\right)\\
&=
\sum_{l=1}^{L}
\left(
V_s[\bm{\lambda}^{(l)}_{j_l}]
-
V_s[\bm{\lambda}^{(l)}_{j'_l}]
\right)\\
&=
V_s\left[
\sum_{l=1}^{L}
\left(
\bm{\lambda}^{(l)}_{j_l}
-
\bm{\lambda}^{(l)}_{j'_l}
\right)
\right]
=
V_s[\Delta(J,J')].
\end{aligned}
\end{equation}
Because \(\Phi_s\) is a bijection on \([1,d]^{2L}\), we can change variables by \((K,K')=\Phi_s(J,J')\). Using the invertibility of \(V_s\), we find that the image of \(\mathfrak R(\bm{\omega})\) under \(\Phi_s\) is
\begin{equation}
\begin{aligned}
\Phi_s\!\left(\mathfrak R(\bm{\omega})\right)
&=
\left\{
\Phi_s(J,J')\in [1,d]^{2L}
\;\middle|\;
\Delta(J,J')=\bm{\omega}
\right\}\\
&=
\left\{
(K,K')\in[1,d]^{2L}
\;\middle|\;
\Delta(\Phi_s^{-1}(K,K'))=\bm{\omega}
\right\}\\
&=
\left\{
(K,K')\in[1,d]^{2L}
\;\middle|\;
V_s[\Delta(\Phi_s^{-1}(K,K'))]=V_s[\bm{\omega}]
\right\}\\
&=
\left\{
(K,K')\in[1,d]^{2L}
\;\middle|\;
\Delta(K,K')=V_s[\bm{\omega}]
\right\}\\
&=
\mathfrak R(V_s[\bm{\omega}]).
\end{aligned}
\end{equation}
Thus \(\Phi_s\) restricts to a bijection from \(\mathfrak R(\bm{\omega})\) to \(\mathfrak R(V_s[\bm{\omega}])\). Taking cardinalities,
\begin{equation}
\left|\mathfrak R(\bm{\omega})\right|
=
\left|\mathfrak R(V_s[\bm{\omega}])\right|.
\end{equation}

\end{proof}

\subsection{Fourier orbit decomposition}
\label{app:FOD}

\begin{restate}{Proposition}{prop:FOD}
Let \(\mathcal S\) be a finite group, \(V:\mathcal S\to O(m)\) an orthogonal representation, \(\Omega\) a finite frequency spectrum, and \(\tilde{\Omega}\) the orbit frequency spectrum. For \(f(\bm{x},\bm{\theta})=\sum_{\bm{\omega}\in\Omega}c_{\bm{\omega}}(\bm{\theta})e^{-i\,\bm{x}\cdot\bm{\omega}}\) defined in Eq.~\eqref{eq:QFM}, set \(c_{\bm{\omega}}(\bm{\theta})=0\) for \(\bm{\omega}\notin\Omega\). Then the averaged function \(\mathcal{A}f\) defined by Eq.~\eqref{eq:group_op_def} is \(\mathcal{S}\)-invariant and satisfies
\begin{equation}
\label{eq:app_FOD}
(\mathcal{A}f)(\bm{x}, \bm{\theta}) = \sum_{\bm{\nu} \in \tilde{\Omega}} a_{\bm{\nu}}(\bm{\theta}) \phi_{[\bm{\nu}]}(\bm{x}),
\end{equation}
where the symmetrized coefficient is
\begin{equation}
\label{eq:app_symmetrized_coefficient}
a_{\bm{\nu}}(\bm{\theta})
:=
\sum_{\bm{\omega}\in[\bm{\nu}]}
c_{\bm{\omega}}(\bm{\theta}),
\end{equation}
and the normalized orbit basis is
\begin{equation}
\phi_{[\bm{\nu}]}(\bm{x})
=
\frac{1}{|[\bm{\nu}]|}
\sum_{\tilde{\bm{\omega}}\in[\bm{\nu}]}
e^{-i\bm{x}\cdot\tilde{\bm{\omega}}}.
\end{equation}
\end{restate}

\begin{proof}
The \(\mathcal S\)-invariance of \(\mathcal A f\) follows from the group average in Eq.~\eqref{eq:group_op_def}. It remains to derive the orbit decomposition. Since \(V_s\in O(m)\) and \(V\) is a group representation,
\begin{align}
V_s[\bm{x}]\cdot\bm{\omega}
&=
\bigl(V_s[\bm{x}]\bigr)^{\top}\bm{\omega}
\nonumber\\
&=
\bm{x}^{\top}V_s^{\top}\bm{\omega}
\nonumber\\
&=
\bm{x}^{\top}V_s^{-1}\bm{\omega}
\nonumber\\
&=
\bm{x}^{\top}V_{s^{-1}}\bm{\omega}
\nonumber\\
&=
\bm{x}\cdot V_{s^{-1}}[\bm{\omega}].
\end{align}
Substituting the Fourier expansion into the average operator gives
\begin{align}
(\mathcal A f)(\bm{x},\bm{\theta})
&=
\frac{1}{|\mathcal S|}
\sum_{s\in\mathcal S}
f(V_s[\bm{x}],\bm{\theta})
\nonumber\\
&=
\frac{1}{|\mathcal S|}
\sum_{s\in\mathcal S}
\sum_{\bm{\omega}\in\Omega}
c_{\bm{\omega}}(\bm{\theta})
e^{-iV_s[\bm{x}]\cdot\bm{\omega}}
\nonumber\\
&=
\sum_{\bm{\omega}\in\Omega}
c_{\bm{\omega}}(\bm{\theta})
\left(
\frac{1}{|\mathcal S|}
\sum_{s\in\mathcal S}
e^{-i\bm{x}\cdot V_{s^{-1}}[\bm{\omega}]}
\right).
\label{eq:app_average_expansion}
\end{align}
Fix \(\bm{\omega}\in\Omega\), with orbit
\begin{equation}
[\bm{\omega}]
=
\{V_s[\bm{\omega}]:s\in\mathcal S\}.
\end{equation}
Define
\begin{equation}
\operatorname{Stab}(\bm{\omega})
=
\{s\in\mathcal S:V_s[\bm{\omega}]=\bm{\omega}\}.
\end{equation}
Applying the orbit-stabilizer theorem to \(\bm{\omega}\) gives
\begin{equation}
|[\bm{\omega}]|
=
\frac{|\mathcal S|}{|\operatorname{Stab}(\bm{\omega})|}.
\end{equation}
Therefore
\begin{equation}
|\operatorname{Stab}(\bm{\omega})|
=
\frac{|\mathcal S|}{|[\bm{\omega}]|}.
\label{eq:app_stabilizer_size}
\end{equation}
For \(\tilde{\bm{\omega}}\in[\bm{\omega}]\), choose \(r\in\mathcal S\) such that \(V_r[\bm{\omega}]=\tilde{\bm{\omega}}\). Then, using Eq.~\eqref{eq:app_stabilizer_size},
\begin{align}
\left|\{s\in\mathcal S:V_s[\bm{\omega}]=\tilde{\bm{\omega}}\}\right|
&=
\left|\{r^{-1}s:s\in\mathcal S,\ V_s[\bm{\omega}]=\tilde{\bm{\omega}}\}\right|
\nonumber\\
&=
\left|\{r^{-1}s:s\in\mathcal S,\ V_s[\bm{\omega}]=V_r[\bm{\omega}]\}\right|
\nonumber\\
&=
\left|\{r^{-1}s:s\in\mathcal S,\ V_{r^{-1}s}[\bm{\omega}]=\bm{\omega}\}\right|
\nonumber\\
&=
|\operatorname{Stab}(\bm{\omega})|
=
\frac{|\mathcal S|}{|[\bm{\omega}]|}.
\label{eq:app_preimage_count}
\end{align}
\begin{align}
\frac{1}{|\mathcal S|}
\sum_{s\in\mathcal S}
e^{-i\bm{x}\cdot V_{s^{-1}}[\bm{\omega}]}
&=
\frac{1}{|\mathcal S|}
\sum_{s\in\mathcal S}
e^{-i\bm{x}\cdot V_s[\bm{\omega}]}.
\label{eq:app_inverse_average}
\end{align}
\begin{equation}
\mathcal S
=
\bigsqcup_{\tilde{\bm{\omega}}\in[\bm{\omega}]}
\{s\in\mathcal S:V_s[\bm{\omega}]=\tilde{\bm{\omega}}\}.
\label{eq:app_disjoint_union}
\end{equation}
Here \(\bigsqcup\) denotes a disjoint union.
Using Eqs.~\eqref{eq:app_disjoint_union} and \eqref{eq:app_preimage_count},
\begin{align}
\frac{1}{|\mathcal S|}
\sum_{s\in\mathcal S}
e^{-i\bm{x}\cdot V_s[\bm{\omega}]}
&=
\frac{1}{|\mathcal S|}
\sum_{\tilde{\bm{\omega}}\in[\bm{\omega}]}
\sum_{\substack{s\in\mathcal S:\\ V_s[\bm{\omega}]=\tilde{\bm{\omega}}}}
e^{-i\bm{x}\cdot V_s[\bm{\omega}]}
\nonumber\\
&=
\frac{1}{|\mathcal S|}
\sum_{\tilde{\bm{\omega}}\in[\bm{\omega}]}
\sum_{\substack{s\in\mathcal S:\\ V_s[\bm{\omega}]=\tilde{\bm{\omega}}}}
e^{-i\bm{x}\cdot\tilde{\bm{\omega}}}
\nonumber\\
&=
\frac{1}{|\mathcal S|}
\sum_{\tilde{\bm{\omega}}\in[\bm{\omega}]}
\frac{|\mathcal S|}{|[\bm{\omega}]|}
e^{-i\bm{x}\cdot\tilde{\bm{\omega}}}
\nonumber\\
&=
\frac{1}{|[\bm{\omega}]|}
\sum_{\tilde{\bm{\omega}}\in[\bm{\omega}]}
e^{-i\bm{x}\cdot\tilde{\bm{\omega}}},
\qquad
\bm{\omega}\in\Omega.
\label{eq:app_orbit_average}
\end{align}
Combining Eqs.~\eqref{eq:app_inverse_average} and \eqref{eq:app_orbit_average} with the definition of \(\phi_{[\bm{\omega}]}\) gives
\begin{equation}
\frac{1}{|\mathcal S|}
\sum_{s\in\mathcal S}
e^{-i\bm{x}\cdot V_{s^{-1}}[\bm{\omega}]}
=
\phi_{[\bm{\omega}]}(\bm{x}).
\label{eq:app_inner_average_basis}
\end{equation}
After substituting Eq.~\eqref{eq:app_inner_average_basis} into Eq.~\eqref{eq:app_average_expansion}, we can first write the averaged function as a sum over the original frequencies,
\begin{equation}
(\mathcal A f)(\bm{x},\bm{\theta})
=
\sum_{\bm{\omega}\in\Omega}
c_{\bm{\omega}}(\bm{\theta})
\phi_{[\bm{\omega}]}(\bm{x}).
\label{eq:app_original_frequency_sum}
\end{equation}
Choose one representative \(\bm{\nu}\in\Omega\) from each distinct orbit generated by \(\Omega\), and denote the set of representatives by \(\tilde{\Omega}\):
\begin{equation}
\tilde{\Omega}\subseteq\Omega,\qquad
\Omega\subseteq\bigcup_{\bm{\nu}\in\tilde{\Omega}}[\bm{\nu}].
\label{eq:app_orbit_frequency_spectrum}
\end{equation}
In particular, if \(V_s[\bm{\omega}]\in\Omega\) for every \(\bm{\omega}\in\Omega\) and \(s\in\mathcal S\), then
\begin{equation}
\Omega=\bigcup_{\bm{\nu}\in\tilde{\Omega}}[\bm{\nu}].
\end{equation}
This set \(\tilde{\Omega}\) is the orbit frequency spectrum. With \(c_{\bm{\omega}}(\bm{\theta})=0\) for \(\bm{\omega}\notin\Omega\), Eq.~\eqref{eq:app_original_frequency_sum} can be regrouped over full orbits:
\begin{align}
(\mathcal A f)(\bm{x},\bm{\theta})
&=
\sum_{\bm{\nu}\in\tilde{\Omega}}
\sum_{\bm{\omega}\in[\bm{\nu}]}
c_{\bm{\omega}}(\bm{\theta})
\phi_{[\bm{\omega}]}(\bm{x})
\nonumber\\
&=
\sum_{\bm{\nu}\in\tilde{\Omega}}
\sum_{\bm{\omega}\in[\bm{\nu}]}
c_{\bm{\omega}}(\bm{\theta})
\phi_{[\bm{\nu}]}(\bm{x})
\nonumber\\
&=
\sum_{\bm{\nu}\in\tilde{\Omega}}
\left(
\sum_{\bm{\omega}\in[\bm{\nu}]}
c_{\bm{\omega}}(\bm{\theta})
\right)
\phi_{[\bm{\nu}]}(\bm{x})
\nonumber\\
&=
\sum_{\bm{\nu}\in\tilde{\Omega}}
a_{\bm{\nu}}(\bm{\theta})
\phi_{[\bm{\nu}]}(\bm{x}),
\end{align}
where the second equality uses \([\bm{\omega}]=[\bm{\nu}]\) for \(\bm{\omega}\in[\bm{\nu}]\), and the last equality uses Eq.~\eqref{eq:app_symmetrized_coefficient}. This proves Eq.~\eqref{eq:app_FOD}.
\end{proof}

\section{Exact 2-design coefficient results}

\subsection{Coefficient decoupling}
\label{app:fourier_coefficient_decoupling}

\begin{restate}{Lemma}{lemma:fourier_coefficient_decoupling}
Consider the QFM hypothesis function \(f(\bm{x},\bm{\theta})\) in Eq.~\eqref{eq:QFM}, with Fourier coefficients \(c_{\bm{\omega}}(\bm{\theta})\). If each trainable layer \(W^{(l)}\) forms an independent exact 2-design, then for any distinct frequencies \(\bm{\omega}\neq\bm{\omega}'\),
\begin{align}
\mathbb{E}_{\bm{\theta}} [c_{\bm{\omega}} c_{\bm{\omega}'}^*] &= 0,\\
\operatorname{Cov}_{\bm{\theta}}(c_{\bm{\omega}}, c_{\bm{\omega}'}) &= 0.
\end{align}
\end{restate}

\begin{proof}
Recall that the Fourier coefficient is defined by Eq.~\eqref{eq:QFM_coefficient}:
\[
c_{\bm{\omega}}=\sum_{\substack{\left(J, J^{\prime}\right): \\ \bm{\Lambda}_J-\bm{\Lambda}_{J^{\prime}}=\bm{\omega}}} \sum_{j^{L+1}, j^{\prime L+1}} O_{j^{\prime L+1} j^{L+1}} \prod_{l=1}^{L+1}\left(W_{j^{\prime l}, j^{\prime l-1}}^{(l)}\right)^* W_{j^l, j^{l-1}}^{(l)}
\]
Thus,
\[
\mathbb{E}_{\bm{\theta}} [c_{\bm{\omega}} c_{\bm{\omega}'}^*] = \sum_{J, J'} \sum_{K, K'} \sum_{\substack{j^{L+1}, j^{\prime L+1}\\
k^{L+1}, k^{\prime L+1}}} O_{j'^{L+1} j^{L+1}} O_{k'^{L+1} k^{L+1}}^* \cdot \mathbb{E}_{\bm{\theta}} \left[ \prod_{l=1}^{L+1} \left( W^{(l)}_{j'^l, j'^{l-1}} \right)^* W^{(l)}_{j^l, j^{l-1}} W^{(l)}_{k'^l, k'^{l-1}} \left( W^{(l)}_{k^l, k^{l-1}} \right)^* \right]
\]
By the independence of the layers,
\[
\mathbb{E}_{\bm{\theta}} \left[ \prod_{l=1}^{L+1} (\cdots) \right] = \prod_{l=1}^{L+1} \mathbb{E}_{\bm{\theta}_l} \left[ W^{(l)}_{j^l, j^{l-1}} W^{(l)}_{k'^l, k'^{l-1}} \left( W^{(l)}_{j'^l, j'^{l-1}} \right)^* \left( W^{(l)}_{k^l, k^{l-1}} \right)^* \right]
\]

A probability distribution $\nu$ over the unitary group $U(d)$ is a unitary $k$-design if it exactly reproduces the first $k$ moments of the Haar measure $\mu_H$. In terms of tensor products, this definition states:
\[
\underset{V \sim \nu}{\mathbb{E}}\left[V^{\otimes k} \otimes V^{* \otimes k}\right]=\underset{U \sim \mu_H}{\mathbb{E}}\left[U^{\otimes k} \otimes U^{* \otimes k}\right].
\]
Because the expectation of a tensor product corresponds directly to the expectation of the products of its individual matrix elements, the $k$-design condition guarantees that for any set of $2k$ indices, the expectation under $\nu$ matches that under the Haar measure. In our case, $k=2$. Assuming each trainable layer $W^{(l)}$ forms an exact 2-design over its parameter space $\bm{\theta}_l$, any 4th-order polynomial of its matrix elements simplifies precisely to the Haar expectation:
\[
\mathbb{E}_{\bm{\theta}_l} \left[ W^{(l)}_{j^l, j^{l-1}} W^{(l)}_{k'^l, k'^{l-1}} \left( W^{(l)}_{j'^l, j'^{l-1}} \right)^* \left( W^{(l)}_{k^l, k^{l-1}} \right)^* \right] = \mathbb{E}_{U\sim\mu_H} \left[ U_{j^l, j^{l-1}} U_{k'^l, k'^{l-1}} U_{j'^l, j'^{l-1}}^* U_{k^l, k^{l-1}}^* \right]
\]

Using the Weingarten calculus formula \cite[Eq.~(151)]{Mele2024},
\[
\begin{aligned}
\underset{U \sim \mu_H}{\mathbb{E}}\left[U_{i_1, j_1} U_{i_2, j_2} U_{i_3, j_3}^* U_{i_4, j_4}^*\right] & =\frac{1}{d^2-1}\left[\delta_{i_1, i_3} \delta_{i_2, i_4} \delta_{j_1, j_3} \delta_{j_2, j_4}-\frac{1}{d} \delta_{i_1, i_3} \delta_{i_2, i_4} \delta_{j_1, j_4} \delta_{j_2, j_3}\right] \\
& +\frac{1}{d^2-1}\left[-\frac{1}{d} \delta_{i_1, i_4} \delta_{i_2, i_3} \delta_{j_1, j_3} \delta_{j_2, j_4}+\delta_{i_1, i_4} \delta_{i_2, i_3} \delta_{j_1, j_4} \delta_{j_2, j_3}\right]
\end{aligned}
\]
We substitute our specific indices ($i_1=j^l, i_2=k'^l, i_3=j'^l, i_4=k^l$ and corresponding $j$ indices). By factoring out the row index combinations, we can decouple the pairing constraints from the subsequent column terms:
\[
\begin{aligned}
\mathbb{E}_{U\sim\mu_H} \left[ U_{j^l, j^{l-1}} \dots U_{k^l, k^{l-1}}^* \right]
&= \delta_{j^l, j'^l} \delta_{k'^l, k^l} \left[ \frac{1}{d^2-1} \left( \delta_{j^{l-1}, j'^{l-1}} \delta_{k'^{l-1}, k^{l-1}} - \frac{1}{d} \delta_{j^{l-1}, k^{l-1}} \delta_{k'^{l-1}, j'^{l-1}} \right) \right] \\
&\quad + \delta_{j^l, k^l} \delta_{k'^l, j'^l} \left[ \frac{1}{d^2-1} \left( \delta_{j^{l-1}, k^{l-1}} \delta_{k'^{l-1}, j'^{l-1}} - \frac{1}{d} \delta_{j^{l-1}, j'^{l-1}} \delta_{k'^{l-1}, k^{l-1}} \right) \right]
\end{aligned}
\]

In this factored form, it becomes explicit that to yield a non-zero expectation, the row indices at layer $l$ must satisfy either the Identity pairing ($\delta_{j^l, j'^l}\delta_{k'^l, k^l} = 1$) or the Swap pairing ($\delta_{j^l, k^l}\delta_{k'^l, j'^l} = 1$). Both constraints mathematically enforce that the frequency increment at layer $l$ is identical for both paths.

Recall that $\boldsymbol\omega = \boldsymbol\Lambda_J - \boldsymbol\Lambda_{J'}$, where $\boldsymbol\Lambda_J = \bigl(\Lambda_J^1,\dots,\Lambda_J^m\bigr)$ with $\Lambda_J^i = \sum_{l=1}^L\lambda_{j_i^l}$ and $J=(j^1,\dots,j^L)\in(\prod_{i=1}^m[1,d_i])^L$. We can express $\bm{\omega}$ as a sum of layer-wise frequency increments $\bm{\omega} = \sum_{l=1}^L \Delta \bm{\omega}_l$, where
\[
\Delta \bm{\omega}_l = \boldsymbol\lambda_{j^l} - \boldsymbol\lambda_{j'^l}
\]
and $\boldsymbol\lambda_{j^l} = \bigl(\lambda_{j_1^l},\dots,\lambda_{j_m^l}\bigr)$. Hence, $\bm{\omega} = \sum_{l=1}^L (\bm{\lambda}_{j^l} - \bm{\lambda}_{j'^l})$, and similarly $\bm{\omega}' = \sum_{l=1}^L \Delta \bm{\omega}'_l$.

The Identity pairing implies $j^l=j'^l$ and $k^l=k'^l$, yielding $\Delta\bm{\omega}_l = \Delta\bm{\omega}'_l = 0$. Conversely, the Swap pairing implies $j^l=k^l$ and $j'^l=k'^l$, yielding exact index matching such that $\Delta \bm{\omega}_l = \Delta \bm{\omega}'_l$. In both scenarios, a non-zero expectation requires $\Delta \bm{\omega}_l = \Delta \bm{\omega}'_l$ for all layers $l$, which directly implies $\bm{\omega} = \bm{\omega}'$. However, this contradicts our explicit premise that $\bm{\omega} \neq \bm{\omega}'$. Consequently, at least one layerwise expectation term must vanish:
\[
\mathbb{E}_{\bm{\theta}_l} \left[ W^{(l)}_{j^l, j^{l-1}} W^{(l)}_{k'^l, k'^{l-1}} \left( W^{(l)}_{j'^l, j'^{l-1}} \right)^* \left( W^{(l)}_{k^l, k^{l-1}} \right)^* \right] = 0
\]
Because the total expectation is a product over all layers, the entire product evaluates to zero. Hence, we conclude:
\[
\mathbb{E}_{\bm{\theta}} [c_{\bm{\omega}} c_{\bm{\omega}'}^*] = 0
\]
Given that the coefficients satisfy $\mathbb{E}_{\bm{\theta}}\left[c_{\bm{\omega}}\right]=\frac{\operatorname{Tr}(O)}{d} \delta_{\bm{\omega}}^0$ \cite{Hela2025}, and that $\mathbb{E}_{\bm{\theta}} [c_{\bm{\omega}} c_{\bm{\omega}'}^*] = 0$ for all $\bm{\omega}\neq \bm{\omega}'$ as established above, the covariance between distinct coefficients is
\[
\operatorname{Cov}_{\bm{\theta}}\left(c_{\bm{\omega}}, c_{\bm{\omega}^{\prime}}\right)=\mathbb{E}_{\bm{\theta}}\left[c_{\bm{\omega}} c_{\bm{\omega}^{\prime}}^*\right]-\mathbb{E}_{\bm{\theta}}\left[c_{\bm{\omega}}\right] \mathbb{E}_{\bm{\theta}}\left[c_{\bm{\omega}^{\prime}}^*\right]= 0 - \frac{|\operatorname{Tr}(O)|^2}{d^2} \delta_{\bm{\omega}}^0 \delta_{\bm{\omega}^{\prime}}^0
\]
Since $\bm{\omega} \neq \bm{\omega}^{\prime}$, at least one of the frequencies must be non-zero, which means the product of the Kronecker deltas is zero ($\delta_{\bm{\omega}}^0 \delta_{\bm{\omega}^{\prime}}^0 = 0$). Thus, the covariance terms vanish entirely: $\operatorname{Cov}_{\bm{\theta}}\left(c_{\bm{\omega}}, c_{\bm{\omega}^{\prime}}\right) = 0$.

\end{proof}

\subsection{Symmetrized coefficient variance}
\label{app:symmetrized_coefficient_variance}

\begin{restate}{Theorem}{theorem:exp_variance}
Consider the QFM hypothesis function \(f(\bm{x},\bm{\theta})\) in Eq.~\eqref{eq:QFM}, with Fourier coefficients \(c_{\bm{\omega}}(\bm{\theta})\) satisfying \(c_{\bm{\omega}}(\bm{\theta})=0\) for \(\bm{\omega}\notin\Omega\). Let \(a_{\bm{\nu}}(\bm{\theta})\) be the symmetrized coefficient defined in Eq.~\eqref{eq:symmetrized_coefficient}. If each trainable layer forms an independent exact 2-design, then for any orbit frequency \(\bm{\nu}\) in the orbit frequency spectrum \(\tilde{\Omega}\),
\begin{equation}
\operatorname{Var}_{\bm{\theta}}(a_{\bm{\nu}}) = \sum_{\bm{\omega} \in [\bm{\nu}]} \operatorname{Var}_{\bm{\theta}}(c_{\bm{\omega}}) = |[\bm{\nu}]| \cdot \overline{\operatorname{Var}}_{\bm{\theta}}(c_{\bm{\omega}}),
\end{equation}
where \(\overline{\operatorname{Var}}_{\bm{\theta}}(c_{\bm{\omega}})\) denotes the mean variance over the distinct frequencies within the orbit \([\bm{\nu}]=\{ V_s[\bm{\nu}] \mid s \in \mathcal{S} \}\).
\end{restate}

\begin{proof}
The aggregated coefficient over the orbit \([\bm{\nu}]=\{ V_s[\bm{\nu}] \mid s \in \mathcal{S} \}\) is defined as \(a_{\bm{\nu}} := \sum_{\bm{\omega} \in [\bm{\nu}]} c_{\bm{\omega}}\). Thus, the variance of \(a_{\bm{\nu}}\) expands as
\begin{equation}
\label{eq:var_expansion}
\operatorname{Var}_{\bm{\theta}}(a_{\bm{\nu}}) = \sum_{\bm{\omega} \in [\bm{\nu}]} \operatorname{Var}_{\bm{\theta}}(c_{\bm{\omega}}) + \sum_{\substack{\bm{\omega} \neq \bm{\omega}' \\ \bm{\omega}, \bm{\omega}' \in [\bm{\nu}]}} \operatorname{Cov}_{\bm{\theta}}(c_{\bm{\omega}}, c_{\bm{\omega}'}).
\end{equation}
By Lemma~\ref{lemma:fourier_coefficient_decoupling}, the covariance terms vanish, giving
\begin{equation}
\operatorname{Var}_{\bm{\theta}}(a_{\bm{\nu}}) = \sum_{\bm{\omega} \in [\bm{\nu}]} \operatorname{Var}_{\bm{\theta}}(c_{\bm{\omega}})
= |[\bm{\nu}]| \cdot \overline{\operatorname{Var}}_{\bm{\theta}}(c_{\bm{\omega}}).
\end{equation}
\end{proof}

\section{Approximate 2-design coefficient results}
\label{sec:fourier_coefficient_approximate_2design}

\subsection{Moment operators and preliminary bounds}

Recall that we consider the situation of an \(\varepsilon\)-approximate 2-design when the number of layers is $L=1$. Consider a single-uploading QFM
\begin{equation}
U(\bm{x},\bm{\theta})
=
W^{(2)}(\bm{\theta}_2)E(\bm{x})W^{(1)}(\bm{\theta}_1),
\end{equation}
where $d=2^n$, $W^{(1)}$ and $W^{(2)}$ are independent trainable unitaries, and the encoding layer is diagonal in a fixed basis $E(\bm{x})|j\rangle=e^{-i\bm{x}\cdot\bm{\lambda}_j}|j\rangle$. Since \(L=1\), we drop the layer superscript and write \(j\in[1,d]\) for the flattened label of \(j^1=(j_1^1,\ldots,j_m^1)\in\prod_{i=1}^m[1,d_i]\), with \(d=\prod_{i=1}^m d_i\) and \(\bm{\lambda}_j\in\mathbb R^m\). Let the Fourier decomposition of the hypothesis function be
\begin{equation}
f(\bm{x},\bm{\theta})
=
\langle 0|
U(\bm{x},\bm{\theta})^\dagger
\hat O
U(\bm{x},\bm{\theta})
|0\rangle
=
\sum_{\bm{\omega}\in\Omega}
c_{\bm{\omega}}(\bm{\theta})
e^{-i\bm{x}\cdot\bm{\omega}}.
\end{equation}
For each frequency $\bm{\omega}\in\Omega$, define its frequency generator as
\begin{equation}
\mathfrak{R}(\bm{\omega})
:=
\left\{
(j,j'):\bm{\lambda}_{j}-\bm{\lambda}_{j'}=\bm{\omega}
\right\}.
\end{equation}
The two channels below are layer-specific notational substitutes for the generic second-moment channel in Eq.~\eqref{eq:second_moment_channel}. Define the second-moment channel of $W^{(1)}$ and the reverse second-moment channel of $W^{(2)}$ by
\begin{align}
&\mathcal{M}^{(2)}_1[X]
:=
\mathbb{E}_{\bm{\theta}_1}
\left[
(W^{(1)}\otimes W^{(1)})
X
(W^{(1)\dagger}\otimes W^{(1)\dagger})
\right],\\
&\widetilde{\mathcal{M}}^{(2)}_2[X]
:=
\mathbb{E}_{\bm{\theta}_2}
\left[
(W^{(2)\dagger}\otimes W^{(2)\dagger})
X
(W^{(2)}\otimes W^{(2)})
\right],
\end{align}
where $W^{(1)},W^{(2)}$ are the two trainable layers with $L=1$. Recall that the exact Haar measure channel is defined by
\begin{equation}
\mathcal{H}^{(2)}[X]
:=
\mathbb{E}_{U\sim\mu_H}
\left[
(U\otimes U)
X
(U^\dagger\otimes U^\dagger)
\right]=\mathbb{E}_{U\sim\mu_H}
\left[
(U^\dagger\otimes U^\dagger)
X
(U\otimes U)
\right].
\end{equation}
Following the property used in \cite[Appendix B.2]{Hela2025}, the trainable-layer average satisfies \(\mathbb{E}_{\bm{\theta}}[W^{\otimes 2}(\cdot)(W^\dagger)^{\otimes 2}]=\mathbb{E}_{\bm{\theta}}[(W^\dagger)^{\otimes 2}(\cdot)W^{\otimes 2}]\). Recall that Eq.~\eqref{eq:approx_2_design_channel_assumption} defines the \(\varepsilon\)-approximate 2-design condition \(\|
\mathcal{M}^{(2)}-\mathcal{H}^{(2)}\|_{\infty} \leq \varepsilon\). In the present single-uploading setting, this condition gives the two layer-specific bounds
\begin{equation}
\left\|
\mathcal{M}^{(2)}_1-\mathcal{H}^{(2)}
\right\|_{\infty}
\leq
\varepsilon,
\qquad
\left\|
\widetilde{\mathcal{M}}^{(2)}_2-\mathcal{H}^{(2)}
\right\|_{\infty}
\leq
\varepsilon .
\end{equation}
Here $\|\cdot\|_{\infty}$ denotes the spectral norm of the matrix form of the corresponding superoperator.

To streamline our analysis of the Fourier coefficients, we first rigorously define the moment operators for the trainable layers and their deviations from the exact Haar measure.

\begin{definition}[Moment Operators and Deviations]
\label{def:moment_operators}
For the independent trainable layers $W^{(1)}$ and $W^{(2)}$, we define the actual first-moment operators $\bar A, \bar B$ and second-moment operators $A, B$ over their respective index pairs. Here \(W^{(1)}_{j,0}:=\langle j|W^{(1)}|0\rangle\) with \(j \in [1,d]\). Both \(|0\rangle\langle0|\) and \(I/d\) are \(d\times d\) operators. They can be equivalently expressed via parameter expectations or superoperator channels:
\begin{align}
\bar A_{j,j'}
&:=
\mathbb{E}_{\bm{\theta}_1} \left[ W^{(1)}_{j,0} \left(W^{(1)}_{j',0}\right)^* \right]
=
\left\langle j\right| \operatorname{Tr}_2 \left[ \mathcal{M}^{(2)}_1 \left( |0\rangle\langle0|\otimes\frac{I}{d} \right) \right] \left|j'\right\rangle, \\
\bar B_{j,j'}
&:=
\mathbb{E}_{\bm{\theta}_2} \left[ \left\langle j'\right| W^{(2)\dagger}\hat O W^{(2)} \left|j\right\rangle \right]
=
\left\langle j'\right| \operatorname{Tr}_2 \left[ \widetilde{\mathcal{M}}^{(2)}_2 \left( \hat O\otimes\frac{I}{d} \right) \right] \left|j\right\rangle, \\
A_{j,j';k,k'}
&:=
\mathbb{E}_{\bm{\theta}_1} \left[ W^{(1)}_{j,0} W^{(1)}_{k',0} \left(W^{(1)}_{j',0}\right)^* \left(W^{(1)}_{k,0}\right)^* \right]
=
\left\langle j,k'\right| \mathcal{M}^{(2)}_1 \left( |0,0\rangle\langle0,0| \right) \left|j',k\right\rangle, \\
B_{j,j';k,k'}
&:=
\mathbb{E}_{\bm{\theta}_2} \left[ \left\langle j'\right| W^{(2)\dagger}\hat O W^{(2)} \left|j\right\rangle \left( \left\langle k'\right| W^{(2)\dagger}\hat O W^{(2)} \left|k\right\rangle \right)^* \right]
=
\left\langle j',k\right| \widetilde{\mathcal{M}}^{(2)}_2 \left( \hat O\otimes\hat O \right) \left|j,k'\right\rangle.
\end{align}

The corresponding ideal Haar moment operators, denoted with the superscript $\mathrm{H}$, are obtained by replacing the parameterized layers with unitaries drawn from the exact Haar measure $U \sim \mu_H$:
\begin{align}
\bar A^{\mathrm H}_{j,j'}
&:=
\mathbb{E}_{U\sim\mu_H} \left[ U_{j,0}U_{j',0}^* \right]
=
\left\langle j\right| \operatorname{Tr}_2 \left[ \mathcal{H}^{(2)} \left( |0\rangle\langle0|\otimes\frac{I}{d} \right) \right] \left|j'\right\rangle, \\
\bar B^{\mathrm H}_{j,j'}
&:=
\mathbb{E}_{U\sim\mu_H} \left[ \left\langle j'\right| U^\dagger\hat O U \left|j\right\rangle \right]
=
\left\langle j'\right| \operatorname{Tr}_2 \left[ \mathcal{H}^{(2)} \left( \hat O\otimes\frac{I}{d} \right) \right] \left|j\right\rangle, \\
A^{\mathrm H}_{j,j';k,k'}
&:=
\mathbb{E}_{U\sim\mu_H} \left[ U_{j,0} U_{k',0} U^*_{j',0} U^*_{k,0} \right]
=
\left\langle j,k'\right| \mathcal H^{(2)} \left( |0,0\rangle\langle0,0| \right) \left|j',k\right\rangle, \\
B^{\mathrm H}_{j,j';k,k'}
&:=
\mathbb{E}_{U\sim\mu_H} \left[ \left\langle j'\right| U^\dagger\hat O U \left|j\right\rangle \left( \left\langle k'\right| U^\dagger\hat O U \left|k\right\rangle \right)^* \right]
=
\left\langle j',k\right| \mathcal H^{(2)} \left( \hat O\otimes\hat O \right) \left|j,k'\right\rangle.
\end{align}

Finally, the deviation operators ($\Delta\bar A, \Delta\bar B, \Delta A, \Delta B$) represent the differences between the actual moments and the ideal Haar moments and correspond directly to the relevant deviation channels $(\mathcal{M} - \mathcal{H}^{(2)})$:
\begin{align}
\Delta\bar A_{j,j'}
&:=
\bar A_{j,j'} - \bar A^{\mathrm H}_{j,j'}
=
\left\langle j\right| \operatorname{Tr}_2 \left[ \left( \mathcal M^{(2)}_1-\mathcal H^{(2)} \right) \left( |0\rangle\langle0|\otimes\frac{I}{d} \right) \right] \left|j'\right\rangle, \\
\Delta\bar B_{j,j'}
&:=
\bar B_{j,j'} - \bar B^{\mathrm H}_{j,j'}
=
\left\langle j'\right| \operatorname{Tr}_2 \left[ \left( \widetilde{\mathcal M}^{(2)}_2-\mathcal H^{(2)} \right) \left( \hat O\otimes\frac{I}{d} \right) \right] \left|j\right\rangle, \\
\Delta A_{j,j';k,k'}
&:=
A_{j,j';k,k'} - A^{\mathrm H}_{j,j';k,k'}
=
\left\langle j,k'\right| \left( \mathcal M^{(2)}_1-\mathcal H^{(2)} \right) \left( |0,0\rangle\langle0,0| \right) \left|j',k\right\rangle, \\
\Delta B_{j,j';k,k'}
&:=
B_{j,j';k,k'} - B^{\mathrm H}_{j,j';k,k'}
=
\left\langle j',k\right| \left( \widetilde{\mathcal M}^{(2)}_2-\mathcal H^{(2)} \right) \left( \hat O\otimes\hat O \right) \left|j,k'\right\rangle.
\end{align}
\end{definition}

With these operators established, we can express the statistical moments of the Fourier coefficients concisely.
\begin{lemma}[Coefficient Moment Relations]
\label{lemma:coefficient_moment_relations}
For a single-uploading circuit ($L=1$), the expectation, cross-moment, and the aggregated second moment of the symmetrized coefficient $a_{\bm{\nu}} = \sum_{\bm{\omega}\in[\bm{\nu}]} c_{\bm{\omega}}$ over the frequency orbit $[\bm{\nu}]$ take the following decoupled forms. For the parameterized expectations, we have:
\begin{align}
\mathbb{E}_{\bm{\theta}} \left[ c_{\bm{\omega}} \right] &= \sum_{(j,j')\in\mathfrak R(\bm{\omega})} \bar A_{j,j'} \bar B_{j,j'}, \label{eq:coeff_first_moment_theta} \\
\mathbb{E}_{\bm{\theta}} \left[ c_{\bm{\omega}} c_{\bm{\omega}'}^* \right] &= \sum_{(j,j')\in\mathfrak R(\bm{\omega})} \sum_{(k,k')\in\mathfrak R(\bm{\omega}')} A_{j,j';k,k'} B_{j,j';k,k'}, \label{eq:coeff_cross_moment_theta} \\
\mathbb{E}_{\bm{\theta}} \left[ |a_{\bm{\nu}}|^2 \right] &= \sum_{(j,j'), (k,k') \in \mathfrak{R}_{[\bm{\nu}]}} A_{j,j';k,k'} B_{j,j';k,k'}. \label{eq:coeff_aggregated_moment_theta}
\end{align}
Correspondingly, for the exact Haar measure expectations, the same structural decoupling holds:
\begin{align}
\mathbb{E}_{\mathcal{H}} \left[ c_{\bm{\omega}} \right] &= \sum_{(j,j')\in\mathfrak R(\bm{\omega})} \bar A^{\mathrm{H}}_{j,j'} \bar B^{\mathrm{H}}_{j,j'}, \label{eq:coeff_first_moment_haar} \\
\mathbb{E}_{\mathcal{H}} \left[ c_{\bm{\omega}} c_{\bm{\omega}'}^* \right] &= \sum_{(j,j')\in\mathfrak R(\bm{\omega})} \sum_{(k,k')\in\mathfrak R(\bm{\omega}')} A^{\mathrm{H}}_{j,j';k,k'} B^{\mathrm{H}}_{j,j';k,k'}, \label{eq:coeff_cross_moment_haar} \\
\mathbb{E}_{\mathcal{H}} \left[ |a_{\bm{\nu}}|^2 \right] &= \sum_{(j,j'), (k,k') \in \mathfrak{R}_{[\bm{\nu}]}} A^{\mathrm{H}}_{j,j';k,k'} B^{\mathrm{H}}_{j,j';k,k'}, \label{eq:coeff_aggregated_moment_haar}
\end{align}
where $\mathfrak{R}_{[\bm{\nu}]} := \bigcup_{\bm{\omega}\in[\bm{\nu}]} \mathfrak{R}(\bm{\omega})$ is the orbit frequency generator.
\end{lemma}

\begin{proof}
Recall the Fourier coefficient is given by
\[
c_{\boldsymbol\omega}
=\sum_{(J,J'):\,\boldsymbol\Lambda_J-\boldsymbol\Lambda_{J'}=\boldsymbol\omega}
\sum_{k',k}
\bigl(W^{(L+1)}_{\,k',j'^L}\cdots W^{(1)}_{\,j'^1,j^0}\bigr)^*
\,O_{k'k}\,
\bigl(W^{(L+1)}_{\,k,j^L}\cdots W^{(1)}_{\,j^1,j^0}\bigr).
\]
For the single-uploading circuit with $L=1$, since $\sum_{k^{\prime}, k}\left(W_{k^{\prime}, j^{\prime}}^{(2)}\right)^* O_{k^{\prime} k} W_{k, j}^{(2)}=\left\langle j^{\prime}\right| W^{(2) \dagger} \hat{O} W^{(2)}|j\rangle$, the Fourier coefficient can be written as
\begin{equation}
\label{eq:single_uploading_coefficient}
c_{\bm{\omega}}
=
\sum_{(j,j')\in\mathfrak{R}(\bm{\omega})}
W^{(1)}_{j,0}
\left(W^{(1)}_{j',0}\right)^*
\left\langle j'\right|
W^{(2)\dagger}\hat O W^{(2)}
\left|j\right\rangle .
\end{equation}
When we take the expectation over the parameters $\bm{\theta} = (\bm{\theta}_1, \bm{\theta}_2)$, the independence of the layers $W^{(1)}$ and $W^{(2)}$ allows the expectation to factorize. Applying the definitions of $\bar A$ and $\bar B$ yields the first moment:
\begin{align}
\mathbb{E}_{\bm{\theta}} \left[ c_{\bm{\omega}} \right]
&=
\sum_{(j,j')\in\mathfrak{R}(\bm{\omega})}
\mathbb{E}_{\bm{\theta}_1} \left[ W^{(1)}_{j,0} \left(W^{(1)}_{j',0}\right)^* \right]
\mathbb{E}_{\bm{\theta}_2} \left[ \left\langle j'\right| W^{(2)\dagger}\hat O W^{(2)} \left|j\right\rangle \right] \nonumber \\
&=
\sum_{(j,j')\in\mathfrak R(\bm{\omega})} \bar A_{j,j'} \bar B_{j,j'}.
\end{align}
By identical reasoning, replacing the parameter distribution with the exact Haar measure $U \sim \mu_H$ directly yields $\mathbb{E}_{\mathcal{H}} [ c_{\bm{\omega}} ] = \sum_{\mathfrak R(\bm{\omega})} \bar A^{\mathrm{H}}_{j,j'} \bar B^{\mathrm{H}}_{j,j'}$. To derive the cross-moment, we write the complex conjugate of a coefficient for a frequency $\bm{\omega}'$:
\begin{equation}
c_{\bm{\omega}'}^*
=
\sum_{(k,k')\in\mathfrak{R}(\bm{\omega}')}
\left(W^{(1)}_{k,0}\right)^*
W^{(1)}_{k',0}
\left\langle k\right| W^{(2)\dagger}\hat O W^{(2)} \left|k'\right\rangle.
\end{equation}
Multiplying $c_{\bm{\omega}}$ and $c_{\bm{\omega}'}^*$ produces a double sum over the index pairs generating both frequencies. Taking the expectation $\mathbb{E}_{\bm{\theta}}$ and factoring the independent layer contributions gives:
\begin{align}
\mathbb{E}_{\bm{\theta}} \left[ c_{\bm{\omega}} c_{\bm{\omega}'}^* \right]
&=
\sum_{(j,j')\in\mathfrak R(\bm{\omega})} \sum_{(k,k')\in\mathfrak R(\bm{\omega}')}
\mathbb{E}_{\bm{\theta}_1} \left[ W^{(1)}_{j,0} W^{(1)}_{k',0} \left(W^{(1)}_{j',0}\right)^* \left(W^{(1)}_{k,0}\right)^* \right] \nonumber \\
&\quad \times
\mathbb{E}_{\bm{\theta}_2} \left[ \left\langle j'\right| W^{(2)\dagger}\hat O W^{(2)} \left|j\right\rangle \left( \left\langle k'\right| W^{(2)\dagger}\hat O W^{(2)} \left|k\right\rangle \right)^* \right] \nonumber \\
&=
\sum_{(j,j')\in\mathfrak R(\bm{\omega})} \sum_{(k,k')\in\mathfrak R(\bm{\omega}')} A_{j,j';k,k'} B_{j,j';k,k'}.
\end{align}
Under the exact Haar measure, this factorization analogously produces
\begin{equation}
\mathbb{E}_{\mathcal{H}} [ c_{\bm{\omega}} c_{\bm{\omega}'}^* ] = \sum_{(j,j')\in\mathfrak R(\bm{\omega})} \sum_{(k,k')\in\mathfrak R(\bm{\omega}')}  A^{\mathrm{H}}_{j,j';k,k'} B^{\mathrm{H}}_{j,j';k,k'}.
\end{equation}
Finally, for the aggregated second moment of the symmetrized coefficient $a_{\bm{\nu}} = \sum_{\bm{\omega}\in[\bm{\nu}]} c_{\bm{\omega}}$, we expand the absolute square as a sum over all frequency pairs in the orbit:
\begin{equation}
\mathbb{E}_{\bm{\theta}} \left[ |a_{\bm{\nu}}|^2 \right]
=
\mathbb{E}_{\bm{\theta}} \left[ \left( \sum_{\bm{\omega}\in[\bm{\nu}]} c_{\bm{\omega}} \right) \left( \sum_{\bm{\omega}'\in[\bm{\nu}]} c_{\bm{\omega}'}^* \right) \right]
=
\sum_{\bm{\omega}\in[\bm{\nu}]} \sum_{\bm{\omega}'\in[\bm{\nu}]} \mathbb{E}_{\bm{\theta}} \left[ c_{\bm{\omega}} c_{\bm{\omega}'}^* \right].
\end{equation}
Substituting the cross-moment expression expands this into nested sums over $\bm{\omega}, \bm{\omega}'$ and their respective index sets $\mathfrak{R}(\bm{\omega}), \mathfrak{R}(\bm{\omega}')$. Because the frequency generators partition the index space by frequency, the disjoint union $\bigcup_{\bm{\omega}\in[\bm{\nu}]} \mathfrak{R}(\bm{\omega}) = \mathfrak{R}_{[\bm{\nu}]}$ allows us to compactly merge the sums into a single double-summation over the orbit set:
\begin{equation}
\mathbb{E}_{\bm{\theta}} \left[ |a_{\bm{\nu}}|^2 \right]
=
\sum_{(j,j') \in \mathfrak{R}_{[\bm{\nu}]}} \sum_{(k,k') \in \mathfrak{R}_{[\bm{\nu}]}} A_{j,j';k,k'} B_{j,j';k,k'}.
\end{equation}
Applying the same orbit aggregation logic to the ideal Haar cross-moments gives the result for $\mathbb{E}_{\mathcal{H}} [ |a_{\bm{\nu}}|^2 ]$, completing the proof.
\end{proof}

\subsection{Fourier coefficient moments and covariance}

\begin{restate}{Lemma}{lemma:distinct_frequency_properties}
Consider the single-uploading QFM hypothesis function \(f(\bm{x},\bm{\theta})\) in Eq.~\eqref{eq:QFM} with \(L=1\) and Fourier coefficients \(c_{\bm{\omega}}(\bm{\theta})\). If each trainable layer \(W^{(l)}\) forms an \(\varepsilon\)-approximate 2-design in the sense of Eq.~\eqref{eq:approx_2_design_channel_assumption}, then for any two distinct nonzero frequencies \(\bm{\omega}\neq\bm{\omega}'\), we have
\begin{align}
&\left|
\mathbb{E}_{\bm{\theta}}
\left[
c_{\bm{\omega}}
\right]
\right|
\leq
\varepsilon^2\|\hat O\|_F,\\
&\left|
\mathbb{E}_{\bm{\theta}}
\left[
c_{\bm{\omega}}c_{\bm{\omega}'}^*
\right]
\right|
\leq
\varepsilon^2
\|\hat O\|_F^2,\\
&\left|
\operatorname{Cov}_{\bm{\theta}}
\left(
c_{\bm{\omega}},c_{\bm{\omega}'}
\right)
\right|
\leq
\|\hat O\|_F^2
\left(
\varepsilon^2+\varepsilon^4
\right),
\end{align}
where \(\|\cdot\|_{F}\) denotes the Frobenius norm.
\end{restate}

\begin{lemma}[Superoperator Frobenius norm]
\label{lemma:superoperator_norm}
For any superoperator (or quantum channel) $\mathcal{E}$ acting on a matrix $X$, the following inequality holds:
\begin{equation}
    \|\mathcal{E}(X)\|_F \leq \|\mathcal{E}\|_{\infty} \|X\|_F,
\end{equation}
where $\|\cdot\|_F$ denotes the Frobenius norm, and $\|\mathcal{E}\|_{\infty}$ denotes the spectral norm of the matrix form of the corresponding superoperator.
\end{lemma}

\begin{proof}
Let $|X\rangle\rangle = \operatorname{vec}(X)$ denote the vectorized column vector of matrix $X$. The Frobenius norm of $X$ is equal to the $L_2$ norm of its vectorized form:
\begin{equation}
    \|X\|_F = \||X\rangle\rangle\|_2.
\end{equation}
Since the superoperator $\mathcal{E}$ is a linear map, it can be written as a matrix $\hat{\mathcal{E}}$ acting on the vectorized state, such that $|\mathcal{E}(X)\rangle\rangle = \hat{\mathcal{E}}|X\rangle\rangle$. For a linear operator $A: V \rightarrow W$, define the spectral norm (Schatten $\infty$-norm) of $A$ as $\|A\|_{\infty}:=\sup\{\frac{\|A v\|_2}{\|v\|_2}: v \neq 0\; \text{and}\; v \in V\}$. By this definition, we also have $\|Av\|_2 \leq \|A\|_{\infty}\|v\|_2$. Applying this fundamental property directly to the matrix form $\hat{\mathcal{E}}$ and vector $|X\rangle\rangle$ yields:
\begin{equation}
    \| \hat{\mathcal{E}} |X\rangle\rangle \|_2 \leq \|\hat{\mathcal{E}}\|_{\infty} \| |X\rangle\rangle \|_2.
\end{equation}
By definition, the spectral norm of the superoperator $\mathcal{E}$ is exactly the spectral norm of its matrix form ($\|\mathcal{E}\|_{\infty} = \|\hat{\mathcal{E}}\|_{\infty}$). Substituting the equivalent Frobenius norms back into the inequality, we immediately obtain:
\begin{equation}
    \|\mathcal{E}(X)\|_F \leq \|\mathcal{E}\|_{\infty} \|X\|_F.
\end{equation}
This completes the proof.
\end{proof}

\begin{lemma}[Restricted Frobenius bounds from approximate second moments]
\label{lemma:restricted_frobenius_from_design}
Let $\mathfrak R(\bm{\omega})$ and $\mathfrak R(\bm{\omega}')$ be any two frequency generators. Under the $\varepsilon$-approximate 2-design assumptions defined in Eq.~\eqref{eq:approx_2_design_channel_assumption}, for every operator $X$ on $\mathbb C^d\otimes\mathbb C^d$,
\begin{align}
&
\left(
\sum_{(j,j')\in\mathfrak R(\bm{\omega})}
\sum_{(k,k')\in\mathfrak R(\bm{\omega}')}
\left|
\left\langle j,k'\right|
\left(
\mathcal M^{(2)}_1-\mathcal H^{(2)}
\right)(X)
\left|j',k\right\rangle
\right|^2
\right)^{1/2}
\leq
\varepsilon\|X\|_F,
\label{eq:restricted_bound_M1}
\\
&
\left(
\sum_{(j,j')\in\mathfrak R(\bm{\omega})}
\sum_{(k,k')\in\mathfrak R(\bm{\omega}')}
\left|
\left\langle j',k\right|
\left(
\widetilde{\mathcal M}^{(2)}_2-\mathcal H^{(2)}
\right)(X)
\left|j,k'\right\rangle
\right|^2
\right)^{1/2}
\leq
\varepsilon\|X\|_F.
\label{eq:restricted_bound_M2}
\end{align}
Moreover, for every operator $Y$ on $\mathbb C^d$,
\begin{align}
&
\left(
\sum_{(j,j')\in\mathfrak R(\bm{\omega})}
\left|
\left\langle j\right|
\operatorname{Tr}_2
\left[
\left(
\mathcal M^{(2)}_1-\mathcal H^{(2)}
\right)
\left(
Y\otimes\frac{I}{d}
\right)
\right]
\left|j'\right\rangle
\right|^2
\right)^{1/2}
\leq
\varepsilon\|Y\|_F,
\label{eq:partial_trace_bound_M1}
\\
&
\left(
\sum_{(j,j')\in\mathfrak R(\bm{\omega})}
\left|
\left\langle j'\right|
\operatorname{Tr}_2
\left[
\left(
\widetilde{\mathcal M}^{(2)}_2-\mathcal H^{(2)}
\right)
\left(
Y\otimes\frac{I}{d}
\right)
\right]
\left|j\right\rangle
\right|^2
\right)^{1/2}
\leq
\varepsilon\|Y\|_F.
\label{eq:partial_trace_bound_M2}
\end{align}
\end{lemma}

\begin{proof}
We prove Eq.~\eqref{eq:restricted_bound_M1}; the proof for Eq.~\eqref{eq:restricted_bound_M2} is identical. Let $Z_X = (\mathcal M^{(2)}_1-\mathcal H^{(2)})(X)$. The restricted sum of matrix elements is naturally bounded by the full Frobenius norm:
\begin{equation}
\sum_{(j,j')\in\mathfrak R(\bm{\omega})}
\sum_{(k,k')\in\mathfrak R(\bm{\omega}')}
\left|
\left\langle j,k'\right|
Z_X
\left|j',k\right\rangle
\right|^2
\leq
\sum_{a,b,c,e}
\left|
\left\langle a,b\right|
Z_X
\left|c,e\right\rangle
\right|^2
=
\|Z_X\|_F^2.
\end{equation}
Using Eq.~\eqref{eq:approx_2_design_channel_assumption} and Lemma~\ref{lemma:superoperator_norm}, we get $\|Z_X\|_F \leq \|\mathcal M^{(2)}_1-\mathcal H^{(2)}\|_{\infty} \|X\|_F \leq \varepsilon\|X\|_F$. This proves Eq.~\eqref{eq:restricted_bound_M1}. For the partial-trace estimate, let $Z_Y = (\mathcal M^{(2)}_1-\mathcal H^{(2)})(Y\otimes\frac{I}{d})$. Then
\begin{align}
&\sum_{(j,j')\in\mathfrak R(\bm{\omega})}
\left|
\left\langle j\right|
\operatorname{Tr}_2(Z_Y)
\left|j'\right\rangle
\right|^2
\leq
\left\| \operatorname{Tr}_2(Z_Y) \right\|_F^2
\leq
d\|Z_Y\|_F^2 \nonumber \\
&\leq
d\varepsilon^2
\left\| Y\otimes\frac{I}{d} \right\|_F^2
=
d\varepsilon^2 \|Y\|_F^2 \left\|\frac{I}{d}\right\|_F^2
=
\varepsilon^2\|Y\|_F^2,
\end{align}
where we used $\|\operatorname{Tr}_2(Z_Y)\|_F\leq\sqrt d\|Z_Y\|_F$ (this formula is shown in \cite[Proposition 1]{Rastegin2012}) and $\|I/d\|_F=1/\sqrt d$. This proves Eq.~\eqref{eq:partial_trace_bound_M1}, and Eq.~\eqref{eq:partial_trace_bound_M2} follows identically.
\end{proof}

\begin{lemma}[Restricted Frobenius bounds over frequency orbits]
\label{lemma:orbit_restricted_frobenius}
Let $[\bm{\nu}]$ and $[\bm{\nu}']$ be any two frequency orbits, and let $\mathfrak R_{[\bm{\nu}]} = \bigcup_{\bm{\omega}\in[\bm{\nu}]} \mathfrak R(\bm{\omega})$ and $\mathfrak R_{[\bm{\nu}']} = \bigcup_{\bm{\omega}\in[\bm{\nu}']} \mathfrak R(\bm{\omega})$ denote their aggregated frequency generators. Under the $\varepsilon$-approximate 2-design assumptions defined in Eq.~\eqref{eq:approx_2_design_channel_assumption}, for every operator $X$ on $\mathbb C^d\otimes\mathbb C^d$,
\begin{align}
&
\left(
\sum_{(j,j')\in\mathfrak R_{[\bm{\nu}]}}
\sum_{(k,k')\in\mathfrak R_{[\bm{\nu}']}}
\left|
\left\langle j,k'\right|
\left(
\mathcal M^{(2)}_1-\mathcal H^{(2)}
\right)(X)
\left|j',k\right\rangle
\right|^2
\right)^{1/2}
\leq
\varepsilon\|X\|_F,
\label{eq:orbit_restricted_bound_M1}
\\
&
\left(
\sum_{(j,j')\in\mathfrak R_{[\bm{\nu}]}}
\sum_{(k,k')\in\mathfrak R_{[\bm{\nu}']}}
\left|
\left\langle j',k\right|
\left(
\widetilde{\mathcal M}^{(2)}_2-\mathcal H^{(2)}
\right)(X)
\left|j,k'\right\rangle
\right|^2
\right)^{1/2}
\leq
\varepsilon\|X\|_F.
\label{eq:orbit_restricted_bound_M2}
\end{align}
Moreover, for every operator $Y$ on $\mathbb C^d$,
\begin{align}
&
\left(
\sum_{(j,j')\in\mathfrak R_{[\bm{\nu}]}}
\left|
\left\langle j\right|
\operatorname{Tr}_2
\left[
\left(
\mathcal M^{(2)}_1-\mathcal H^{(2)}
\right)
\left(
Y\otimes\frac{I}{d}
\right)
\right]
\left|j'\right\rangle
\right|^2
\right)^{1/2}
\leq
\varepsilon\|Y\|_F,
\label{eq:orbit_partial_trace_bound_M1}
\\
&
\left(
\sum_{(j,j')\in\mathfrak R_{[\bm{\nu}]}}
\left|
\left\langle j'\right|
\operatorname{Tr}_2
\left[
\left(
\widetilde{\mathcal M}^{(2)}_2-\mathcal H^{(2)}
\right)
\left(
Y\otimes\frac{I}{d}
\right)
\right]
\left|j\right\rangle
\right|^2
\right)^{1/2}
\leq
\varepsilon\|Y\|_F.
\label{eq:orbit_partial_trace_bound_M2}
\end{align}
\end{lemma}

\begin{proof}
We prove Eq.~\eqref{eq:orbit_restricted_bound_M1}; the proof for Eq.~\eqref{eq:orbit_restricted_bound_M2} is identical. Let $Z_X = (\mathcal M^{(2)}_1-\mathcal H^{(2)})(X)$. Since $\mathfrak R_{[\bm{\nu}]}$ and $\mathfrak R_{[\bm{\nu}']}$ are simply subsets of the full Hilbert space index pairs, the restricted sum of matrix elements over the orbit generators is naturally bounded by the full Frobenius norm:
\begin{equation}
\sum_{(j,j')\in\mathfrak R_{[\bm{\nu}]}}
\sum_{(k,k')\in\mathfrak R_{[\bm{\nu}']}}
\left|
\left\langle j,k'\right|
Z_X
\left|j',k\right\rangle
\right|^2
\leq
\sum_{a,b,c,e}
\left|
\left\langle a,b\right|
Z_X
\left|c,e\right\rangle
\right|^2
=
\|Z_X\|_F^2.
\end{equation}
Using Eq.~\eqref{eq:approx_2_design_channel_assumption} and Lemma~\ref{lemma:superoperator_norm}, we get $\|Z_X\|_F \leq \|\mathcal M^{(2)}_1-\mathcal H^{(2)}\|_{\infty} \|X\|_F \leq \varepsilon\|X\|_F$. Taking the square root proves Eq.~\eqref{eq:orbit_restricted_bound_M1}.

For the partial-trace estimate, let $Z_Y = (\mathcal M^{(2)}_1-\mathcal H^{(2)})(Y\otimes\frac{I}{d})$. Then
\begin{align}
&\sum_{(j,j')\in\mathfrak R_{[\bm{\nu}]}}
\left|
\left\langle j\right|
\operatorname{Tr}_2(Z_Y)
\left|j'\right\rangle
\right|^2
\leq
\left\| \operatorname{Tr}_2(Z_Y) \right\|_F^2
\leq
d\|Z_Y\|_F^2 \nonumber \\
&\leq
d\varepsilon^2
\left\| Y\otimes\frac{I}{d} \right\|_F^2
=
d\varepsilon^2 \|Y\|_F^2 \left\|\frac{I}{d}\right\|_F^2
=
\varepsilon^2\|Y\|_F^2,
\end{align}
where we used $\|\operatorname{Tr}_2(Z_Y)\|_F\leq\sqrt d\|Z_Y\|_F$ (this formula is shown in \cite[Proposition 1]{Rastegin2012}) and $\|I/d\|_F=1/\sqrt d$. This proves Eq.~\eqref{eq:orbit_partial_trace_bound_M1}, and Eq.~\eqref{eq:orbit_partial_trace_bound_M2} follows identically.
\end{proof}

\begin{lemma}[Haar pairings vanish on incompatible frequencies]
\label{lemma:haar_pairing_zero_frequency}
Recall the ideal Haar operators from Definition~\ref{def:moment_operators}. If $\bm{\omega}\neq\bm{\omega}'$, $\bm{\omega}\neq0$, and $\bm{\omega}'\neq0$, then
\begin{equation}
A^{\mathrm H}_{j,j';k,k'} = B^{\mathrm H}_{j,j';k,k'} = 0
\end{equation}
for every $(j,j')\in\mathfrak R(\bm{\omega})$ and $(k,k')\in\mathfrak R(\bm{\omega}')$. Furthermore, for the first-moment Haar quantities, if $\bm{\omega}\neq0$, then
\begin{equation}
\bar A^{\mathrm H}_{j,j'} = \bar B^{\mathrm H}_{j,j'} = 0
\end{equation}
for every $(j,j')\in\mathfrak R(\bm{\omega})$.
\end{lemma}

\begin{proof}
By the symmetric-subspace formula for the Haar second moment in \citet[Theorem~22]{Mele2024},
\begin{equation}
\mathcal H^{(2)}(|0,0\rangle\langle0,0|) = \frac{1}{d(d+1)}(I+F),
\end{equation}
where $I$ is the identity matrix and \(F=\sum_{i, j=1}^d|i, j\rangle\langle j, i|\) is the swap operator. Substituting this into the definition of $A^{\mathrm H}_{j,j';k,k'}=\left\langle j,k'\right| \mathcal H^{(2)} \left(|0,0\rangle\langle0,0| \right) \left|j',k\right\rangle$ yields:
\begin{equation}
A^{\mathrm H}_{j,j';k,k'}
=
\frac{1}{d(d+1)}
\left\langle j,k'\right| (I+F) \left|j',k\right\rangle
=
\frac{1}{d(d+1)}
\left( \delta_{j,j'}\delta_{k',k} + \delta_{j,k}\delta_{k',j'} \right).
\label{eq:AH_pairing_compact}
\end{equation}
The first pairing $\delta_{j,j'}\delta_{k',k}$ implies $\bm{\omega}=0$ and $\bm{\omega}'=0$, while the second pairing $\delta_{j,k}\delta_{k',j'}$ implies $\bm{\omega}' = \bm{\lambda}_k-\bm{\lambda}_{k'} = \bm{\lambda}_j-\bm{\lambda}_{j'} = \bm{\omega}$. Both conditions contradict our assumptions; therefore $A^{\mathrm H}_{j,j';k,k'}=0$. Similarly, by the Haar second-moment twirling formula in \citet[Corollary~13]{Mele2024}, there exist constants $\alpha,\beta$ such that
\begin{equation}
\mathcal H^{(2)}(\hat O\otimes\hat O) = \alpha I\otimes I+\beta F,
\end{equation}
where
\begin{equation}
\label{eq:decomposition_coefficient}
\alpha=\frac{d\operatorname{Tr}(\hat{O})^2- \|\hat{O}\|_F^2}{d(d^2-1)} \; \text{and} \; \beta=\frac{d\|\hat{O}\|_F^2-(\operatorname{Tr} \hat{O})^2}{d\left(d^2-1\right)}.
\end{equation}
Applying this to $B^{\mathrm H}_{j,j';k,k'}=\left\langle j',k\right| \mathcal H^{(2)} \left(\hat O\otimes\hat O \right) \left|j,k'\right\rangle$ gives:
\begin{equation}
B^{\mathrm H}_{j,j';k,k'}
=
\left\langle j',k\right| (\alpha I\otimes I+\beta F) \left|j,k'\right\rangle
=
\alpha\delta_{j',j}\delta_{k,k'} + \beta\delta_{j',k'}\delta_{k,j}.
\label{eq:BH_pairing_compact}
\end{equation}
The first pairing again forces both frequencies to be zero, and the second forces $\bm{\omega}'=\bm{\omega}$, so $B^{\mathrm H}_{j,j';k,k'}=0$. For the first moments, standard Haar integration gives $\mathbb{E}_{U\sim\mu_H} [ U A U^\dagger ] = \frac{\operatorname{Tr}(A)}{d}I$ \cite{Roberts2017}. Applying this
to $\bar A^{\mathrm H}_{j,j'}=
\mathbb{E}_{U\sim\mu_H}
\left[
\left\langle j\right|
U|0\rangle\langle0|U^\dagger
\left|j'\right\rangle
\right]
\nonumber,
\bar B^{\mathrm H}_{j,j'}
=
\mathbb{E}_{U\sim\mu_H} \left[ \left\langle j'\right| U^\dagger\hat O U \left|j\right\rangle \right]$ yields:
\begin{equation}
\bar A^{\mathrm H}_{j,j'} = \frac{\delta_{j,j'}}{d}, \qquad \bar B^{\mathrm H}_{j,j'} = \frac{\operatorname{Tr}(\hat O)}{d}\delta_{j,j'}.
\end{equation}
For any $(j,j')\in\mathfrak R(\bm{\omega})$ with $\bm{\omega}\neq0$, the diagonal pairing $j=j'$ is impossible since $\bm{\lambda}_j-\bm{\lambda}_{j'}\neq 0$. Thus $\bar A^{\mathrm H}_{j,j'}=\bar B^{\mathrm H}_{j,j'}=0$.
\end{proof}

\begin{lemma}[Nonzero-frequency and Orbit First Moment Bound]
\label{lemma:nonzero_frequency_first_moment_bound}
For every nonzero frequency $\bm{\omega}\neq0$, the expectation of the individual Fourier coefficient satisfies:
\begin{equation}
\left| \mathbb{E}_{\bm{\theta}} \left[ c_{\bm{\omega}} \right] \right| \leq \varepsilon^2\|\hat O\|_F.
\end{equation}
Furthermore, for any frequency orbit $[\bm{\nu}]$ containing only nonzero frequencies, the expectation of the symmetrized coefficient $a_{\bm{\nu}} = \sum_{\bm{\omega}\in[\bm{\nu}]} c_{\bm{\omega}}$ satisfies the same bound:
\begin{equation}
\left| \mathbb{E}_{\bm{\theta}} \left[ a_{\bm{\nu}} \right] \right| \leq \varepsilon^2\|\hat O\|_F.
\end{equation}
\end{lemma}

\begin{proof}
By Lemma~\ref{lemma:coefficient_moment_relations}, the expectation of an individual coefficient is $\mathbb{E}_{\bm{\theta}} [ c_{\bm{\omega}} ] = \sum_{(j,j')\in\mathfrak R(\bm{\omega})} \bar A_{j,j'}\bar B_{j,j'}$. Expanding the operators into their Haar and deviation components (Definition~\ref{def:moment_operators}) gives $\bar A_{j,j'} = \bar A^{\mathrm H}_{j,j'} + \Delta\bar A_{j,j'}$ and $\bar B_{j,j'} = \bar B^{\mathrm H}_{j,j'} + \Delta\bar B_{j,j'}$. By Lemma~\ref{lemma:haar_pairing_zero_frequency}, since $\bm{\omega}\neq0$, the ideal Haar components vanish exactly ($\bar A^{\mathrm H}_{j,j'} = \bar B^{\mathrm H}_{j,j'} = 0$). Thus, the expectation reduces entirely to the deviations:
\begin{equation}
\mathbb{E}_{\bm{\theta}} [ c_{\bm{\omega}} ]
=
\sum_{(j,j')\in\mathfrak R(\bm{\omega})} \Delta\bar A_{j,j'} \Delta\bar B_{j,j'}.
\end{equation}
Similarly, for the symmetrized coefficient $a_{\bm{\nu}} = \sum_{\bm{\omega}\in[\bm{\nu}]} c_{\bm{\omega}}$, we can sum the expectations over all frequencies in the orbit. Because the frequency generators $\mathfrak{R}(\bm{\omega})$ partition the index space into disjoint sets, the disjoint union $\bigcup_{\bm{\omega}\in[\bm{\nu}]} \mathfrak{R}(\bm{\omega}) = \mathfrak{R}_{[\bm{\nu}]}$ allows us to merge the sums compactly into the orbit frequency generator:
\begin{equation}
\mathbb{E}_{\bm{\theta}} [ a_{\bm{\nu}} ]
=
\sum_{\bm{\omega}\in[\bm{\nu}]} \sum_{(j,j')\in\mathfrak R(\bm{\omega})} \Delta\bar A_{j,j'} \Delta\bar B_{j,j'}
=
\sum_{(j,j')\in\mathfrak R_{[\bm{\nu}]}} \Delta\bar A_{j,j'} \Delta\bar B_{j,j'}.
\end{equation}
Applying the Cauchy--Schwarz inequality to both expressions yields their respective upper bounds:
\begin{align}
\left| \mathbb{E}_{\bm{\theta}} [ c_{\bm{\omega}} ] \right|
&\leq
\left( \sum_{(j,j')\in\mathfrak R(\bm{\omega})} |\Delta\bar A_{j,j'}|^2 \right)^{1/2}
\left( \sum_{(j,j')\in\mathfrak R(\bm{\omega})} |\Delta\bar B_{j,j'}|^2 \right)^{1/2}, \label{eq:cs_bound_single}\\
\left| \mathbb{E}_{\bm{\theta}} [ a_{\bm{\nu}} ] \right|
&\leq
\left( \sum_{(j,j')\in\mathfrak R_{[\bm{\nu}]}} |\Delta\bar A_{j,j'}|^2 \right)^{1/2}
\left( \sum_{(j,j')\in\mathfrak R_{[\bm{\nu}]}} |\Delta\bar B_{j,j'}|^2 \right)^{1/2}. \label{eq:cs_bound_orbit}
\end{align}
Recall the explicit partial-trace forms of the deviation operators from Definition~\ref{def:moment_operators}:
\begin{align*}
\Delta\bar A_{j,j'} &= \left\langle j\right| \operatorname{Tr}_2 \left[ \left( \mathcal M^{(2)}_1-\mathcal H^{(2)} \right) \left( |0\rangle\langle0|\otimes\frac{I}{d} \right) \right] \left|j'\right\rangle, \\
\Delta\bar B_{j,j'} &= \left\langle j'\right| \operatorname{Tr}_2 \left[ \left( \widetilde{\mathcal M}^{(2)}_2-\mathcal H^{(2)} \right) \left( \hat O\otimes\frac{I}{d} \right) \right] \left|j\right\rangle.
\end{align*}
These structures match the restricted Frobenius bounds derived in Lemma~\ref{lemma:restricted_frobenius_from_design} (for the individual frequency $\mathfrak R(\bm{\omega})$) and Lemma~\ref{lemma:orbit_restricted_frobenius} (for the orbit $\mathfrak R_{[\bm{\nu}]}$).

For the first factor in both Eq.~\eqref{eq:cs_bound_single} and Eq.~\eqref{eq:cs_bound_orbit}, we apply the respective lemmas with $Y = |0\rangle\langle0|$. Since the Frobenius norm of a rank-1 pure state projector is exactly $\||0\rangle\langle0|\|_F = 1$, the bounds evaluate to:
\begin{align}
\left( \sum_{(j,j')\in\mathfrak R(\bm{\omega})} |\Delta\bar A_{j,j'}|^2 \right)^{1/2} &\leq \varepsilon \||0\rangle\langle0|\|_F = \varepsilon, \\
\left( \sum_{(j,j')\in\mathfrak R_{[\bm{\nu}]}} |\Delta\bar A_{j,j'}|^2 \right)^{1/2} &\leq \varepsilon \||0\rangle\langle0|\|_F = \varepsilon.
\end{align}
Similarly, for the second factor in both inequalities, we apply the respective lemmas using the observable $Y = \hat O$. This directly yields:
\begin{align}
\left( \sum_{(j,j')\in\mathfrak R(\bm{\omega})} |\Delta\bar B_{j,j'}|^2 \right)^{1/2} &\leq \varepsilon \|\hat O\|_F, \\
\left( \sum_{(j,j')\in\mathfrak R_{[\bm{\nu}]}} |\Delta\bar B_{j,j'}|^2 \right)^{1/2} &\leq \varepsilon \|\hat O\|_F.
\end{align}
Multiplying these respective factors together immediately produces the desired bounds:
\begin{align}
\left| \mathbb{E}_{\bm{\theta}} [ c_{\bm{\omega}} ] \right| &\leq \varepsilon^2\|\hat O\|_F, \\
\left| \mathbb{E}_{\bm{\theta}} [ a_{\bm{\nu}} ] \right| &\leq \varepsilon^2\|\hat O\|_F.
\end{align}
\end{proof}


\begin{lemma}[Approximate Fourier Coefficient Decoupling]
\label{lemma:approx_fourier_coefficient_decoupling}
For any two distinct nonzero frequencies $\bm{\omega}\neq\bm{\omega}'$ ($\bm{\omega}\neq0, \bm{\omega}'\neq0$), we have
\begin{equation}
\label{eq:approx_cross_moment_final}
\left| \mathbb{E}_{\bm{\theta}} \left[ c_{\bm{\omega}}c_{\bm{\omega}'}^* \right] \right| \leq \varepsilon^2 \|\hat O\|_F^2.
\end{equation}
\end{lemma}

\begin{proof}
By Lemma~\ref{lemma:coefficient_moment_relations}, the mixed second moment is $\mathbb{E}_{\bm{\theta}} [ c_{\bm{\omega}}c_{\bm{\omega}'}^* ] = \sum_{(j,j')\in\mathfrak R(\bm{\omega})} \sum_{(k,k')\in\mathfrak R(\bm{\omega}')}  A_{j,j';k,k'} B_{j,j';k,k'}$. Expanding the second-moment operators into their Haar and deviation components (Definition~\ref{def:moment_operators}) gives:
\begin{align}
\mathbb{E}_{\bm{\theta}} \left[ c_{\bm{\omega}}c_{\bm{\omega}'}^* \right]
&=
\sum_{(j,j')\in\mathfrak R(\bm{\omega})} \sum_{(k,k')\in\mathfrak R(\bm{\omega}')}  A^{\mathrm H}_{j,j';k,k'} B^{\mathrm H}_{j,j';k,k'}
+ \sum_{(j,j')\in\mathfrak R(\bm{\omega})} \sum_{(k,k')\in\mathfrak R(\bm{\omega}')}  \Delta A_{j,j';k,k'} B^{\mathrm H}_{j,j';k,k'} \nonumber\\
&\quad + \sum_{(j,j')\in\mathfrak R(\bm{\omega})} \sum_{(k,k')\in\mathfrak R(\bm{\omega}')}  A^{\mathrm H}_{j,j';k,k'} \Delta B_{j,j';k,k'}
+ \sum_{(j,j')\in\mathfrak R(\bm{\omega})} \sum_{(k,k')\in\mathfrak R(\bm{\omega}')}  \Delta A_{j,j';k,k'} \Delta B_{j,j';k,k'}.
\label{eq:AB_expansion_short}
\end{align}
By Lemma~\ref{lemma:haar_pairing_zero_frequency}, since the frequencies are distinct and nonzero, the Haar pairings vanish pointwise ($A^{\mathrm H}_{j,j';k,k'} = B^{\mathrm H}_{j,j';k,k'} = 0$). Consequently, the first three terms in the expansion evaluate to zero, leaving only the pure deviation term:
\begin{equation}
\mathbb{E}_{\bm{\theta}} \left[ c_{\bm{\omega}}c_{\bm{\omega}'}^* \right]
=
\sum_{(j,j')\in\mathfrak R(\bm{\omega})} \sum_{(k,k')\in\mathfrak R(\bm{\omega}')}  \Delta A_{j,j';k,k'} \Delta B_{j,j';k,k'}.
\end{equation}
Taking the absolute value of both sides and applying the Cauchy--Schwarz inequality, we can bound this double sum by the product of their restricted Frobenius norms:
\begin{equation}
\left| \mathbb{E}_{\bm{\theta}} \left[ c_{\bm{\omega}}c_{\bm{\omega}'}^* \right] \right|
\leq
\|\Delta A\|_{F;\bm{\omega},\bm{\omega}'} \|\Delta B\|_{F;\bm{\omega},\bm{\omega}'},
\end{equation}
where the restricted Frobenius norms over the given frequency generators are explicitly defined as:
\begin{align}
\|\Delta A\|_{F;\bm{\omega},\bm{\omega}'}
&:=
\left( \sum_{(j,j')\in\mathfrak R(\bm{\omega})} \sum_{(k,k')\in\mathfrak R(\bm{\omega}')} |\Delta A_{j,j';k,k'}|^2 \right)^{1/2}, \\
\|\Delta B\|_{F;\bm{\omega},\bm{\omega}'}
&:=
\left( \sum_{(j,j')\in\mathfrak R(\bm{\omega})} \sum_{(k,k')\in\mathfrak R(\bm{\omega}')} |\Delta B_{j,j';k,k'}|^2 \right)^{1/2}.
\end{align}
To evaluate these restricted norms, we recall the exact channel expressions of the deviation operators from Definition~\ref{def:moment_operators}:
\begin{align*}
\Delta A_{j,j';k,k'} &= \left\langle j,k'\right| \left( \mathcal M^{(2)}_1-\mathcal H^{(2)} \right) \left( |0,0\rangle\langle0,0| \right) \left|j',k\right\rangle, \\
\Delta B_{j,j';k,k'} &= \left\langle j',k\right| \left( \widetilde{\mathcal M}^{(2)}_2-\mathcal H^{(2)} \right) \left( \hat O\otimes\hat O \right) \left|j,k'\right\rangle.
\end{align*}
These structures directly match the restricted bounding conditions of Lemma~\ref{lemma:restricted_frobenius_from_design}. For the first term, we apply the lemma with the input operator $X=|0,0\rangle\langle0,0|$. Since the Frobenius norm of a normalized pure state is $1$, this yields:
\begin{equation}
\|\Delta A\|_{F;\bm{\omega},\bm{\omega}'} \leq \varepsilon \| |0,0\rangle\langle0,0| \|_F = \varepsilon.
\end{equation}
For the second term, we apply the lemma with the input operator $X=\hat O\otimes\hat O$. Using the multiplicative property of the Frobenius norm under tensor products ($\|M \otimes N\|_F = \|M\|_F \|N\|_F$), we get:
\begin{equation}
\|\Delta B\|_{F;\bm{\omega},\bm{\omega}'} \leq \varepsilon \|\hat O\otimes\hat O\|_F = \varepsilon \|\hat O\|_F^2.
\end{equation}
Finally, multiplying these two derived bounds together confirms the original claim:
\begin{equation}
\left| \mathbb{E}_{\bm{\theta}} \left[ c_{\bm{\omega}}c_{\bm{\omega}'}^* \right] \right|
\leq
\varepsilon \cdot \left( \varepsilon\|\hat O\|_F^2 \right)
=
\varepsilon^2 \|\hat O\|_F^2.
\end{equation}

\end{proof}

\begin{corollary}[Approximate Fourier Coefficient Covariance Decoupling]
\label{cor:approx_fourier_coefficient_covariance_decoupling}
Under the assumptions of Lemma~\ref{lemma:approx_fourier_coefficient_decoupling}, for any two distinct nonzero frequencies $\bm{\omega}\neq\bm{\omega}'$, we have
\begin{equation}
\left| \operatorname{Cov}_{\bm{\theta}} \left( c_{\bm{\omega}},c_{\bm{\omega}'} \right) \right|
\leq
\|\hat O\|_F^2 \left( \varepsilon^2+\varepsilon^4 \right).
\end{equation}
\end{corollary}

\begin{proof}
Let $\mu_{\bm{\omega}} := \mathbb{E}_{\bm{\theta}} [ c_{\bm{\omega}} ]$ and $\mu_{\bm{\omega}'} := \mathbb{E}_{\bm{\theta}} [ c_{\bm{\omega}'} ]$. By Lemma~\ref{lemma:nonzero_frequency_first_moment_bound}, we know $|\mu_{\bm{\omega}}| \leq \varepsilon^2\|\hat O\|_F$ and $|\mu_{\bm{\omega}'}| \leq \varepsilon^2\|\hat O\|_F$. Furthermore, Lemma~\ref{lemma:approx_fourier_coefficient_decoupling} gives $\left| \mathbb{E}_{\bm{\theta}} [ c_{\bm{\omega}}c_{\bm{\omega}'}^* ] \right| \leq \varepsilon^2\|\hat O\|_F^2$. Since the covariance is defined as $\operatorname{Cov}_{\bm{\theta}} ( c_{\bm{\omega}},c_{\bm{\omega}'} ) = \mathbb{E}_{\bm{\theta}} [ c_{\bm{\omega}}c_{\bm{\omega}'}^* ] - \mu_{\bm{\omega}}\mu_{\bm{\omega}'}^*$, applying the triangle inequality yields:
\begin{equation}
\left| \operatorname{Cov}_{\bm{\theta}} ( c_{\bm{\omega}},c_{\bm{\omega}'} ) \right|
\leq
\left| \mathbb{E}_{\bm{\theta}} [ c_{\bm{\omega}}c_{\bm{\omega}'}^* ] \right| + |\mu_{\bm{\omega}}| |\mu_{\bm{\omega}'}|
\leq
\varepsilon^2\|\hat O\|_F^2 + \varepsilon^4\|\hat O\|_F^2,
\end{equation}
which proves the claim.
\end{proof}

\subsection{Symmetrized coefficient variance}
\label{app:tight_symmetrized_variance}

Recall that the total frequency redundancy of the orbit is denoted by $R = |\mathfrak{R}_{[\bm{\nu}]}| = \sum_{\bm{\omega}\in[\bm{\nu}]} |\mathfrak{R}(\bm{\omega})|$. To rigorously bound the deviation over this set, we define the restricted Frobenius norm of an operator $M$ on the orbit subspace $\mathfrak{R}_{[\bm{\nu}]}$ as
\begin{equation}
\| M |_{\mathfrak{R}_{[\bm{\nu}]}} \|_F := \left(\sum_{(j,j') \in \mathfrak{R}_{[\bm{\nu}]}} \sum_{(k,k') \in \mathfrak{R}_{[\bm{\nu}]}} |M_{j,j';k,k'}|^2\right)^{1/2}.
\end{equation}

\begin{lemma}[Restricted Haar Frobenius Norms]
\label{lemma:restricted_haar_norm}
Let $[\bm{\nu}]$ be a frequency orbit containing only nonzero frequencies. For the exact Haar operators $A^{\mathrm{H}}$ and $B^{\mathrm{H}}$ (Definition~\ref{def:moment_operators}), we have:
\begin{equation}
\left\| A^{\mathrm{H}} |_{\mathfrak{R}_{[\bm{\nu}]}} \right\|_F = \frac{\sqrt{R}}{d(d+1)}, \qquad
\left\| B^{\mathrm{H}} |_{\mathfrak{R}_{[\bm{\nu}]}} \right\|_F = \frac{ \sqrt{R}(d\|\hat{O}\|_F^2-(\operatorname{Tr} \hat{O})^2)}{d\left(d^2-1\right)}.
\end{equation}
\end{lemma}

\begin{proof}
Recall from Eq.~\eqref{eq:AH_pairing_compact} that $A^{\mathrm H}_{j,j';k,k'} = \frac{1}{d(d+1)} ( \delta_{j,j'}\delta_{k',k} + \delta_{j,k}\delta_{k',j'} )$. Since $[\bm{\nu}]$ excludes the zero frequency, for any $(j,j') \in \mathfrak{R}_{[\bm{\nu}]}$, we have $\bm{\lambda}_j \neq \bm{\lambda}_{j'}$, which implies $\delta_{j,j'} = 0$ by the contrapositive of $j=j'\Rightarrow\bm{\lambda}_j =\bm{\lambda}_{j'}$. Consequently, the first pairing vanishes, and the operator becomes diagonal with respect to the index pairs:
\begin{equation}
A^{\mathrm H}_{j,j';k,k'} = \frac{1}{d(d+1)} \delta_{j,k}\delta_{j',k'}.
\end{equation}
Squaring this term and summing over $(j,j') \in \mathfrak{R}_{[\bm{\nu}]}$ and $(k,k') \in \mathfrak{R}_{[\bm{\nu}]}$ yields exactly $R$ non-zero elements (where $j=k$ and $j'=k'$):
\begin{equation}
\left\| A^{\mathrm{H}} |_{\mathfrak{R}_{[\bm{\nu}]}} \right\|_F = \left( \sum_{(j,j') \in \mathfrak{R}_{[\bm{\nu}]}} \left(\frac{1}{d(d+1)}\right)^2 \right)^{1/2} = \frac{\sqrt{R}}{d(d+1)}.
\end{equation}
Similarly, Eq.~\eqref{eq:BH_pairing_compact} gives $B^{\mathrm H}_{j,j';k,k'} = \alpha\delta_{j',j}\delta_{k,k'} + \beta\delta_{j',k'}\delta_{k,j}$. For the same non-zero frequency condition, $\delta_{j',j} = 0$, leaving
\[
B^{\mathrm H}_{j,j';k,k'} = \beta \delta_{j',k'}\delta_{k,j}
\]
Using the coefficient $\beta=\frac{d\|\hat{O}\|_F^2-(\operatorname{Tr} \hat{O})^2}{d\left(d^2-1\right)}$ from Eq.~\eqref{eq:decomposition_coefficient} and the fact that $\beta\geq0$ by the Cauchy--Schwarz inequality, we obtain the restricted Frobenius norm:
\begin{equation}
\left\| B^{\mathrm{H}} |_{\mathfrak{R}_{[\bm{\nu}]}} \right\|_F = \left( \sum_{(j,j') \in \mathfrak{R}_{[\bm{\nu}]}} |\beta|^2 \right)^{1/2} = \sqrt{R} |\beta|=\frac{ \sqrt{R}(d\|\hat{O}\|_F^2-(\operatorname{Tr} \hat{O})^2)}{d\left(d^2-1\right)}.
\end{equation}

\end{proof}

\begin{lemma}[Restricted Error Frobenius Norms]
\label{lemma:restricted_error_norm}
Under the $\varepsilon$-approximate 2-design assumptions defined in Eq.~\eqref{eq:approx_2_design_channel_assumption}, the deviation operators $\Delta A$ and $\Delta B$ (Definition~\ref{def:moment_operators}) satisfy:
\begin{equation}
\left\| \Delta A |_{\mathfrak{R}_{[\bm{\nu}]}} \right\|_F \leq \varepsilon, \qquad \left\| \Delta B |_{\mathfrak{R}_{[\bm{\nu}]}} \right\|_F \leq \varepsilon \|\hat{O}\|_F^2.
\end{equation}
\end{lemma}

\begin{proof}
By definition, the restricted Frobenius norm over a subset of index pairs $\mathfrak{R}_{[\bm{\nu}]}$ is a partial sum of non-negative absolute square terms. It is bounded by the full Frobenius norm over the entire Hilbert space:
\begin{equation}
\left\| \Delta A |_{\mathfrak{R}_{[\bm{\nu}]}} \right\|_F = \left(\sum_{(j,j') \in \mathfrak{R}_{[\bm{\nu}]}} \sum_{(k,k') \in \mathfrak{R}_{[\bm{\nu}]}} |\Delta A _{j,j';k,k'}|^2\right)^{1/2}\leq \left(\sum_{\text{all } j,j',k,k'} |\Delta A _{j,j';k,k'}|^2\right)^{1/2} = \|\Delta A\|_F,
\end{equation}
and identically, $\left\| \Delta B |_{\mathfrak{R}_{[\bm{\nu}]}} \right\|_F \leq \|\Delta B\|_F$. We recall $\Delta A$ and $\Delta B$ from Definition~\ref{def:moment_operators}:
\begin{align}
\Delta A &= (\mathcal{M}_1^{(2)} - \mathcal{H}^{(2)})(|0,0\rangle\langle0,0|), \\
\Delta B &= (\widetilde{\mathcal{M}}_2^{(2)} - \mathcal{H}^{(2)})(\hat{O} \otimes \hat{O}).
\end{align}
By Lemma~\ref{lemma:superoperator_norm}, any superoperator $\mathcal{E}$ acting on a matrix $X$ satisfies $\|\mathcal{E}(X)\|_F \leq \|\mathcal{E}\|_{\infty} \|X\|_F$. Applying the $\varepsilon$-approximate 2-design assumptions $\| \mathcal{M}^{(2)}_1 - \mathcal{H}^{(2)} \|_{\infty} \leq \varepsilon$ and $\| \widetilde{\mathcal{M}}^{(2)}_2 - \mathcal{H}^{(2)} \|_{\infty} \leq \varepsilon$, we obtain the strict bounds:
\begin{align}
\|\Delta A\|_F &\leq \varepsilon \| |0,0\rangle\langle0,0| \|_F, \\
\|\Delta B\|_F &\leq \varepsilon \| \hat{O} \otimes \hat{O} \|_F.
\end{align}
Since $\| |0,0\rangle\langle0,0| \|_F =1$ and $\| \hat{O} \otimes \hat{O} \|_F = \|\hat{O}\|_F^2$, combining these bounds completes the proof:
\begin{align}
\left\| \Delta A |_{\mathfrak{R}_{[\bm{\nu}]}} \right\|_F &\leq \|\Delta A\|_F \leq \varepsilon \cdot 1 = \varepsilon, \\
\left\| \Delta B |_{\mathfrak{R}_{[\bm{\nu}]}} \right\|_F &\leq \|\Delta B\|_F \leq \varepsilon \|\hat{O}\|_F^2.
\end{align}

\end{proof}

Now we present the exact and tight variance bound for the symmetrized coefficients.

\subsection{Variance bound for symmetrized coefficients}
\label{app:approximate_theorem}
\begin{restate}{Theorem}{theorem:tight_approx_exp_variance}
Consider the \(n\)-qubit single-uploading QFM in Eq.~\eqref{eq:QFM} with \(L=1\) and symmetrized coefficients \(a_{\bm{\nu}}(\bm{\theta})\) defined in Eq.~\eqref{eq:symmetrized_coefficient}. Suppose the two trainable layers form independent \(\varepsilon\)-approximate 2-designs in the sense of Eq.~\eqref{eq:approx_2_design_channel_assumption}. For any frequency orbit $[\bm{\nu}]$ excluding the zero frequency, let $R = |\mathfrak{R}_{[\bm{\nu}]}|$ be the total redundancy of the orbit. Then the deviation from the exact 2-design variance satisfies:
\begin{equation}
\left|
\operatorname{Var}_{\bm{\theta}}(a_{\bm{\nu}})
-
\operatorname{Var}_{\text{2-design}}(a_{\bm{\nu}})
\right|
\in
\mathcal{O}\left( \varepsilon \|\hat{O}\|_F^2 \left( \frac{\sqrt{R}}{4^n} + \varepsilon \right) \right).
\end{equation}
\end{restate}

\begin{proof}
By Lemma~\ref{lemma:coefficient_moment_relations}, the aggregated second moment expands over the orbit as $\mathbb{E}_{\bm{\theta}} \left[ |a_{\bm{\nu}}|^2 \right] = \sum_{(j,j'), (k,k') \in \mathfrak{R}_{[\bm{\nu}]}} A_{j,j';k,k'} B_{j,j';k,k'}$. By substituting the second-moment operators with their exact Haar components and deviation components, $A = A^{\mathrm{H}} + \Delta A$ and $B = B^{\mathrm{H}} + \Delta B$, we can expand this product into four separate summation terms over all valid index pairs $(j,j'), (k,k') \in \mathfrak{R}_{[\bm{\nu}]}$:
\begin{align}
\label{eq:four_terms_expansion}
\mathbb{E}_{\bm{\theta}} \left[ |a_{\bm{\nu}}|^2 \right]
&= \sum_{(j,j'), (k,k') \in \mathfrak{R}_{[\bm{\nu}]}} A^{\mathrm{H}}_{j,j';k,k'} B^{\mathrm{H}}_{j,j';k,k'} \nonumber \\
&\quad + \sum_{(j,j'), (k,k') \in \mathfrak{R}_{[\bm{\nu}]}} A^{\mathrm{H}}_{j,j';k,k'} \Delta B_{j,j';k,k'} \nonumber \\
&\quad + \sum_{(j,j'), (k,k') \in \mathfrak{R}_{[\bm{\nu}]}} \Delta A_{j,j';k,k'} B^{\mathrm{H}}_{j,j';k,k'} \nonumber \\
&\quad + \sum_{(j,j'), (k,k') \in \mathfrak{R}_{[\bm{\nu}]}} \Delta A_{j,j';k,k'} \Delta B_{j,j';k,k'}.
\end{align}

For the first term, applying the relations from Lemma~\ref{lemma:coefficient_moment_relations} in reverse under the exact Haar measure yields the sum of the ideal expected squared coefficients:
\begin{equation}
\sum_{(j,j'), (k,k') \in \mathfrak{R}_{[\bm{\nu}]}} A^{\mathrm{H}}_{j,j';k,k'} B^{\mathrm{H}}_{j,j';k,k'}
=
\sum_{\bm{\omega}\in[\bm{\nu}]} \mathbb{E}_{\mathcal{H}}\left[|c_{\bm{\omega}}|^2\right].
\end{equation}
Since it is a mathematical property of the exact Haar measure that the expectation of any non-zero frequency coefficient is zero ($\mathbb{E}_{\mathcal{H}}[c_{\bm{\omega}}] = 0$), we have $\mathbb{E}_{\mathcal{H}}[|c_{\bm{\omega}}|^2] = \operatorname{Var}_{\mathcal{H}}(c_{\bm{\omega}})$. Thus, the first term identically evaluates to the sum of the ideal Haar variances:
\begin{equation}
\sum_{(j,j'), (k,k') \in \mathfrak{R}_{[\bm{\nu}]}} A^{\mathrm{H}}_{j,j';k,k'} B^{\mathrm{H}}_{j,j';k,k'}
=
\sum_{\bm{\omega}\in[\bm{\nu}]} \operatorname{Var}_{\mathcal{H}}(c_{\bm{\omega}}).
\end{equation}
For the remaining three cross-deviation terms, we bound their absolute values using the Cauchy--Schwarz inequality. The absolute sum of the element-wise products over the restricted index set $\mathfrak{R}_{[\bm{\nu}]}$ is bounded by the product of their restricted Frobenius norms. Specifically, the first cross-term yields:
\begin{align}
\left| \sum_{(j,j'), (k,k') \in \mathfrak{R}_{[\bm{\nu}]}} A^{\mathrm{H}}_{j,j';k,k'} \Delta B_{j,j';k,k'} \right|
&\leq \left( \sum_{(j,j'), (k,k') \in \mathfrak{R}_{[\bm{\nu}]}} |A^{\mathrm{H}}_{j,j';k,k'}|^2 \right)^{1/2} \left( \sum_{(j,j'), (k,k') \in \mathfrak{R}_{[\bm{\nu}]}} |\Delta B_{j,j';k,k'}|^2 \right)^{1/2} \nonumber \\
&= \left\| A^{\mathrm{H}} |_{\mathfrak{R}_{[\bm{\nu}]}} \right\|_F \left\| \Delta B |_{\mathfrak{R}_{[\bm{\nu}]}} \right\|_F.
\end{align}
Applying the specific bounds from Lemma~\ref{lemma:restricted_haar_norm} and Lemma~\ref{lemma:restricted_error_norm} to these norms, we evaluate this upper bound:
\begin{equation}
\left\| A^{\mathrm{H}} |_{\mathfrak{R}_{[\bm{\nu}]}} \right\|_F \left\| \Delta B |_{\mathfrak{R}_{[\bm{\nu}]}} \right\|_F
\leq
\frac{\sqrt{R}}{d(d+1)} \left( \varepsilon \|\hat{O}\|_F^2 \right)
=
\varepsilon \frac{\sqrt{R}}{d(d+1)} \|\hat{O}\|_F^2.
\end{equation}
By applying the identical Cauchy--Schwarz bounding procedure to the other two components and retaining the exact tight bounds, we obtain:
\begin{align}
\left| \sum_{(j,j'), (k,k') \in \mathfrak{R}_{[\bm{\nu}]}} \Delta A_{j,j';k,k'} B^{\mathrm{H}}_{j,j';k,k'} \right| &\leq \left\| \Delta A |_{\mathfrak{R}_{[\bm{\nu}]}} \right\|_F \left\| B^{\mathrm{H}} |_{\mathfrak{R}_{[\bm{\nu}]}} \right\|_F \leq \frac{ \varepsilon \sqrt{R}(d\|\hat{O}\|_F^2-(\operatorname{Tr} \hat{O})^2)}{d\left(d^2-1\right)}, \\
\left| \sum_{(j,j'), (k,k') \in \mathfrak{R}_{[\bm{\nu}]}} \Delta A_{j,j';k,k'} \Delta B_{j,j';k,k'} \right| &\leq \left\| \Delta A |_{\mathfrak{R}_{[\bm{\nu}]}} \right\|_F \left\| \Delta B |_{\mathfrak{R}_{[\bm{\nu}]}} \right\|_F \leq \varepsilon \left( \varepsilon \|\hat{O}\|_F^2 \right) = \varepsilon^2 \|\hat{O}\|_F^2.
\end{align}
Defining $\text{Cross-Terms}$ as the sum of these three deviation components, we can combine their bounds using the triangle inequality:
\begin{equation}
|\text{Cross-Terms}| \leq \varepsilon \sqrt{R} \left( \frac{\|\hat{O}\|_F^2}{d(d+1)} + \frac{d\|\hat{O}\|_F^2-(\operatorname{Tr} \hat{O})^2}{d(d^2-1)} \right) + \varepsilon^2 \|\hat{O}\|_F^2.
\end{equation}
After factoring out the common denominator $d(d^2-1)$ and using the relation $\frac{1}{d(d+1)} = \frac{d-1}{d(d^2-1)}$, we obtain:
\begin{equation}
|\text{Cross-Terms}| \leq \varepsilon \frac{\sqrt{R}}{d(d^2-1)} \Big( (2d-1)\|\hat{O}\|_F^2 - (\operatorname{Tr} \hat{O})^2 \Big) + \varepsilon^2 \|\hat{O}\|_F^2.
\end{equation}
To connect this second-moment derivation to the variance of the physical model, we utilize the variance identity:
\begin{equation}
\mathbb{E}_{\bm{\theta}}\left[|a_{\bm{\nu}}|^2\right] = \operatorname{Var}_{\bm{\theta}}(a_{\bm{\nu}}) + |\mathbb{E}_{\bm{\theta}}[a_{\bm{\nu}}]|^2.
\end{equation}
Substituting this identity into our aggregated second moment expansion (Eq.~\ref{eq:four_terms_expansion}) gives:
\begin{equation}
\operatorname{Var}_{\bm{\theta}}(a_{\bm{\nu}}) + |\mathbb{E}_{\bm{\theta}}[a_{\bm{\nu}}]|^2 = \sum_{\bm{\omega}\in[\bm{\nu}]} \operatorname{Var}_{\mathcal{H}}(c_{\bm{\omega}}) + \text{Cross-Terms}.
\end{equation}
Isolating the difference between the physical variance and the ideal Haar variance sum yields:
\begin{equation}
\operatorname{Var}_{\bm{\theta}}(a_{\bm{\nu}}) - \sum_{\bm{\omega}\in[\bm{\nu}]} \operatorname{Var}_{\mathcal{H}}(c_{\bm{\omega}}) = \text{Cross-Terms} - |\mathbb{E}_{\bm{\theta}}[a_{\bm{\nu}}]|^2.
\end{equation}
Taking the absolute value of both sides and applying the triangle inequality ($|X - Y| \leq |X| + |Y|$) provides the structure of our final bound:
\begin{equation}
\left| \operatorname{Var}_{\bm{\theta}}(a_{\bm{\nu}}) - \sum_{\bm{\omega}\in[\bm{\nu}]} \operatorname{Var}_{\mathcal{H}}(c_{\bm{\omega}}) \right| \leq |\text{Cross-Terms}| + |\mathbb{E}_{\bm{\theta}}[a_{\bm{\nu}}]|^2.
\end{equation}
Finally, we evaluate the squared expectation shift $|\mathbb{E}_{\bm{\theta}}[a_{\bm{\nu}}]|^2$. By our tighter bound established in Lemma~\ref{lemma:nonzero_frequency_first_moment_bound}, the expectation of the symmetrized coefficient is bounded independently of the orbit size:
\begin{equation}
|\mathbb{E}_{\bm{\theta}}[a_{\bm{\nu}}]| \leq \varepsilon^2 \|\hat{O}\|_F.
\end{equation}
Squaring this result introduces a shift of at most $|\mathbb{E}_{\bm{\theta}}[a_{\bm{\nu}}]|^2 \leq \varepsilon^4 \|\hat{O}\|_F^2$. Substituting the derived upper bounds for both $|\text{Cross-Terms}|$ and $|\mathbb{E}_{\bm{\theta}}[a_{\bm{\nu}}]|^2$ into our absolute difference inequality produces the exact theoretical bound:
\begin{equation}
\label{eq:haar_exact_bound}
\left| \operatorname{Var}_{\bm{\theta}}(a_{\bm{\nu}}) - \sum_{\bm{\omega}\in[\bm{\nu}]} \operatorname{Var}_{\mathcal{H}}(c_{\bm{\omega}}) \right|
\leq \varepsilon \frac{\sqrt{R}}{d(d^2-1)} \Big( (2d-1)\|\hat{O}\|_F^2 - (\operatorname{Tr} \hat{O})^2 \Big) + (\varepsilon^2+ \varepsilon^4) \|\hat{O}\|_F^2
\end{equation}
Using Big-O notation, we simplify our result as follows:
\begin{align}
\left| \operatorname{Var}_{\bm{\theta}}(a_{\bm{\nu}}) - \sum_{\bm{\omega}\in[\bm{\nu}]} \operatorname{Var}_{\mathcal{H}}(c_{\bm{\omega}}) \right|
&\leq \varepsilon \frac{\sqrt{R} \cdot 2d}{d^3-d} \|\hat{O}\|_F^2 + \mathcal{O}(\varepsilon^2 \|\hat{O}\|_F^2) \notag \\
&= \mathcal{O}\left( \left( \frac{\varepsilon \sqrt{R}}{d^2} + \varepsilon^2 \right) \|\hat{O}\|_F^2 \right).
\end{align}
Substituting $d = 2^n$ turns the term $d^{-2}$ into $4^{-n}$, making the exponential decay with the number of qubits explicitly clear:
\begin{equation}
\left| \operatorname{Var}_{\bm{\theta}}(a_{\bm{\nu}}) - \sum_{\bm{\omega}\in[\bm{\nu}]} \operatorname{Var}_{\mathcal{H}}(c_{\bm{\omega}}) \right|
\in \mathcal{O}\left( \left( \frac{\varepsilon \sqrt{R}}{4^n} + \varepsilon^2 \right) \|\hat{O}\|_F^2 \right).
\end{equation}
From Theorem~\ref{theorem:exp_variance}, we have $\operatorname{Var}_{\text{2-design}}(a_{\bm{\nu}}) = \sum_{\bm{\omega} \in [\bm{\nu}]} \operatorname{Var}_{\text{2-design}}(c_{\bm{\omega}})$. Since $\operatorname{Var}_{\mathcal{H}}(c_{\bm{\omega}})=\operatorname{Var}_{\text{2-design}}(c_{\bm{\omega}})$, we finally obtain our simplified asymptotic bound:
\begin{equation}
\left|
\operatorname{Var}_{\bm{\theta}}(a_{\bm{\nu}})
-
\operatorname{Var}_{\text{2-design}}(a_{\bm{\nu}})
\right|
\in
\mathcal{O}\left( \varepsilon \|\hat{O}\|_F^2 \left( \frac{\sqrt{R}}{4^n} + \varepsilon \right) \right).
\end{equation}
This also gives the exact 2-design theoretical bound from Eq.~\eqref{eq:haar_exact_bound}:
\begin{equation}
\label{eq:exact_coeff_bound}
\left| \operatorname{Var}_{\bm{\theta}}(a_{\bm{\nu}}) - \sum_{\bm{\omega}\in[\bm{\nu}]} \operatorname{Var}_{\text{2-design}}(c_{\bm{\omega}}) \right|
\leq \varepsilon \frac{\sqrt{R}}{d(d^2-1)} \Big( (2d-1)\|\hat{O}\|_F^2 - (\operatorname{Tr} \hat{O})^2 \Big) + (\varepsilon^2+ \varepsilon^4) \|\hat{O}\|_F^2
\end{equation}

\end{proof}

\section{Sampling cost analysis}
\label{app:sampling_cost}

\subsection{Symmetry sampling shot complexity}
\label{app:randomized_encoding_shot_complexity}

\begin{restate}{Proposition}{prop:randomized_encoding_shot_complexity}
The estimator in Eq.~\eqref{eq:randomized_encoding_empirical_mean} is unbiased for \((\mathcal A f)(\bm{x},\bm{\theta})\). Its single shot variance is
\begin{equation}
\label{eq:app_randomized_encoding_variance}
\operatorname{Var}(Y_{\bm{x}})
=
\|h\|_1^2
-
\left((\mathcal A f)(\bm{x},\bm{\theta})\right)^2.
\end{equation}
Moreover, for any \(\epsilon>0\) and \(\delta\in(0,1)\),
\begin{equation}
\label{eq:app_randomized_encoding_hoeffding}
\Pr\!\left[
\left|
\widehat{\mathcal A f}_N(\bm{x},\bm{\theta})
-
(\mathcal A f)(\bm{x},\bm{\theta})
\right|
\leq
\epsilon
\right]
\geq
1-\delta
\end{equation}
whenever
\begin{equation}
\label{eq:app_randomized_encoding_sample_bound}
N
\geq
\frac{2\|h\|_1^2}{\epsilon^2}
\log\!\left(\frac{2}{\delta}\right).
\end{equation}
\end{restate}

\begin{proof}
By the law of total expectation, the conditional distribution in Eq.~\eqref{eq:randomized_encoding_single_shot}, and the definition of \(q_i\),
\begin{align}
\mathbb E[Y_{\bm{x}}]
&=
\frac{1}{|\mathcal S|}
\sum_{s\in\mathcal S}
\sum_{i\in\mathcal I}
q_i
\mathbb E
\left[
\frac{h_i}{q_i}
Z^{(i)}_{s,\bm{x}}
\right] \notag\\
&=
\frac{1}{|\mathcal S|}
\sum_{s\in\mathcal S}
\sum_{i\in\mathcal I}
h_i
\operatorname{Tr}
\left(
O^{(i)}
\rho_s(\bm{x},\bm{\theta})
\right) \notag\\
&=
\operatorname{Tr}
\left(
\hat O
\frac{1}{|\mathcal S|}
\sum_{s\in\mathcal S}
\rho_s(\bm{x},\bm{\theta})
\right)
=
(\mathcal A f)(\bm{x},\bm{\theta}).
\end{align}
Thus the empirical estimator in Eq.~\eqref{eq:randomized_encoding_empirical_mean} is unbiased.

The second moment is independent of \(s\), because \((Z^{(i)}_{s,\bm{x}})^2=1\). Hence
\begin{align}
\mathbb E[Y_{\bm{x}}^2]
&=
\frac{1}{|\mathcal S|}
\sum_{s\in\mathcal S}
\sum_{i\in\mathcal I}
q_i
\mathbb E
\left[
\left(
\frac{h_i}{q_i}
Z^{(i)}_{s,\bm{x}}
\right)^2
\right] \notag\\
&=
\frac{1}{|\mathcal S|}
\sum_{s\in\mathcal S}
\sum_{i\in\mathcal I}
q_i
\frac{h_i^2}{q_i^2}
=
\sum_{i\in\mathcal I}
\frac{h_i^2}{q_i} \notag\\
&=
\sum_{i\in\mathcal I}
|h_i|\|h\|_1
=
\|h\|_1^2.
\end{align}
Therefore
\begin{equation}
\operatorname{Var}(Y_{\bm{x}})
=
\mathbb E[Y_{\bm{x}}^2]
-
\mathbb E[Y_{\bm{x}}]^2
=
\|h\|_1^2
-
\left((\mathcal A f)(\bm{x},\bm{\theta})\right)^2.
\end{equation}
This proves Eq.~\eqref{eq:app_randomized_encoding_variance}.
Finally, since \(q_i=|h_i|/\|h\|_1\), the conditional random variable \(Y_{\bm{x}}\mid S=s,I=i\) takes values in \(\{-\|h\|_1,\|h\|_1\}\). Hence \(Y_{\bm{x}}\in[-\|h\|_1,\|h\|_1]\). Hoeffding's inequality for independent bounded random variables gives
\begin{equation}
\Pr\!\left[
\left|
\widehat{\mathcal A f}_N(\bm{x},\bm{\theta})
-
(\mathcal A f)(\bm{x},\bm{\theta})
\right|
\geq
\epsilon
\right]
\leq
2\exp\!\left(
-
\frac{N\epsilon^2}{2\|h\|_1^2}
\right).
\end{equation}
Solving the right-hand side for a failure probability at most \(\delta\) gives Eq.~\eqref{eq:app_randomized_encoding_sample_bound}. Substituting this bound into the preceding tail inequality gives Eq.~\eqref{eq:app_randomized_encoding_hoeffding}.
\end{proof}

\subsection{Variance comparison}
\label{app:variance_comparison}

\begin{restate}{Proposition}{prop:variance_comparison}
Let \(Y_{\bm{x}}\) be the symmetry sampling estimator in Eq.~\eqref{eq:randomized_encoding_single_shot}. Then \(Y_{\bm{x}}\) and the baseline estimator \(K_{\bm{x}}\) satisfy
\begin{equation}
\label{eq:app_variance_difference_sampling_cost}
\operatorname{Var}(Y_{\bm{x}})
-
\operatorname{Var}(K_{\bm{x}})
=
f(\bm{x},\bm{\theta})^2
-
\left((\mathcal A f)(\bm{x},\bm{\theta})\right)^2.
\end{equation}
Consequently, for a fixed input \(\bm{x}\),
\begin{equation}
\label{eq:app_variance_comparison_condition}
\operatorname{Var}(Y_{\bm{x}})
<
\operatorname{Var}(K_{\bm{x}})
\quad\Longleftrightarrow\quad
\left|
(\mathcal A f)(\bm{x},\bm{\theta})
\right|
>
\left|
f(\bm{x},\bm{\theta})
\right|.
\end{equation}
If \(\bm{x}\) is uniformly distributed over \(\mathcal X\) and the representation elements \(V_s\) are uniformity-preserving on \(\mathcal X\), then
\begin{equation}
\label{eq:app_average_variance_inflation}
\mathbb E_{\bm{x}}
\left[
\operatorname{Var}(Y_{\bm{x}})
\right]
\geq
\mathbb E_{\bm{x}}
\left[
\operatorname{Var}(K_{\bm{x}})
\right].
\end{equation}
\end{restate}

\begin{proof}
Combining Eq.~\eqref{eq:baseline_single_shot_variance} with Eq.~\eqref{eq:app_randomized_encoding_variance} gives Eq.~\eqref{eq:app_variance_difference_sampling_cost}.
Equation~\eqref{eq:app_variance_comparison_condition} follows immediately from this variance difference. Now assume that \(\bm{x}\) is uniformly distributed over \(\mathcal X\) and that the representation elements \(V_s\) are uniformity-preserving on \(\mathcal X\). Since \(v\mapsto v^2\) is convex,
\begin{equation}
\left((\mathcal A f)(\bm{x},\bm{\theta})\right)^2
=
\left(
\frac{1}{|\mathcal S|}
\sum_{s\in\mathcal S}
f(V_s[\bm{x}],\bm{\theta})
\right)^2
\leq
\frac{1}{|\mathcal S|}
\sum_{s\in\mathcal S}
f(V_s[\bm{x}],\bm{\theta})^2.
\end{equation}
Taking expectation over \(\bm{x}\) and applying the uniformity-preserving condition to the function \(\bm{x}\mapsto f(\bm{x},\bm{\theta})^2\) gives
\begin{equation}
\mathbb E_{\bm{x}}
\left[
\left((\mathcal A f)(\bm{x},\bm{\theta})\right)^2
\right]
\leq
\frac{1}{|\mathcal S|}
\sum_{s\in\mathcal S}
\mathbb E_{\bm{x}}
\left[
f(V_s[\bm{x}],\bm{\theta})^2
\right]
=
\mathbb E_{\bm{x}}
\left[
f(\bm{x},\bm{\theta})^2
\right].
\end{equation}
Substituting this inequality into the two variance formulas proves Eq.~\eqref{eq:app_average_variance_inflation}.
\end{proof}

\section{QPINN residual and loss gradient formulas}
\label{app:loss_gradient_expression}

This appendix rewrites the loss gradients in Sec.~\ref{sec:experiment} in the form used by the QPINN estimators. The superscripts \(\mathrm{SP}\) and \(\mathrm{HJ}\) distinguish the screened Poisson and Hamilton-Jacobi equations, while the subscripts \(\mathrm{soft}\) and \(\mathrm{hard}\) distinguish the ansatz used before evaluating the residual. The main text defines the pointwise residuals in terms of \(u_{\bm{\theta}}\). Here we substitute either the soft ansatz \(u_{\bm{\theta}}(\bm{x})=v_{\bm{\theta}}(\bm z(\bm{x}))\) or the hard ansatz \(u_{\bm{\theta}}(\bm{x})=B(\bm{x})v_{\bm{\theta}}(\bm z(\bm{x}))\). When \(v_{\bm{\theta}}\) and its derivatives appear without an argument in an expression indexed by \(\bm{x}_i\), they are evaluated at \(\bm z(\bm{x}_i)\). In a boundary term indexed by \(\bm{y}_j\), they are evaluated at \(\bm z(\bm{y}_j)\).

\subsection{Screened Poisson equation}
\label{app:loss_gradient_expression_screen}

For the screened Poisson equation, the residual for a model prediction \(u_{\bm{\theta}}\) is
\begin{equation}
r_{\bm{\theta}}^{\mathrm{SP}}(\bm{x})
=
-\sum_{d=1}^2
\frac{\partial^2 u_{\bm{\theta}}}{\partial x_d^2}(\bm{x})
+
\lambda u_{\bm{\theta}}(\bm{x})
-
q(\bm{x}).
\end{equation}
After substituting the soft and hard ansatzes, we write the resulting residuals as \(r_{\bm{\theta},\mathrm{soft}}^{\mathrm{SP}}\) and \(r_{\bm{\theta},\mathrm{hard}}^{\mathrm{SP}}\), respectively.
The corresponding loss-gradient identities are
\begin{align}
\frac{\partial \mathcal L_{\mathrm{soft}}^{\mathrm{SP}}}{\partial\theta_\ell}
&=
\frac{2}{N_\Omega}
\sum_{i=1}^{N_\Omega}
r_{\bm{\theta},\mathrm{soft}}^{\mathrm{SP}}(\bm{x}_i)
\frac{\partial r_{\bm{\theta},\mathrm{soft}}^{\mathrm{SP}}}{\partial\theta_\ell}(\bm{x}_i)
+
\frac{2\gamma}{N_{\partial\Omega}}
\sum_{j=1}^{N_{\partial\Omega}}
u_{\bm{\theta}}(\bm{y}_j)
\frac{\partial u_{\bm{\theta}}}{\partial\theta_\ell}(\bm{y}_j),
\nonumber\\
\frac{\partial \mathcal L_{\mathrm{hard}}^{\mathrm{SP}}}{\partial\theta_\ell}
&=
\frac{2}{N_\Omega}
\sum_{i=1}^{N_\Omega}
r_{\bm{\theta},\mathrm{hard}}^{\mathrm{SP}}(\bm{x}_i)
\frac{\partial r_{\bm{\theta},\mathrm{hard}}^{\mathrm{SP}}}{\partial\theta_\ell}(\bm{x}_i).
\end{align}
For the soft ansatz \(u_{\bm{\theta}}(\bm{x})=v_{\bm{\theta}}(\bm z(\bm{x}))\), the chain rule gives
\begin{align}
\frac{\partial u_{\bm{\theta}}}{\partial x_d}(\bm{x}_i)
&=
\frac{d z_d}{d x_d}(x_{i,d})
\frac{\partial v_{\bm{\theta}}}{\partial z_d},
\nonumber\\
\frac{\partial^2 u_{\bm{\theta}}}{\partial x_d^2}(\bm{x}_i)
&=
\frac{d^2 z_d}{d x_d^2}(x_{i,d})
\frac{\partial v_{\bm{\theta}}}{\partial z_d}
+
\left[
\frac{d z_d}{d x_d}(x_{i,d})
\right]^2
\frac{\partial^2 v_{\bm{\theta}}}{\partial z_d^2}.
\end{align}
Substituting these identities into the residual gives
\begin{align}
r_{\bm{\theta},\mathrm{soft}}^{\mathrm{SP}}(\bm{x}_i)
&=
\lambda v_{\bm{\theta}}
-
\sum_{d=1}^2
\frac{d^2 z_d}{d x_d^2}(x_{i,d})
\frac{\partial v_{\bm{\theta}}}{\partial z_d}
-
\sum_{d=1}^2
\left[
\frac{d z_d}{d x_d}(x_{i,d})
\right]^2
\frac{\partial^2 v_{\bm{\theta}}}{\partial z_d^2}
-
q(\bm{x}_i),
\nonumber\\
\frac{\partial r_{\bm{\theta},\mathrm{soft}}^{\mathrm{SP}}}{\partial\theta_\ell}(\bm{x}_i)
&=
\lambda
\frac{\partial v_{\bm{\theta}}}{\partial\theta_\ell}
-
\sum_{d=1}^2
\frac{d^2 z_d}{d x_d^2}(x_{i,d})
\frac{\partial^2 v_{\bm{\theta}}}{\partial\theta_\ell\partial z_d}
-
\sum_{d=1}^2
\left[
\frac{d z_d}{d x_d}(x_{i,d})
\right]^2
\frac{\partial^3 v_{\bm{\theta}}}{\partial\theta_\ell\partial z_d^2}.
\label{eq:app_poisson_soft_residual_gradient}
\end{align}
The soft boundary contribution is also a circuit quantity because
\begin{equation}
u_{\bm{\theta}}(\bm{y}_j)
=
v_{\bm{\theta}},
\quad
\frac{\partial u_{\bm{\theta}}}{\partial\theta_\ell}(\bm{y}_j)
=
\frac{\partial v_{\bm{\theta}}}{\partial\theta_\ell}.
\end{equation}
For the hard ansatz \(u_{\bm{\theta}}(\bm{x})=B(\bm{x})v_{\bm{\theta}}(\bm z(\bm{x}))\), the first and second spatial derivatives are
\begin{align}
\frac{\partial u_{\bm{\theta}}}{\partial x_d}(\bm{x}_i)
&=
\frac{\partial B}{\partial x_d}(\bm{x}_i)
v_{\bm{\theta}}
+
B(\bm{x}_i)
\frac{d z_d}{d x_d}(x_{i,d})
\frac{\partial v_{\bm{\theta}}}{\partial z_d},
\nonumber\\
\frac{\partial^2 u_{\bm{\theta}}}{\partial x_d^2}(\bm{x}_i)
&=
\frac{\partial^2 B}{\partial x_d^2}(\bm{x}_i)
v_{\bm{\theta}}
+
\left[
2
\frac{\partial B}{\partial x_d}(\bm{x}_i)
\frac{d z_d}{d x_d}(x_{i,d})
+
B(\bm{x}_i)
\frac{d^2 z_d}{d x_d^2}(x_{i,d})
\right]
\frac{\partial v_{\bm{\theta}}}{\partial z_d}
\nonumber\\
&\quad
+
B(\bm{x}_i)
\left[
\frac{d z_d}{d x_d}(x_{i,d})
\right]^2
\frac{\partial^2 v_{\bm{\theta}}}{\partial z_d^2}.
\end{align}
Substituting these derivatives into the screened Poisson residual gives
\begin{align}
r_{\bm{\theta},\mathrm{hard}}^{\mathrm{SP}}(\bm{x}_i)
&=
\left[
\lambda B(\bm{x}_i)
-
\sum_{d=1}^2
\frac{\partial^2 B}{\partial x_d^2}(\bm{x}_i)
\right]
v_{\bm{\theta}}
\nonumber\\
&\quad
-
\sum_{d=1}^2
\left[
2
\frac{\partial B}{\partial x_d}(\bm{x}_i)
\frac{d z_d}{d x_d}(x_{i,d})
+
B(\bm{x}_i)
\frac{d^2 z_d}{d x_d^2}(x_{i,d})
\right]
\frac{\partial v_{\bm{\theta}}}{\partial z_d}
\nonumber\\
&\quad
-
\sum_{d=1}^2
B(\bm{x}_i)
\left[
\frac{d z_d}{d x_d}(x_{i,d})
\right]^2
\frac{\partial^2 v_{\bm{\theta}}}{\partial z_d^2}
-
q(\bm{x}_i).
\end{align}
The quantities \(B\), \(q\), and \(z\) do not depend on \(\theta_\ell\). Therefore
\begin{align}
\frac{\partial r_{\bm{\theta},\mathrm{hard}}^{\mathrm{SP}}}{\partial\theta_\ell}(\bm{x}_i)
&=
\left[
\lambda B(\bm{x}_i)
-
\sum_{d=1}^2
\frac{\partial^2 B}{\partial x_d^2}(\bm{x}_i)
\right]
\frac{\partial v_{\bm{\theta}}}{\partial\theta_\ell}
\nonumber\\
&\quad
-
\sum_{d=1}^2
\left[
2
\frac{\partial B}{\partial x_d}(\bm{x}_i)
\frac{d z_d}{d x_d}(x_{i,d})
+
B(\bm{x}_i)
\frac{d^2 z_d}{d x_d^2}(x_{i,d})
\right]
\frac{\partial^2 v_{\bm{\theta}}}{\partial\theta_\ell\partial z_d}
\nonumber\\
&\quad
-
\sum_{d=1}^2
B(\bm{x}_i)
\left[
\frac{d z_d}{d x_d}(x_{i,d})
\right]^2
\frac{\partial^3 v_{\bm{\theta}}}{\partial\theta_\ell\partial z_d^2}.
\label{eq:app_poisson_hard_residual_gradient}
\end{align}

\subsection{Stationary viscous Hamilton-Jacobi equation}
\label{app:loss_gradient_expression_HJ}

For the stationary viscous Hamilton-Jacobi equation, the residual for a model prediction \(u_{\bm{\theta}}\) is
\begin{equation}
r_{\bm{\theta}}^{\mathrm{HJ}}(\bm{x})
=
-
\sum_{d=1}^2
\frac{\partial^2 u_{\bm{\theta}}}{\partial x_d^2}(\bm{x})
+
\rho
\left[
\sum_{d=1}^2
\left(
\frac{\partial u_{\bm{\theta}}}{\partial x_d}(\bm{x})
\right)^2
\right]^{\beta/2}
-
q(\bm{x})
-
\alpha_0.
\end{equation}
Here \(q\) is the source term defined in Eq.~\eqref{eq:hj_source}.
After substituting the soft and hard ansatzes, we write the resulting residuals as \(r_{\bm{\theta},\mathrm{soft}}^{\mathrm{HJ}}\) and \(r_{\bm{\theta},\mathrm{hard}}^{\mathrm{HJ}}\), respectively.
The soft and hard loss gradients are
\begin{align}
\frac{\partial \mathcal L_{\mathrm{soft}}^{\mathrm{HJ}}}{\partial\theta_\ell}
&=
\frac{2}{N_\Omega}
\sum_{i=1}^{N_\Omega}
r_{\bm{\theta},\mathrm{soft}}^{\mathrm{HJ}}(\bm{x}_i)
\frac{\partial r_{\bm{\theta},\mathrm{soft}}^{\mathrm{HJ}}}{\partial\theta_\ell}(\bm{x}_i)
+
\frac{2\gamma}{N_{\partial\Omega}}
\sum_{j=1}^{N_{\partial\Omega}}
u_{\bm{\theta}}(\bm{y}_j)
\frac{\partial u_{\bm{\theta}}}{\partial\theta_\ell}(\bm{y}_j),
\nonumber\\
\frac{\partial \mathcal L_{\mathrm{hard}}^{\mathrm{HJ}}}{\partial\theta_\ell}
&=
\frac{2}{N_\Omega}
\sum_{i=1}^{N_\Omega}
r_{\bm{\theta},\mathrm{hard}}^{\mathrm{HJ}}(\bm{x}_i)
\frac{\partial r_{\bm{\theta},\mathrm{hard}}^{\mathrm{HJ}}}{\partial\theta_\ell}(\bm{x}_i).
\end{align}
For the soft ansatz, the same chain-rule identities used above give
\begin{align}
r_{\bm{\theta},\mathrm{soft}}^{\mathrm{HJ}}(\bm{x}_i)
&=
-
\sum_{d=1}^2
\left[
\frac{d^2 z_d}{d x_d^2}(x_{i,d})
\frac{\partial v_{\bm{\theta}}}{\partial z_d}
+
\left[
\frac{d z_d}{d x_d}(x_{i,d})
\right]^2
\frac{\partial^2 v_{\bm{\theta}}}{\partial z_d^2}
\right]
\nonumber\\
&\quad
+
\rho
\left\{
\sum_{d=1}^2
\left[
\frac{d z_d}{d x_d}(x_{i,d})
\frac{\partial v_{\bm{\theta}}}{\partial z_d}
\right]^2
\right\}^{\beta/2}
-
q(\bm{x}_i)
-
\alpha_0.
\end{align}
Differentiating the residual with respect to \(\theta_\ell\) gives two contributions. The Laplacian term is linear in \(v_{\bm{\theta}}\). The nonlinear Hamiltonian term follows from differentiating the \(\beta/2\)-power of \(\sum_d(\partial u_{\bm{\theta}}/\partial x_d)^2\). Hence
\begin{align}
\frac{\partial r_{\bm{\theta},\mathrm{soft}}^{\mathrm{HJ}}}{\partial\theta_\ell}(\bm{x}_i)
&=
-
\sum_{d=1}^2
\left[
\frac{d^2 z_d}{d x_d^2}(x_{i,d})
\frac{\partial^2 v_{\bm{\theta}}}{\partial\theta_\ell\partial z_d}
+
\left[
\frac{d z_d}{d x_d}(x_{i,d})
\right]^2
\frac{\partial^3 v_{\bm{\theta}}}{\partial\theta_\ell\partial z_d^2}
\right]
\nonumber\\
&\quad
+
\rho\beta
\left\{
\sum_{d=1}^2
\left[
\frac{d z_d}{d x_d}(x_{i,d})
\frac{\partial v_{\bm{\theta}}}{\partial z_d}
\right]^2
\right\}^{\beta/2-1}
\nonumber\\
&\quad
\times
\sum_{d=1}^2
\left[
\frac{d z_d}{d x_d}(x_{i,d})
\right]^2
\frac{\partial v_{\bm{\theta}}}{\partial z_d}
\frac{\partial^2 v_{\bm{\theta}}}{\partial\theta_\ell\partial z_d}.
\label{eq:app_hj_soft_residual_gradient}
\end{align}
The boundary term in the soft Hamilton-Jacobi loss is again estimated from
\begin{equation}
u_{\bm{\theta}}(\bm{y}_j)
=
v_{\bm{\theta}},
\quad
\frac{\partial u_{\bm{\theta}}}{\partial\theta_\ell}(\bm{y}_j)
=
\frac{\partial v_{\bm{\theta}}}{\partial\theta_\ell}.
\end{equation}
For the hard ansatz, the chain rule gives the same spatial derivatives as in the hard screened Poisson case. In addition, differentiating those derivatives with respect to \(\theta_\ell\) gives
\begin{align}
\frac{\partial}{\partial\theta_\ell}
\frac{\partial u_{\bm{\theta}}}{\partial x_d}(\bm{x}_i)
&=
\frac{\partial B}{\partial x_d}(\bm{x}_i)
\frac{\partial v_{\bm{\theta}}}{\partial\theta_\ell}
+
B(\bm{x}_i)
\frac{d z_d}{d x_d}(x_{i,d})
\frac{\partial^2 v_{\bm{\theta}}}{\partial\theta_\ell\partial z_d},
\nonumber\\
\frac{\partial}{\partial\theta_\ell}
\frac{\partial^2 u_{\bm{\theta}}}{\partial x_d^2}(\bm{x}_i)
&=
\frac{\partial^2 B}{\partial x_d^2}(\bm{x}_i)
\frac{\partial v_{\bm{\theta}}}{\partial\theta_\ell}
+
\left[
2
\frac{\partial B}{\partial x_d}(\bm{x}_i)
\frac{d z_d}{d x_d}(x_{i,d})
+
B(\bm{x}_i)
\frac{d^2 z_d}{d x_d^2}(x_{i,d})
\right]
\frac{\partial^2 v_{\bm{\theta}}}{\partial\theta_\ell\partial z_d}
\nonumber\\
&\quad
+
B(\bm{x}_i)
\left[
\frac{d z_d}{d x_d}(x_{i,d})
\right]^2
\frac{\partial^3 v_{\bm{\theta}}}{\partial\theta_\ell\partial z_d^2}.
\end{align}
Substituting the hard ansatz into the Hamilton-Jacobi residual gives
\begin{align}
r_{\bm{\theta},\mathrm{hard}}^{\mathrm{HJ}}(\bm{x}_i)
&=
-
\sum_{d=1}^2
\biggl[
\frac{\partial^2 B}{\partial x_d^2}(\bm{x}_i)
v_{\bm{\theta}}
+
\left[
2
\frac{\partial B}{\partial x_d}(\bm{x}_i)
\frac{d z_d}{d x_d}(x_{i,d})
+
B(\bm{x}_i)
\frac{d^2 z_d}{d x_d^2}(x_{i,d})
\right]
\frac{\partial v_{\bm{\theta}}}{\partial z_d}
\nonumber\\
&\qquad
+
B(\bm{x}_i)
\left[
\frac{d z_d}{d x_d}(x_{i,d})
\right]^2
\frac{\partial^2 v_{\bm{\theta}}}{\partial z_d^2}
\biggr]
\nonumber\\
&\quad
+
\rho
\left\{
\sum_{d=1}^2
\left[
\frac{\partial B}{\partial x_d}(\bm{x}_i)
v_{\bm{\theta}}
+
B(\bm{x}_i)
\frac{d z_d}{d x_d}(x_{i,d})
\frac{\partial v_{\bm{\theta}}}{\partial z_d}
\right]^2
\right\}^{\beta/2}
-
q(\bm{x}_i)
-
\alpha_0.
\end{align}
Differentiating this expression with respect to \(\theta_\ell\) gives
\begin{align}
\frac{\partial r_{\bm{\theta},\mathrm{hard}}^{\mathrm{HJ}}}{\partial\theta_\ell}(\bm{x}_i)
&=
-
\sum_{d=1}^2
\biggl[
\frac{\partial^2 B}{\partial x_d^2}(\bm{x}_i)
\frac{\partial v_{\bm{\theta}}}{\partial\theta_\ell}
+
\left[
2
\frac{\partial B}{\partial x_d}(\bm{x}_i)
\frac{d z_d}{d x_d}(x_{i,d})
+
B(\bm{x}_i)
\frac{d^2 z_d}{d x_d^2}(x_{i,d})
\right]
\frac{\partial^2 v_{\bm{\theta}}}{\partial\theta_\ell\partial z_d}
\nonumber\\
&\qquad
+
B(\bm{x}_i)
\left[
\frac{d z_d}{d x_d}(x_{i,d})
\right]^2
\frac{\partial^3 v_{\bm{\theta}}}{\partial\theta_\ell\partial z_d^2}
\biggr]
\nonumber\\
&\quad
+
\rho\beta
\left\{
\sum_{d=1}^2
\left[
\frac{\partial B}{\partial x_d}(\bm{x}_i)
v_{\bm{\theta}}
+
B(\bm{x}_i)
\frac{d z_d}{d x_d}(x_{i,d})
\frac{\partial v_{\bm{\theta}}}{\partial z_d}
\right]^2
\right\}^{\beta/2-1}
\nonumber\\
&\quad
\times
\sum_{d=1}^2
\left[
\frac{\partial B}{\partial x_d}(\bm{x}_i)
v_{\bm{\theta}}
+
B(\bm{x}_i)
\frac{d z_d}{d x_d}(x_{i,d})
\frac{\partial v_{\bm{\theta}}}{\partial z_d}
\right]
\nonumber\\
&\quad
\times
\left[
\frac{\partial B}{\partial x_d}(\bm{x}_i)
\frac{\partial v_{\bm{\theta}}}{\partial\theta_\ell}
+
B(\bm{x}_i)
\frac{d z_d}{d x_d}(x_{i,d})
\frac{\partial^2 v_{\bm{\theta}}}{\partial\theta_\ell\partial z_d}
\right].
\label{eq:app_hj_hard_residual_gradient}
\end{align}

\clearpage
\twocolumngrid
\bibliography{references}

\end{document}